\documentclass[12pt]{article}
\usepackage[a4paper]{geometry}
\usepackage{amsmath}
\usepackage{amsthm}
\usepackage{placeins}
\usepackage{multirow}
\usepackage{subcaption}
\usepackage{natbib}
\usepackage{lscape} 

\usepackage{algorithm}
\usepackage{algpseudocode}
\usepackage[colorlinks,citecolor=blue,urlcolor=blue]{hyperref}
\usepackage{morefloats}
\usepackage{xcolor}
\usepackage{tikz}
\usepackage{bbm}
\usetikzlibrary{arrows,automata,er,positioning,calc,arrows}
\usetikzlibrary{shapes,decorations,arrows,calc,arrows.meta,fit,positioning}
\tikzset{
	-Latex,auto,node distance =1 cm and 1 cm,semithick,
	state/.style ={ellipse, draw, minimum width = 0.7 cm},
	point/.style = {circle, draw, inner sep=0.04cm,fill,node contents={}},
	bidirected/.style={Latex-Latex,dashed},
	el/.style = {inner sep=2pt, align=left, sloped}
}
\usepackage{upgreek}
\usepackage{cancel}

\usepackage{amssymb}
\usepackage{mathtools}
\usepackage{xr}
\usepackage{authblk} 
\usepackage{enumerate}
\usepackage{afterpage}
\usepackage[font=large]{caption}
\usepackage{bm}

\newcommand{\cI}{\mathcal{I}}
\newcommand{\bbeta}{\boldsymbol{\beta}}
\newcommand{\btheta}{\boldsymbol{\theta}}
\newcommand{\bx}{\mathbf{x}}
\newcommand{\bX}{\mathbf{X}}
\newcommand{\by}{\mathbf{y}}
\newcommand{\E}{\mathbb{E}}

\newtheorem{theorem}{Theorem}
\newtheorem{lemma}{Lemma}
\newtheorem{assumption}{A.}
\newtheorem{proposition}{Proposition}

\newtheorem{remark}{Remark}

\begin{document}
	
	\def\spacingset#1{\renewcommand{\baselinestretch}%
		{#1}\small\normalsize} \spacingset{1.5}

	\title{\bf GLM Inference with AI-Generated Synthetic Data Using Misspecified Linear Regression}
	
	\author[1]{Nir Keret}
	\author[1,2]{Ali Shojaie}
	
	\affil[1]{Department of Biostatistics, University of Washington}
	\affil[2]{Department of Statistics, University of Washington}
	
	\date{}
	\maketitle

	\begin{abstract}
		Data privacy concerns have led to the growing interest in synthetic data, which strives to preserve the statistical properties of the original dataset while ensuring privacy by excluding real records. 
		Recent advances in deep neural networks and generative artificial intelligence have facilitated the generation of synthetic data. 
		However, although prediction with synthetic data has been the focus of recent research, statistical inference with synthetic data remains underdeveloped. 
		In particular, in many settings, including generalized linear models (GLMs), the estimator obtained using synthetic data converges much more slowly than in standard settings. To address these limitations, we propose a method that leverages summary statistics from the original data. Using a misspecified linear regression estimator, we then develop inference that greatly improves the convergence rate and restores the standard root-$n$ behavior for GLMs.
	\end{abstract}
	
	{\bf keywords:} Data alignment; Data compression; Electronic Health Records (EHR); Generalized linear models (GLMs); Privacy; Summary statistics; Synthetic data; Transportability

	\section{Introduction}\label{sec:intro}

	Data privacy regulations, institutional barriers, and regulatory overhead hinder the sharing of patient-level healthcare data across institutions. This bottleneck continues to slow progress in statistical modeling critical for scientific discovery. As a result, many machine learning and artificial intelligence (ML/AI) applications and statistical analyses are conducted in isolation, making them difficult to reproduce, extend, or generalize across populations. Open-access de-identified datasets like MIMIC \citep{johnson2016mimic} and AmsterdamUMCdb \citep{thoral2021sharing} are limited in number and scope, highlighting the shortcomings of traditional de-identification methods \citep{rocher2019estimating}, which often fail to meet privacy requirements for broader dissemination.
	
	Privacy-preserving alternatives include federated learning \citep{zhang2021survey} and differential privacy (DP) \citep{dwork2006differential}. Federated learning enables institutions to jointly train models by sharing updates or summary statistics rather than raw data. However, it requires ongoing coordination, incurs significant overhead, and depends on trust among participants. Most importantly, it does not address the underlying problem of the absence of publicly accessible datasets.
	
	DP is a theoretical framework that ensures the inclusion or exclusion of any individual has minimal effect on the distribution of released outputs. It does so by adding controlled noise, obscuring whether a specific person was part of the dataset. While widely adopted in the computer science community, DP has major limitations, most notably degrading data utility and model accuracy \citep{bambauer2013fool,fredrikson2014privacy,bagdasaryan2019differential,domingo2021limits}. As \cite{yoon2020anonymization} note, DP provides probabilistic guarantees over algorithm outputs, not direct protection of specific datasets. This distinction is especially important in healthcare, where re-identification risk in the original data is the primary concern. Moreover, DP has yet to see adoption in healthcare industry standards or regulation.
	
	Synthetic data offer a promising solution by providing realistic yet non-identifiable records that enable secure data sharing and analysis without exposing sensitive information. The Clinical Practice Research Datalink \citep{herrett2015data} release of a synthetic cardiovascular dataset demonstrates this potential, showing that synthetic data can support research while preserving confidentiality. Traditional methods for generating synthetic data typically use parametric models, often within a multiple-imputation framework \citep{liew1985data,rubin1993statistical,raghunathan2003multiple,drechsler2011synthetic,awan2020one,jiang2021balancing,mathur2024fully,raisa2025consistent}. These methods depend on strong assumptions about the true data-generating process and may struggle to capture complex relationships due to their limited expressiveness.
	
	In the ML community, synthetic datasets first gained traction in computer vision, where artificially generated images were used, either alone or as augmentations, to improve training for tasks like classification and segmentation. Instead of depending entirely on large collections of real images, synthetic images rendered from 3D models can provide valuable supplements. See \cite{anderson2022synthetic} and references therein for further discussion.

	More recently, there has been growing interest in flexible models for generating synthetic tabular data. These include Bayesian networks \citep{zhang2017privbayes,Ankan2024}, tree-based methods \citep{reiter2005using,drechsler2010sampling,synthpop,watson2023adversarial}, generative adversarial networks (GANs) \citep{NEURIPS2019_254ed7d2,yoon2020anonymization,van2021decaf}, diffusion models \citep{kotelnikov2023tabddpm}, and large language models \citep{borisov2022language}. These approaches, along with other AI-driven methods, offer expressive frameworks for generating synthetic data that may lead to solid predictive performance while protecting privacy. However, they may not support valid statistical inference such as point estimation and confidence interval construction, because flexible generative models often converge to the true distribution at rates substantially slower than the optimal $\sqrt{n}$ \citep{decruyenaerereal}. Developing rigorous and efficient inferential methods for such synthetic data remains an open challenge. Indeed, the American Statistical Association’s Committee on Privacy and Confidentiality has recently highlighted critical gaps in statistical approaches to data privacy, urging the community to develop rigorous methodologies and tools, including those based on synthetic data, to meet these pressing challenges \citep{ASA2025privacy}.
	
	\cite{ghalebikesabi2022mitigating} proposed an importance-weighting method to correct biases in synthetic datasets, including those generated with DP, but did not address statistical inference. \cite{decruyenaere2024debiasing} introduced a debiasing approach for data generated by deep-learning models, which can substantially improve estimation efficiency. However, this method requires access to both the generative model and the original sample used to train it. As a result, the data center must anticipate researchers' inferential goals and collaborate with them directly.
	
	Our focus differs. We consider the setting where a data center holds the data and trains a general-purpose, one-size-fits-all model without tailoring it to any specific research interests. Given the broad range of potential research questions, employing flexible models to generate synthetic data is a natural choice. Researchers can request and receive a synthetic copy of the data, or even directly download it, but would have no direct collaboration with the data center.  
	
	To mitigate the limitations of AI-generated synthetic data for statistical inference, we assume the data center supplements the synthetic dataset with summary statistics from the original data. Specifically, we assume it provides the Gram matrix $\mathcal{X}^\top \mathcal{X}$, where $\mathcal{X}$ is the matrix of all variables of interest, including an intercept and the outcome. Although this matrix alone is insufficient to fit a generalized linear model (GLM), it fully determines the ordinary least squares (OLS) solution, a property we will exploit to enable efficient GLM inference.
	
	Sharing summary statistics imposes relatively low communication and privacy burdens. The Gram matrix $\mathcal{X}^\top \mathcal{X}$ is a coarse aggregate, and when $n \gg p$, it may offer sufficient privacy protection, see Section~\ref{sec:privacy} for further discussion. Sharing such summaries, including Hessians, is standard in meta-analysis and distributed learning \citep{lin2010relative,liu2015multivariate,lemyre2024distributed,jonker2025bayesian}. However, when $n$ is comparable to or smaller than $p$, the matrix may leak sensitive information. Privacy and communication constraints in such settings may preclude direct sharing of the Gram matrix. This is especially relevant in high-dimensional genomic applications, where the matrix is often too large to share and is instead approximated using public resources \citep{zhou2021fast,listatistical}.
	
	\cite{wilde2021foundations} propose a Bayesian framework for synthetic data, acknowledging that synthetic data alone may be insufficient. Their approach assumes the analyst has a small sample from the true distribution to inform the prior. They also recommend releasing results from a default model fitted to the real data, without tuning to any specific inferential task.
	
	After introducing our new methodology, in Section~\ref{sec:method} we show that, under a local ``transportability" assumption, leveraging the summary statistics can restore the convergence rate of synthetic-data-based GLM estimators to the usual $\sqrt{n}$-rate. This result is supported by theoretical analysis and validated through extensive simulations in Section~\ref{sec:sims} and real-data analysis in Section~\ref{sec:dataAnalysis}. We discuss limitations and potential directions for future research in Section~\ref{sec:discussion}.

	\section{Methodology}\label{sec:method}
	
	\subsection{Estimation of GLM Parameters from Synthetic Data}\label{sec:est}
	We begin by introducing notation and the problem setup. Matrices are denoted by bold uppercase letters, vectors by bold lowercase, and random variables by non-bold uppercase. Let the original sample consist of $n$ independent and identically distributed (i.i.d.) observations from a distribution $F$, with empirical distribution $F_n$. The data include a covariate matrix $\bX$ of size $n \times p$ and an outcome vector $\by$ of length $n$. Each row $\bx_i^\top$ represents the covariates for observation $i$, and $\bX$ includes a column of ones for the intercept.
	
	To generate synthetic data, a flexible model is trained on the observed data $(\bX, \by)$, yielding a fitted distribution $\hat{F}_n$. A synthetic sample of size $m$ is then drawn from this model, with empirical distribution $\tilde{F}_m$. The synthetic covariates and outcomes are denoted by $\tilde{\bX}_n$ and $\tilde{\by}_n$, where the subscript $n$ indicates that they are generated from a model trained on $n$ i.i.d. observations, thus forming a triangular array.
	
	Let $\mathcal{D}_n$ denote the original sample together with $\hat{F}_n$, which may involve randomness in the fitting procedure. Conditionally on $\mathcal{D}_n$, we assume the synthetic observations are i.i.d. We consider a coefficient vector $\bbeta$ and a function $\mu$, such that the conditional mean satisfies $\mathbb{E}(Y|\bx) = \mu(\bx^\top \bbeta)$. Our goal is to infer the true population parameter $\bbeta^*$, and we denote by $\hat{\bbeta}$ the solution to the score equation
	\begin{equation}
		\int\left\{y - \mu\left(\bx^\top\bbeta\right)\right\}\bx \, dF_n(\bx,y) = \frac{1}{n}\sum_{i=1}^n \left\{ Y_i - \mu\left(\bx_i^\top {\bbeta}\right) \right\} \bx_i \equiv \mathbf{0} \, . \label{eq:ScoreMLEOrig}
	\end{equation}
	For simplicity, we assume hereafter that the model is a correctly-specified GLM with canonical link, so that $\hat{\bbeta}$ is the maximum likelihood estimator (MLE). Extensions to non-canonical links and to models specified only through the conditional mean are discussed in Appendix~\ref{appen:extensions}, along with a corresponding simulation study.
	
	Let $\btheta$ be the coefficient vector in a (misspecified) linear model $Y = \bx^\top \btheta + \varepsilon$, where $\E(\varepsilon | \bx) = 0$, and let $\hat{\btheta}$ denote the OLS estimator from the original sample,
	\begin{equation} \label{eq:thetaOLS}
		\hat{\btheta} = \left(\bX^\top \bX\right)^{-1} \bX^\top \by \, ,
	\end{equation}
	which solves the score equation
	\begin{equation}
		\int \left( y - \bx^\top \btheta \right) \bx \, dF_n(\bx, y) = \frac{1}{n} \sum_{i=1}^n \left( Y_i - \bx_i^\top \btheta \right) \bx_i \equiv \mathbf{0} \, . \label{eq:ScoreOLSOrig}
	\end{equation}
	Define $\btheta^*$ as the population limit of $\hat{\btheta}$, that is, the coefficient vector corresponding to the best linear approximation of $Y$ given $\bx$ in the population.
	Together with the synthetic data, we assume access to $\mathcal{X}^T\mathcal{X}$ from the original sample, which includes $\bX^\top \bX$ and $\bX^\top \by$. These statistics are sufficient to compute $\hat{\btheta}$ exactly and will play a pivotal role in our estimation procedure.  
	
	Since the original sample is unavailable, a natural approach is to replace $F_n$ with $\tilde{F}_m$ in score equation \eqref{eq:ScoreMLEOrig} and solve for $\bbeta$. However, although $\hat{F}_n$ (and thus $\tilde{F}_m$) is assumed to converge to $F_n$ under suitable metrics, the convergence may be slow \citep{decruyenaerereal}, leading to substantial bias. We therefore consider the following bias-corrected estimating equation:
	\begin{equation*}
		\int \left\{ y - \mu\big(\bx^\top \bbeta\big) \right\} \bx \, d\tilde{F}_m(\bx, y) - \int \left\{ y - \mu\big(\bx^\top \hat{\bbeta}\big) \right\} \bx \, d\tilde{F}_m(\bx, y) = \int \left\{ \mu\big(\bx^\top \hat{\bbeta}\big) - \mu\big(\bx^\top \bbeta\big) \right\} \bx \, d\tilde{F}_m(\bx),
	\end{equation*}  
	where it is evident that $\hat{\bbeta}$ is the root of this equation. However, since this equation depends on $\hat{\bbeta}$ itself, it is not directly useful for estimation. Nonetheless, from Eq.'s (\ref{eq:ScoreMLEOrig})--(\ref{eq:ScoreOLSOrig}) we have the relation  
	$$
	\int \bx^\top \hat{\btheta} \bx \, dF_n(\bx) = \int \mu\big(\bx^\top \hat{\bbeta}\big) \bx \, dF_n(\bx) \, ,
	$$  
	which leads to the following estimating equation for $\bbeta$:
	\begin{equation}\label{eq:ScoreCorrected}
		\boldsymbol{\psi}_m(\bbeta) = \int \left\{ \bx^\top \hat{\btheta} - \mu\big(\bx^\top \bbeta\big) \right\} \bx \, d\tilde{F}_m(\bx) = \frac{1}{m} \sum_{i=1}^m \left\{ \tilde{\bx}_{n,i}^\top \hat{\btheta} - \mu\big(\tilde{\bx}_{n,i}^\top \bbeta\big) \right\}\tilde{\bx}_{n,i} \equiv \mathbf{0} \, . 
	\end{equation} 
	
	Alternatively, Eq.~(\ref{eq:ScoreCorrected}) can be derived directly from
	\begin{equation} \label{eq:dataFusion}
		\int\!\!\int  \left( y - \bx^\top \hat{\btheta} \right)\bx f(\by \mid \bx; \bbeta)\, d\by\, d\tilde{F}_m(\bx),
	\end{equation}
	where $f$ is the conditional density. This expression matches the constraint used in data-fusion methods combining individual-level data from an ``internal'' source with summary information from an ``external'' source, in this case the synthetic and original datasets, respectively \citep{chatterjee2016constrained,han2019empirical,zhai2022data,fang2025integrated}.

	Eq.~\eqref{eq:ScoreCorrected} combines the synthetic covariates with $\hat{\btheta}$ from the original data (Eq.~\eqref{eq:thetaOLS}), and its root is denoted by $\tilde{\bbeta}$. The key difficulty for the generative model is to recover the conditional distribution of $Y$ given $\bx$. Unlike parametric bootstrap, which assumes a fixed model and estimates only $\bbeta$, a flexible generator must infer both $\bbeta$ and the form of $\E(Y|\bx)$, a much harder task. Eq.~\eqref{eq:ScoreCorrected} sidesteps this difficulty by excluding synthetic outcomes, easing the burden on the generator.

	To leverage information from the original data and enhance the alignment between the synthetic and original datasets, we also apply a ``whiten-recolor" transformation \citep{chiu2019understanding}, assuming $\bX$ and $\tilde{\bX}$ are full rank. Let $\mathbf{R}_{\bx}$ be the Cholesky decomposition upper-triangular matrix satisfying $\mathbf{R}_{\bx}^\top \mathbf{R}_{\bx} = mn^{-1}\bX^\top \bX$. Similarly, let $\mathbf{R}_{\tilde{\bx}}$ be the upper-triangular matrix satisfying $\mathbf{R}_{\tilde{\bx}}^\top \mathbf{R}_{\tilde{\bx}} = \tilde{\bX}_n^\top \tilde{\bX}_n$. We apply the transformation 
	$$
	\tilde{\bX}_{n} \mapsto \tilde{\bX}_{n} \mathbf{R}_{\tilde{\bx}}^{-1} \mathbf{R}_{\bx} \, ,
	$$
	which ensures that $m^{-1}\tilde{\bX}_n^\top \tilde{\bX}_n = n^{-1}\bX^\top \bX$, effectively aligning the second moments of the synthetic dataset with those of the original data. We treat this post-processing step as an inherent part of the data generation process and assume that $\tilde{\bX}_n$ has already undergone this transformation, with the distributions $\hat{F}_n$ and $\tilde{F}_m$ reflecting this adjustment. Appendix~\ref{appen:extensions} presents simulations illustrating the benefits of the whitening-recoloring transformation, showing that it improves estimation efficiency and yields more accurate variance estimates.
	
	Note that while $\mathbf{R}_{\mathbf{x}}$ is deterministic given $\mathcal{D}_n$, the inverse $\mathbf{R}_{\tilde{\mathbf{x}}}^{-1}$ remains random since it is computed from the synthetic data. As a result, the transformed synthetic observations are no longer independent conditional on $\mathcal{D}_n$. However, as $m$ grows, the estimation error in the Cholesky factor, and thus in its inverse, shrinks at rate $O(1/m)$, making the induced dependence asymptotically negligible. A formal proof is given in Proposition~\ref{prop:choleskyindep} in Appendix~\ref{appen:proofs}. We therefore treat the $m$ synthetic observations as conditionally i.i.d.\ in what follows.
	
	\begin{remark}
		A simple alternative to Eq.~\eqref{eq:ScoreCorrected} that also avoids using $\tilde{\by}_n$ is to write
		$$
		\int \mu(\bx^\top \bbeta)\bx\, d\tilde{F}_m(\bx) = \int y\bx\, d\tilde{F}_m(\bx,y),
		$$
		and replace $\tilde{F}_m$ with $F_n$ on the right-hand side, yielding
		$$
		\frac{1}{m}\sum_{i=1}^m \mu(\tilde{\bx}_{n,i}^\top \bbeta)\tilde{\bx}_{n,i} = \frac{1}{n}\bX^\top \by.
		$$
		This approach breaks the correlation between $Y$ and $\mu(\tilde{\bx}_n^\top \bbeta)$ that exists in the original score, where it helps reduce $\text{Var}\left\{Y - \mu(\bx^\top \bbeta)\right\}$. Removing this correlation increases the variance of the estimator, making it less efficient than Eq.~\eqref{eq:ScoreCorrected}. Its advantage is that it requires only $\bX^\top \by$, improving communication efficiency and privacy.
	\end{remark}

	\subsection{Asymptotic Theory}\label{sec:inf}
	
	We first introduce some additional notation. Let $\mathbf{W}(\bbeta)$ and $\tilde{\mathbf{W}}_n(\bbeta)$ be diagonal matrices of dimensions $n \times n$ and $m \times m$, respectively, with diagonal entries $\left\{ \mu'\big(\bx_i^\top\bbeta\big) \right\}_{i=1}^n$ and $\{\mu'\big(\tilde{\bx}_{n,i}^\top \bbeta\big) \}_{i=1}^m$, where $\mu'$ is the derivative of $\mu$. The observed information matrix evaluated at $\bbeta$, for the original sample is $\boldsymbol{\cI}(\bbeta) = n^{-1} \bX^\top \mathbf{W}(\bbeta) \bX$, and its synthetic counterpart is $\tilde{\boldsymbol{\cI}}_n(\bbeta) = m^{-1} \tilde{\bX}_n^\top \tilde{\mathbf{W}}_n(\bbeta) \tilde{\bX}_n$. Denote $\boldsymbol{\phi}(\bx,\bbeta)=\left\{\bx^\top \hat{\btheta} - \mu\left(\bx^\top \bbeta\right) \right\}\bx$, and  let $\boldsymbol{\mathcal{A}}^{(n)}_1 = \text{Var} (\boldsymbol{\phi}(\tilde{\bx}_n,\hat{\bbeta}) | \mathcal{D}_n)$ and \\ $\boldsymbol{\mathcal{A}}^{(n)}_2 = n\text{Var}\left[\E_{\hat{F}_n}\left\{\boldsymbol{\phi} (\bx,\hat{\bbeta})\right\} \,\Big|\, \hat{\bbeta}, \hat{\btheta}\right]$, where $\E_F\left\{g(\bx)\right\} = \int g(\bx)dF(\bx)$. Below, we list the assumptions required for establishing the asymptotic properties of our estimator, followed by a short discussion on these assumptions.
	\begin{assumption} \label{as:reularityParam}
		\textbf{Regularity conditions on the parameter space:} \\
		Assume $\bbeta\in\mathcal{B}$ and $\btheta\in\Theta$, where $\mathcal{B}\subset\mathbb{R}^p$ and $\Theta\subset\mathbb{R}^p$ are compact sets. Consequently, there exists a finite constant $M<\infty$ such that $\|\bbeta\|\le M$ and $\|\btheta\|\le M$ for all $\bbeta\in\mathcal{B}$ and $\btheta\in\Theta$. Moreover, the parameters $\bbeta^*$ and $\btheta^*$ are interior points of $\mathcal{B}$ and $\Theta$, respectively.
	\end{assumption}
	\begin{assumption} \label{as:regularityx}
		\textbf{Regularity conditions on $\mu$, $\bx,  \tilde{\bx}_n$:}
		
		\begin{enumerate}
			\item[\textnormal{(a)}]  
			There exist constants \(C_1, \delta > 0\) and an index \(n^*\) such that for all \(n > n^*\) and all \(\bbeta \in \mathcal{B}\),
			$$
			\E_{\hat{F}_n} \left\{ \|\boldsymbol{\phi}(\bx, \bbeta)\|^{2+\delta} \right\} < C_1 \quad \text{almost surely}.
			$$ 
			
			\item[\textnormal{(b)}]  
			There exists an index $n^*$ such that for all $n>n^*$, the function $\mu$ is Lipschitz, satisfying
			$$
			|\mu(\tilde{\bx}_n^\top \bbeta_1) - \mu(\tilde{\bx}_n^\top \bbeta_2)| \leq L(\tilde{\bx}_n) \|\bbeta_1 - \bbeta_2\|,
			$$
			for all $\bbeta_1, \bbeta_2 \in \mathcal{B}$, where $L$ is a nonnegative function. Moreover, there exists a constant $C_2<\infty$ such that for all \(n > n^*\),
			$$
			\E_{\hat{F}_n} \left\{ L^2(\bx)\|\bx\|^2 \right\} < C_2 \quad \text{almost surely}.
			$$
			
			\item[\textnormal{(c)}]  
			For any $\bbeta_1,\ldots,\bbeta_m \in \mathcal{B}$, we have
			$$ 
			\frac{1}{m}\sum_{i=1}^m\left| \mu''(\tilde{\bx}_{n,i}^\top \bbeta_i) \right| \, \|\tilde{\bx}_{n,i}\|^3  = O_{P|\mathcal{D}_n}(1) \, ,
			$$
			where $\mu''$ is the second derivative of $\mu$.
			
			\item[\textnormal{(d)}]  
			There exists an index \(n^*\) such that for all \(n > n^*\), and all \(\bbeta \in \mathcal{B}\), the matrix
			$$
			\E_{\hat{F}_n} \left\{\mu'(\bx^\top \bbeta)\, \bx \bx^\top \right\}
			$$
			is definite with probability one; that is, it is either positive definite or negative definite.
			
			\item[\textnormal{(e)}] The matrices $\bX$, $\tilde{\bX}_n$ are full rank.
		\end{enumerate}
	\end{assumption}
	
	\begin{assumption} \label{as:varMat}
		\textbf{Convergence of variance matrices:} \\
		There exist positive definite matrices $\boldsymbol{\mathcal{A}}_1$ and $\boldsymbol{\mathcal{A}}_2$ such that as $n \to \infty$,
		$$\boldsymbol{\mathcal{A}}^{(n)}_1 \xrightarrow{P} \boldsymbol{\mathcal{A}}_1  \, , \quad \boldsymbol{\mathcal{A}}^{(n)}_2 \xrightarrow{P} \boldsymbol{\mathcal{A}}_2.$$
	\end{assumption}
	\begin{assumption} \label{as:biasNorm}
		\textbf{Asymptotic normality of the estimating equation's bias (transportability):}
		$$
		\sqrt{n}\E_{\hat{F}_n}\left\{\boldsymbol{\phi}(\bx,\hat{\bbeta})\right\} \,\Big|\, \hat{\bbeta}, \hat{\btheta} \xrightarrow{D} \mathcal{N}(\mathbf{0},\boldsymbol{\mathcal{A}}_2) \, .
		$$
	\end{assumption}
	\begin{assumption} \label{as:infoMat}
		\textbf{Convergence of the observed information matrices:}
		
		There exists a positive definite matrix $\mathbf{\Sigma}$ such that as $(m,n) \to \infty$,
		$$
		\tilde{\boldsymbol{\cI}}_n\big(\hat{\bbeta}\big) \xrightarrow{P|\hat{\bbeta},\hat{\btheta}} \mathbf{\Sigma} \, , \quad \boldsymbol{\cI}\big(\hat{\bbeta}\big) \xrightarrow{P} \mathbf{\Sigma}\,. 
		$$
	\end{assumption}
	
	Assumption~\ref{as:reularityParam} imposes standard regularity conditions on the parameter spaces, ensuring uniform convergence and the existence of derivatives.
	Assumption~\ref{as:regularityx} specifies conditions on $\mu$ and the distributions of $\tilde{\bx}_n$ and $\bx$ required for consistency and asymptotic normality. Although stated for $\tilde{\bx}_n$, these conditions implicitly apply to $\bx$, since $\tilde{\bx}_n$ inherits its distribution from $\bx$. 
	For common choices of $\mu$, including Poisson and binary regression, Appendix~\ref{appen:mu} provides sufficient conditions on $\tilde{\bx}_n$ and $\bx$ such that Assumption~\ref{as:regularityx} is satisfied. Assumption~\ref{as:varMat} requires convergence of certain variance matrix estimators to constant, positive definite limits.

	Assumption~\ref{as:biasNorm} provides the key `bridge' between the empirical distribution of the original data, $F_n$, and the synthetic distribution, $\hat{F}_n$. While $\hat{F}_n$ may converge slowly to $F_n$, often not at the $\sqrt{n}$ rate, our requirement is weaker and focuses on a local property: we require that a specific substructure of the synthetic distribution converges to its counterpart in $F_n$ at the $\sqrt{n}$ rate, where the target is zero. Specifically, we define the estimating equation bias as
	$$
	\E_{\hat{F}_n}\left\{\boldsymbol{\phi}(\bx,\hat{\bbeta})\right\} - \E_{{F}_n}\left\{\boldsymbol{\phi}(\bx,\hat{\bbeta})\right\},
	$$
	where the second term is identically zero. Geometrically, in the original data, the column space of $\bX$ is orthogonal to the difference between the GLM-based conditional mean and its best linear approximation. We require that this orthogonality is increasingly approximated in the synthetic estimating equation as $n$ grows, when evaluated at $\hat{\bbeta}$. This condition is related to the `transportability', `exchangeability' or `homogeneity' assumptions in the data fusion literature \citep{dahabreh2019generalizing,hu2022semiparametric,li2023improving,fang2025integrated}, which typically posit equality of certain functionals across distributions. In contrast, we require only asymptotic agreement at an appropriate rate. Section~\ref{sec:sims} provides empirical support for its plausibility.
	
	Assumption~\ref{as:infoMat} provides another `bridge', ensuring that the observed information matrices of the synthetic and original data, both evaluated at $\hat{\bbeta}$, converge to the same fixed positive definite matrix. This assumption is primarily needed for the sake of variance estimation.
	
	Our main theorem, stated below, establishes the asymptotic normality of $\tilde{\bbeta}$ from Eq.~\eqref{eq:ScoreCorrected} and confirms that it achieves the $\sqrt{n}$ convergence rate. We assume that as $n$ increases, $m$ grows proportionally, such that in the limit, $m/n \to \alpha$ for some positive constant $\alpha$.

	\begin{theorem}\label{thm:normality}
		Under Assumptions~\ref{as:reularityParam}--\ref{as:infoMat}, as $(m,n) \rightarrow \infty$, it holds that
		$$\sqrt{n}(\tilde{\bbeta} - \bbeta^*)\xrightarrow{D}\mathcal{N}(\mathbf{0},\mathbf{\mathcal{V}}) \, , $$
		where $\mathbf{\mathcal{V}} = \mathbf{\Sigma}^{-1}\left(\alpha^{-1}\boldsymbol{\mathcal{A}}_1 + \boldsymbol{\mathcal{A}}_2\right)\mathbf{\Sigma}^{-1} + \mathbf{\Sigma}^{-1}$. 	
	\end{theorem}
	
	The proof of Theorem~\ref{thm:normality} is given in Appendix~\ref{appen:proofs}. The second term in $\mathcal{V}$ represents the variance of $\hat{\bbeta}$ under a correctly-specified canonical GLM. 
	The first term accounts for the additional variability introduced by synthetic data generation and estimation. 
	The matrix $\boldsymbol{\mathcal{A}}_1$ captures variance due to the finite synthetic sample size $m$, and $\boldsymbol{\mathcal{A}}_2$ reflects the variance of the bias term arising from using $\hat{F}_n$ in place of $F_n$, as described in Assumption~\ref{as:biasNorm}.
	In Appendix~\ref{appen:extensions}, we investigate deviations from correctly-specified canonical GLMs and show that the proposed inference procedure remains valid, with the variance taking a `sandwich' form. 
	
	\subsection{Variance Estimation}\label{sec:variance}
	
	Examining the variance expression in Theorem~\ref{thm:normality} and based on Assumption~\ref{as:infoMat}, we use $\tilde{\boldsymbol{\cI}}_n^{-1}$($\tilde{\bbeta}$) as an estimator for $\mathbf{\Sigma}^{-1}$. Regarding $\boldsymbol{\mathcal{A}}_1$ and $\boldsymbol{\mathcal{A}}_2$, the former can be reliably estimated from the synthetic data, whereas the latter cannot, as it heavily depends on the variability across different realizations of the original, unobserved samples.  
	
	Given the challenges of estimating $\boldsymbol{\mathcal{A}}_2$, we focus on $\alpha^{-1} \boldsymbol{\mathcal{A}}_1 + \boldsymbol{\mathcal{A}}_2$, which corresponds, asymptotically, to  $\text{Var}\left\{\sqrt{n} \boldsymbol{\psi}_m(\hat{\bbeta})\right\}$. We will estimate this quantity directly via bootstrapping, generating approximate realizations of $\boldsymbol{\psi}_m(\hat{\bbeta})$. The variability within the synthetic data arises from the fact that we have only $m$ samples from $\hat{F}_n$, while the variability coming from the original sample stems from three sources: (i) the estimator $\hat{\bbeta}$, (ii) the estimator $\hat{\btheta}$, and (iii) the randomness in $\hat{F}_n$, which would change with a different sample. 
	To account for the uncertainty in $\hat{\bbeta}$ and $\hat{\btheta}$, we generate approximate realizations from their asymptotic joint distribution, given by  
\begin{equation}  \label{CovBetaTheta}
	\sqrt{n}  
	\begin{pmatrix}  
		\hat{\bbeta} - \bbeta^* \\  
		\hat{\btheta} - \btheta^*  
	\end{pmatrix}  
	\sim \mathcal{N} \left\{ 
	\mathbf{0},  
	\begin{pmatrix}  
		\boldsymbol{\Sigma}^{-1} & \E(\bx\bx^\top)^{-1} \\  
		\E(\bx\bx^\top)^{-1}  & \E(\bx\bx^\top)^{-1}  \boldsymbol{\Sigma} \E(\bx\bx^\top)^{-1}  
	\end{pmatrix}  
	\right\}.
\end{equation}
	Since $\boldsymbol{\Sigma}$ is unknown in practice, we will sample realizations from
	\begin{equation} \label{SamplingBetaTheta}
		\mathcal{N} \left\{  
		\begin{pmatrix}  
			\tilde{\bbeta} \\  
			\hat{\btheta}  
		\end{pmatrix},  
		\begin{pmatrix}  
			n^{-1}\tilde{\boldsymbol{\cI}}_n(\tilde{\bbeta})^{-1} & (\bX^\top \bX)^{-1} \\  
			(\bX^\top \bX)^{-1} & n(\bX^\top \bX)^{-1} \tilde{\boldsymbol{\cI}}_n(\tilde{\bbeta}) (\bX^\top \bX)^{-1}  
		\end{pmatrix}  
		\right\}.
	\end{equation}

	Next, we aim to capture the variability in $\tilde{\bX}_n$ that originates from the original sample. While we do not have access to the original sample or the generative model, we do have access to $\bX^\top\bX$. Since $\bX$ includes a column of ones, we can retrieve the empirical covariance matrix as
	$$\widehat{\text{Cov}}(\bx)= 
	\left( \frac{1}{n} \bX^\top \bX \right)_{2:p,2:p} - \bar{\bx}\bar{\bx}^\top  
	\quad , \quad \bar{\bx} =  \left( \frac{1}{n} \bX^\top \bX \right)_{2:p,1} \, .$$  
	When $\bx$ follows a normal distribution, the empirical covariance matrix follows a Wishart distribution. However, even if $\bx$ is not exactly Gaussian, in the low-dimensional setting ($n \gg p$), its covariance matrix can often be reasonably approximated by a Wishart distribution \citep{kollo1995approximating}. Building on this, we sample from a Wishart distribution and realign the synthetic data according to that matrix, using the ``whiten-recolor" procedure, as described above. This process injects the sampled variability in the original sample into a new, second-generation synthetic dataset. The full procedure is outlined in Algorithm~\ref{alg:boot}. As demonstrated in our simulation study, this variance estimation method performs satisfactorily. Appendix~\ref{appen:extensions} describes modifications to the variance-estimation procedure when the model is misspecified, uses a non-canonical link, or defined just through the conditional expectation.
	
	When the model of interest is linear regression, then not only can we obtain $\hat{\btheta}$ exactly based on the summary statistics, but the variance estimation, required for statistical inference, is also derived from those summary statistics, as
	$$\widehat{\text{Var}}(\hat{\btheta})=\hat{\sigma}^2(\bX^\top\bX)^{-1} \, $$
	where 
	$$\hat{\sigma}^2 = \frac{1}{(n-p)}\left(\by - \bX\hat{\btheta}\right)^\top\left(\by - \bX\hat{\btheta}\right) = \frac{1}{(n-p)}\left(\by^\top\by -2\hat{\btheta}^\top\bX^\top\by + \hat{\btheta}^\top\bX^\top\bX\hat{\btheta} \right) \, . $$

	\begin{algorithm}[t] 
		\spacingset{1}
		\caption{Bootstrap for Estimating  $\text{Var}\left\{ \boldsymbol{\psi}_m(\hat{\bbeta})\right\}$ from Synthetic Data \label{alg:boot}}
		\begin{algorithmic}
			\For{$b = 1,\ldots,B$}
			\begin{enumerate}[i.]	
				\item Sample $\hat{\bbeta}^{(b)}$ and $\hat{\btheta}^{(b)}$ according to (\ref{SamplingBetaTheta}).  
							
				\item Sample $\bar{\bx}^{(b)} \sim \mathcal{N}\left(\bar{\bx}, \widehat{\text{Cov}}(\bx)\right)$.
				
				\item Draw $\mathbf{\mathcal{C}}^{(b)}_{2:p,2:p}$ from a Wishart distribution with parameters $\left\{n, \widehat{\text{Cov}}(\bx)\right\}$
				and define  
				$$
				\mathbf{\mathcal{C}}^{(b)} = 
				\begin{pmatrix} 
					n & n\bar{\bx}^{(b)T} \\ 
					n\bar{\bx}^{(b)} & \mathbf{\mathcal{C}}^{(b)}_{2:p,2:p} + n\bar{\bx}^{(b)}\bar{\bx}^{(b)T} 
				\end{pmatrix}.
				$$

				\item Sample with replacement $m$ observations from $\tilde{\bX}_n$.
				\item Compute the Cholesky decomposition of $\mathbf{\mathcal{C}}^{(b)}$ and apply a ``whiten-recolor" transformation to the resampled observations, obtaining $\tilde{\bX}_n^{(b)}$. 
				\item Compute  
				$$
				\boldsymbol{\psi}^{(b)}_m = \frac{1}{m} \sum_{i=1}^m 
				\left\{ \tilde{\bx}_{n,i}^{(b)T} \hat{\btheta}^{(b)} - \mu\left(\tilde{\bx}_{n,i}^{(b)T} \hat{\bbeta}^{(b)}\right) \right\} \tilde{\bx}^{(b)}_{n,i}.
				$$  
			\end{enumerate}
			\EndFor 
			
			\Return the empirical variance matrix of $\boldsymbol{\psi}^{(b)}_m$ for $b=1,\ldots,B$.
		\end{algorithmic}
	\end{algorithm}
	
	
	\section{Privacy Concerns}\label{sec:privacy}
	
	Releasing both the Gram matrix $\bX^\top \bX$ and the synthetic data $\tilde{\bX}_n$ may raise privacy concerns. While synthetic data reduces re-identification risk, the additional disclosure of $\bX^\top \bX$ invites the question of whether an adversary could exploit both to reconstruct parts of the original data $\bX$. The feasibility of such an attack depends on the degrees of freedom in reconstructing $\bX$, the additional structural information revealed by $\tilde{\bX}_n$, and whether limiting the size of the synthetic sample mitigates this risk.
	
	Given only $\bX^\top \bX$, the possible solutions for $\bX$ form a large equivalence class under left-orthogonal transformations. If $\bX_0$ satisfies $\bX_0^\top \bX_0 = \bX^\top \bX$, then so does $\mathbf{Q}\bX_0$ for any $n \times n$ orthogonal matrix $\mathbf{Q}$. The solution space has dimension $np - p(p+1)/2$, since $\bX$ has $np$ free parameters and $\bX^\top \bX$ imposes $p(p+1)/2$ independent symmetric constraints. Each additional observation adds $p$ degrees of freedom, so when $n \gg p$, the solution space becomes large, making exact recovery of $\bX$ from $\bX^\top \bX$ alone impractical.

	The synthetic data $\tilde{\bX}_n$ may reveal structural information that constrains the solution space for $\bX$. For example, if a covariate is known to lie in $[a, b]$, the degrees of freedom remain unchanged, but the feasible set is restricted to a subset of the continuous manifold. Similarly, if a column takes only discrete values, each entry is confined to a finite set, further limiting possible reconstructions. Still, even for binary variables, exact recovery of individual observations from marginal and pairwise co-occurrence counts is generally infeasible when $n \gg p$.
	
	One way to reduce privacy risk is to limit the number of synthetic observations. In the simulation section, we show that for large $n$, generating $m < n$ synthetic samples still provides reasonable statistical efficiency. A smaller synthetic dataset lowers the chance that any synthetic observation closely matches an original one. While one might imagine selecting a reconstruction solution $\bX_0$ that is close to $\tilde{\bX}_n$, an attacker could embed all $m$ synthetic rows into $\bX_0$ while retaining $(n - m)p - p(p+1)/2$ remaining degrees of freedom. Thus, incorporating $\tilde{\bX}_n$ into a reconstruction does not necessarily yield a solution closer to the true $\bX$. Whether $\tilde{\bX}_n$ provides meaningful information depends on how it was generated.
	
	Nevertheless, increasing privacy by adding noise to the Gram matrix, whether or not formally satisfying DP, may still be desirable to guard against broader classes of privacy attacks. Methods from private linear regression, such as sufficient statistics perturbation \citep{vu2009differential,wang2018revisiting,ferrando2024private}, could be imported to our setting. However, these perturbations substantially complicate inference, and we defer investigation to future work.

	\FloatBarrier
	\section{Simulation Study} \label{sec:sims}
	
	We conducted a simulation study to evaluate the performance of Poisson and logistic regression models under two covariate structures, labeled Setting A and Setting B. For each setting, datasets of size $n = 200$ and $n = 500$ were generated. From each dataset, $m = n$ synthetic observations were produced using two methods. The first method uses the {\tt R synthpop} package \citep{synthpop}, which employs classification and regression trees with spline-based smoothing to ensure synthetic samples differ from the originals. The second method uses ADS-GAN \citep{yoon2020anonymization}, implemented via the Python {\tt synthcity} package, which generates synthetic data using a generative adversarial network.

	In Setting A, covariates were generated using a Gaussian copula: a 19-dimensional vector was sampled from a multivariate normal distribution with mean zero and pairwise correlation 0.5, and then transformed using the normal CDF. The true coefficient vector is $\bbeta^* = (0.5, 0.9, -0.7, -0.5, 1, -0.7, -0.2, 0, \dots, 0)$, where only the intercept and the first six covariates were assigned nonzero coefficients. The response was simulated under either a Poisson model, by exponentiating the linear predictor to obtain the intensity, or a logistic model, by mapping the linear predictor to a probability via the logistic function.
	
	In Setting B, covariates were generated from exponential and normal distributions, with additional dependencies introduced through linear combinations. Specifically, $X_{1:3}$ and $X_{16:19}$ were drawn independently from an exponential distribution with mean one, while $X_{4:6}$ and $X_{11:15}$ were drawn independently from a standard normal distribution. The remaining covariates were defined as $X_7 = X_1 / 3 + 2 X_2 / 3 + \varepsilon_7$, $X_8 = X_1 / 2 + X_6 / 2 + \varepsilon_8$, $X_9 = X_2 + \varepsilon_9$, and $X_{10} = X_5 / 5 + 3 X_6 / 4 + \varepsilon_{10}$, with each $\varepsilon_j$ independently sampled from a standard normal distribution. The coefficient vector was set to $\bbeta^* = (0.15, 0.1, 0.15, 0.2, -0.1, 0.1, 0.2, -0.1, -0.1, 0.1, -0.1, 0, \dots, 0)$, where only the first eleven coefficients were nonzero, including the intercept. In each scenario, defined by setting, sample size, regression type, and synthetic model, 500 repetitions were conducted.

	We report here only the results for the first ten coefficients and only for the ADS-GAN model. Full results, including coefficients 11–20 for ADS-GAN and all results from \texttt{synthpop}, are provided in Appendix~\ref{appen:sim}. Table~\ref{tab:PoisCity1:10} presents the Poisson regression results, and Table~\ref{tab:LogisCity1:10} presents the logistic regression results. In both tables, MLE refers to the mean of $\hat{\bbeta}$ across repetitions based on the original sample. `Syn-novel' denotes the mean of our proposed estimator $\tilde{\bbeta}$ and `Syn-naive' refers to the coefficients obtained by fitting Poisson or logistic regression directly to the synthetic data $(\tilde{\by}_n, \tilde{\bX}_n)$. The tables also report the empirical standard error (SE) of each estimator and the coverage of 95\% Wald-type confidence intervals for $\bbeta^*$, based on our proposed estimator with bootstrap inference using $B = 1,000$ repetitions, as described in Algorithm~\ref{alg:boot}.
	
	Across both regression types and data generation settings, our estimator closely approximates the original-data-based MLE and the true parameter values, while Syn-naive exhibits substantial bias. The SEs of our estimator are only slightly higher than those of the MLE when $n = 200$ and nearly identical when $n = 500$, whereas those of Syn-naive are considerably larger. The confidence intervals achieve nominal coverage, demonstrating the effectiveness of the bootstrap procedure. These findings indicate that even with $n = 200$, our improved synthetic-data-based estimator is already highly efficient. The results for coefficients 11–20 and for \texttt{synthpop}, presented in Tables~\ref{tab:PoisCity11:20}--\ref{tab:LogisBPop} in Appendix~\ref{appen:sim}, support the same conclusion.
	
	To assess how performance varies with $m$, we generated a dataset of size $n = 20{,}000$ and trained the synthetic data generator. Synthetic datasets of sizes $m = 200, 500, 1000, 5000,$ and $20000$ were then created, applying both Syn-novel and Syn-naive. These simulations, conducted for Poisson regression under Setting~A, involved 500 repetitions, and due to the computational cost of ADS-GAN, only \texttt{synthpop} was used. Table~\ref{tab:res20k1:4} reports results for coefficients $\bbeta^*_{1:4}$. Since the MLE does not depend on $m$, both the MLE and its standard error (“MLE-SE”) appear only once. The table also presents empirical standard errors (“Syn-novel-emp-SE”), computed across repetitions, and estimated standard errors (“Syn-novel-est-SE”), representing the mean estimated SE from our variance estimator. The close agreement between empirical and estimated SEs, along with accurate coverage, demonstrates the reliability of our variance estimator. Although Syn-naive improves relative to Tables~\ref{tab:PoisCity1:10}--\ref{tab:LogisCity1:10}, substantial bias remains even at $n = m = 20{,}000$, whereas our method achieves lower SEs with $m = 500$ than Syn-naive does with $m = 20{,}000$. Additionally, Figure~\ref{fig:Histplots} presents overlaid density histograms comparing the MLE and Syn-novel distributions for varying $m$, showing increasing alignment as $m$ grows.
	
	A key takeaway from these results is that when $n$ is large, a moderately sized synthetic dataset often suffices. While $m = 200$ inflates the SE relative to the MLE-SE, choosing $m = 500$ or $m = 1000$ provides a good balance between synthetic data size and efficiency, and $m = 5000$ achieves nearly the same efficiency as the original MLE. Reducing $m$ is beneficial not only for computational and memory efficiency but also for privacy, as it lowers the chance that synthetic observations closely resemble real ones (see Section~\ref{sec:privacy}). Complete results for the remaining 16 coefficients are reported in Tables~\ref{tab:res20k5:8}--\ref{tab:res20k17:20} in Appendix~\ref{appen:sim}.
	
	To assess whether Assumption~\ref{as:biasNorm} holds empirically, we examined Poisson regression under Settings~A and B and logistic regression under Setting~B, using $n$ values ranging from 200 to 20,000 (logistic regression under Setting~A was not tested, as it is expected to behave well). In each case, we generated $m = 500{,}000$ synthetic observations, and approximated the expectation under the synthetic model with $\hat{\bbeta}$ plugged in. This was repeated 500 times to obtain realizations of the bias term $\E_{\hat{F}_n} \{\boldsymbol{\phi}_m (\hat{\bbeta})\}$. For computational reasons, \texttt{synthpop} was used in these experiments. We evaluated two aspects: root-$n$ scaling and normality. Figure~\ref{fig:MSEplots} shows the mean squared error (MSE) of the bias terms, separately for each covariate (grey lines) and averaged (black line), multiplied by $n$. Ideally, the scaled MSE should stabilize as $n$ increases. For Poisson regression in Setting~A and logistic regression in Setting~B, the MSE stabilizes and remains extremely small even after scaling. In Poisson regression under Setting~B, some covariates (those exponentially distributed) show an initial decrease followed by an increase, suggesting a possible violation of Assumption~\ref{as:biasNorm}. Nevertheless, the scaled MSE remains below 1 at $n = 20{,}000$, and estimation remains accurate, demonstrating robustness.

	For normality, we used Chi-squared Q-Q plots from the {\tt heplots} package in {\tt R} \citep{heplots}, and the results are shown in Figure~\ref{fig:CQplots} in Appendix~\ref{appen:sim}. In Poisson regression under Setting~A, normality appears reasonable at $n = 5{,}000$ and holds at $n = 20{,}000$. In logistic regression under Setting~B, normality is not satisfied at $n = 5{,}000$ but holds at $n = 20{,}000$. In Poisson regression under Setting~B, normality is not satisfied even at $n = 20{,}000$, indicating a violation of Assumption~\ref{as:biasNorm}. Overall, Assumption~\ref{as:biasNorm} appears reasonable for logistic regression and for Poisson regression in Setting~A but may not hold in Poisson regression under Setting~B.
	
	In summary, the simulations highlight the robustness and effectiveness of our method, even when certain assumptions are violated, as demonstrated by its performance in Poisson regression under Setting~B. They also underscore the substantial potential of incorporating summary statistics into inference based on synthetic data when such statistics can be shared.
	
	\begin{table}[ht]
		\spacingset{1}
		\centering 
		\begin{tabular}{lcccccccccc} 
			\hline
			$\bbeta^*$ (Setting A)  & 0.500 & 0.900 & -0.700 & -0.500 & 1.000 & -0.700 & -0.200 & 0.000 & 0.000 & 0.000 \\  
			\hline
			$n = m = 200$ &&&&&&&&&&\\
			MLE & 0.469 & 0.881 & -0.693 & -0.511 & 1.012 & -0.687 & -0.181 & 0.017 & -0.001 & -0.001 \\ 
			Syn-novel & 0.469 & 0.899 & -0.705 & -0.517 & 1.037 & -0.698 & -0.185 & 0.020 & -0.001 & -0.003 \\ 
			Syn-naive & 0.458 & 0.228 & -0.182 & -0.146 & 0.277 & -0.204 & -0.058 & -0.003 & 0.004 & 0.010 \\ 
			MLE-SE & 0.148 & 0.279 & 0.275 & 0.282 & 0.286 & 0.274 & 0.262 & 0.280 & 0.280 & 0.283 \\ 
			Syn-novel-SE & 0.151 & 0.293 & 0.288 & 0.290 & 0.308 & 0.293 & 0.270 & 0.292 & 0.287 & 0.297 \\ 
			Syn-naive-SE & 0.344 & 0.380 & 0.399 & 0.390 & 0.403 & 0.368 & 0.408 & 0.403 & 0.395 & 0.409 \\ 
			Coverage-MLE & 0.968 & 0.964 & 0.954 & 0.952 & 0.948 & 0.956 & 0.972 & 0.952 & 0.948 & 0.948 \\ 
			Coverage-novel & 0.974 & 0.970 & 0.962 & 0.952 & 0.952 & 0.958 & 0.972 & 0.958 & 0.958 & 0.946 \\ 
			\hline
			$n = m = 500$ &&&&&&&&&&\\
			MLE & 0.489 & 0.902 & -0.701 & -0.496 & 1.010 & -0.711 & -0.198 & -0.009 & -0.003 & 0.004 \\ 
			Syn-novel & 0.489 & 0.906 & -0.706 & -0.496 & 1.015 & -0.714 & -0.198 & -0.008 & -0.003 & 0.005 \\ 
			Syn-naive & 0.508 & 0.227 & -0.205 & -0.155 & 0.256 & -0.189 & -0.064 & 0.001 & -0.004 & -0.012 \\ 
			MLE-SE & 0.096 & 0.174 & 0.166 & 0.164 & 0.166 & 0.177 & 0.169 & 0.162 & 0.167 & 0.173 \\ 
			Syn-novel-SE & 0.097 & 0.184 & 0.171 & 0.170 & 0.171 & 0.181 & 0.170 & 0.161 & 0.172 & 0.177 \\ 
			Syn-naive-SE & 0.209 & 0.235 & 0.246 & 0.244 & 0.244 & 0.238 & 0.256 & 0.237 & 0.244 & 0.252 \\ 
			Coverage-MLE & 0.950 & 0.954 & 0.944 & 0.956 & 0.964 & 0.936 & 0.960 & 0.966 & 0.948 & 0.948 \\ 
			Coverage-novel & 0.944 & 0.946 & 0.956 & 0.960 & 0.966 & 0.942 & 0.958 & 0.966 & 0.952 & 0.956 \\ 
			\hline
			$\bbeta^*$ (Setting B) & 0.150 & 0.100 & 0.150 & 0.200 & -0.100 & 0.100 & 0.200 & -0.100 & -0.100 & 0.100 \\  
			\hline
			$n = m = 200$ &&&&&&&&&&\\
			MLE & 0.137 & 0.095 & 0.144 & 0.204 & -0.102 & 0.101 & 0.199 & -0.098 & -0.099 & 0.101 \\ 
			Syn-novel & 0.109 & 0.094 & 0.155 & 0.218 & -0.105 & 0.104 & 0.203 & -0.098 & -0.101 & 0.099 \\ 
			Syn-naive & 0.374 & 0.017 & 0.046 & 0.072 & -0.050 & 0.047 & 0.046 & -0.006 & -0.025 & 0.060 \\ 
			MLE-SE & 0.168 & 0.065 & 0.087 & 0.051 & 0.055 & 0.057 & 0.070 & 0.057 & 0.056 & 0.056 \\ 
			Syn-novel-SE & 0.179 & 0.068 & 0.090 & 0.060 & 0.060 & 0.060 & 0.075 & 0.060 & 0.061 & 0.055 \\ 
			Syn-naive-SE & 0.293 & 0.092 & 0.100 & 0.091 & 0.104 & 0.105 & 0.125 & 0.090 & 0.090 & 0.081 \\ 
			Coverage-MLE & 0.948 & 0.948 & 0.954 & 0.932 & 0.950 & 0.942 & 0.962 & 0.944 & 0.950 & 0.956 \\ 
			Coverage-novel & 0.940 & 0.960 & 0.956 & 0.920 & 0.940 & 0.950 & 0.968 & 0.944 & 0.946 & 0.966 \\ 
			\hline
			$n = m = 500$ &&&&&&&&&&\\
			MLE & 0.144 & 0.102 & 0.147 & 0.199 & -0.100 & 0.100 & 0.201 & -0.101 & -0.102 & 0.102 \\ 
			Syn-novel & 0.130 & 0.101 & 0.151 & 0.207 & -0.101 & 0.100 & 0.202 & -0.099 & -0.104 & 0.101 \\ 
			Syn-naive & 0.234 & 0.013 & 0.063 & 0.096 & -0.065 & 0.055 & 0.063 & -0.019 & -0.030 & 0.086 \\ 
			MLE-SE & 0.094 & 0.043 & 0.050 & 0.029 & 0.035 & 0.035 & 0.044 & 0.033 & 0.034 & 0.034 \\ 
			Syn-novel-SE & 0.100 & 0.043 & 0.048 & 0.035 & 0.037 & 0.036 & 0.047 & 0.034 & 0.036 & 0.031 \\ 
			Syn-naive-SE & 0.253 & 0.063 & 0.065 & 0.065 & 0.067 & 0.068 & 0.083 & 0.060 & 0.069 & 0.052 \\ 
			Coverage-MLE & 0.958 & 0.930 & 0.954 & 0.942 & 0.940 & 0.952 & 0.962 & 0.954 & 0.952 & 0.944 \\ 
			Coverage-novel & 0.952 & 0.934 & 0.958 & 0.908 & 0.930 & 0.952 & 0.964 & 0.952 & 0.942 & 0.964 \\ 
			\hline
		\end{tabular}
		\caption{Simulation results for coefficients $\bbeta^*_{1:10}$ in a Poisson regression under Settings A and B,   based on sample sizes of $n = 200$ and $n = 500$, using the ADS-GAN synthetic model and $m=n$. Results are based on 500 repetitions.} \label{tab:PoisCity1:10}
	\end{table}

	\begin{table}[ht]
		\spacingset{1}
		\centering
		\begin{tabular}{lcccccccccc}
			\hline
			$\bbeta^*$ (Setting A) & 0.500 & 0.900 & -0.700 & -0.500 & 1.000 & -0.700 & -0.200 & 0.000 & 0.000 & 0.000 \\ 
			\hline
			$n = m = 200$ &&&&&&&&&&\\
			MLE & 0.563 & 1.018 & -0.777 & -0.598 & 1.093 & -0.804 & -0.262 & -0.012 & -0.014 & -0.022 \\ 
			Syn-novel & 0.566 & 1.028 & -0.782 & -0.606 & 1.107 & -0.815 & -0.264 & -0.019 & -0.019 & -0.025 \\ 
			Syn-naive & 0.765 & 0.662 & -0.588 & -0.300 & 0.693 & -0.383 & -0.218 & -0.005 & 0.016 & -0.100 \\ 
			MLE-SE & 0.426 & 0.783 & 0.784 & 0.836 & 0.889 & 0.793 & 0.787 & 0.819 & 0.837 & 0.805 \\ 
			Syn-novel-SE & 0.424 & 0.794 & 0.799 & 0.844 & 0.907 & 0.809 & 0.799 & 0.834 & 0.852 & 0.815 \\ 
			Syn-naive-SE & 1.065 & 1.419 & 1.368 & 1.437 & 1.407 & 1.397 & 1.419 & 1.408 & 1.405 & 1.389 \\ 
			Coverage-MLE & 0.956 & 0.956 & 0.948 & 0.944 & 0.916 & 0.946 & 0.960 & 0.930 & 0.944 & 0.948 \\ 
			Coverage-novel & 0.958 & 0.956 & 0.952 & 0.946 & 0.918 & 0.948 & 0.964 & 0.942 & 0.942 & 0.950 \\ 
			\hline
			$n = m = 500$ &&&&&&&&&&\\
			MLE & 0.519 & 0.939 & -0.753 & -0.513 & 1.058 & -0.719 & -0.233 & 0.021 & 0.010 & 0.005 \\ 
			Syn-novel & 0.519 & 0.943 & -0.756 & -0.517 & 1.062 & -0.723 & -0.234 & 0.022 & 0.012 & 0.003 \\ 
			Syn-naive & 0.421 & 0.556 & -0.387 & -0.294 & 0.599 & -0.438 & -0.129 & -0.008 & 0.025 & -0.024 \\ 
			MLE-SE & 0.263 & 0.485 & 0.462 & 0.487 & 0.439 & 0.493 & 0.461 & 0.455 & 0.459 & 0.464 \\ 
			Syn-novel-SE & 0.262 & 0.488 & 0.465 & 0.491 & 0.442 & 0.495 & 0.463 & 0.456 & 0.462 & 0.466 \\ 
			Syn-naive-SE & 0.554 & 0.638 & 0.689 & 0.666 & 0.679 & 0.706 & 0.628 & 0.676 & 0.672 & 0.663 \\ 
			Coverage-MLE & 0.940 & 0.934 & 0.930 & 0.936 & 0.952 & 0.926 & 0.952 & 0.948 & 0.946 & 0.934 \\ 
			Coverage-novel & 0.942 & 0.936 & 0.928 & 0.936 & 0.950 & 0.926 & 0.950 & 0.950 & 0.948 & 0.934 \\ 
			\hline
			$\bbeta^*$ (Setting B) & 0.150 & 0.100 & 0.150 & 0.200 & -0.100 & 0.100 & 0.200 & -0.100 & -0.100 & 0.100 \\ 
			\hline
			$n = m = 200$ &&&&&&&&&&\\
			MLE & 0.127 & 0.122 & 0.192 & 0.230 & -0.114 & 0.112 & 0.235 & -0.126 & -0.113 & 0.111 \\ 
			Syn-novel & 0.136 & 0.122 & 0.191 & 0.226 & -0.114 & 0.113 & 0.236 & -0.126 & -0.114 & 0.113 \\ 
			Syn-naive & 0.362 & 0.095 & 0.153 & 0.176 & -0.115 & 0.121 & 0.154 & -0.087 & -0.108 & 0.173 \\ 
			MLE-SE & 0.495 & 0.218 & 0.295 & 0.202 & 0.175 & 0.185 & 0.218 & 0.186 & 0.175 & 0.176 \\ 
			Syn-novel-SE & 0.493 & 0.219 & 0.305 & 0.194 & 0.175 & 0.187 & 0.220 & 0.189 & 0.176 & 0.184 \\ 
			Syn-naive-SE & 1.236 & 0.452 & 0.455 & 0.440 & 0.525 & 0.503 & 0.530 & 0.444 & 0.436 & 0.373 \\ 
			Coverage-MLE & 0.940 & 0.934 & 0.926 & 0.932 & 0.940 & 0.934 & 0.956 & 0.936 & 0.934 & 0.940 \\ 
			Coverage-novel & 0.948 & 0.946 & 0.938 & 0.956 & 0.946 & 0.942 & 0.962 & 0.946 & 0.944 & 0.948 \\ 
			\hline
			$n = m = 500$ &&&&&&&&&&\\
			MLE & 0.158 & 0.101 & 0.141 & 0.209 & -0.103 & 0.098 & 0.207 & -0.096 & -0.103 & 0.110 \\ 
			Syn-novel & 0.161 & 0.103 & 0.140 & 0.207 & -0.103 & 0.097 & 0.207 & -0.097 & -0.103 & 0.111 \\ 
			Syn-naive & 0.256 & 0.026 & 0.085 & 0.121 & -0.109 & 0.072 & 0.137 & -0.048 & -0.059 & 0.116 \\ 
			MLE-SE & 0.277 & 0.123 & 0.167 & 0.110 & 0.105 & 0.101 & 0.138 & 0.097 & 0.099 & 0.105 \\ 
			Syn-novel-SE & 0.275 & 0.124 & 0.172 & 0.108 & 0.105 & 0.101 & 0.137 & 0.098 & 0.098 & 0.109 \\ 
			Syn-naive-SE & 0.564 & 0.187 & 0.195 & 0.191 & 0.222 & 0.215 & 0.239 & 0.177 & 0.174 & 0.166 \\ 
			Coverage-MLE & 0.966 & 0.942 & 0.928 & 0.956 & 0.932 & 0.956 & 0.952 & 0.960 & 0.954 & 0.932 \\ 
			Coverage-novel & 0.964 & 0.936 & 0.930 & 0.960 & 0.934 & 0.958 & 0.948 & 0.958 & 0.962 & 0.934 \\ 
			\hline
		\end{tabular}
		\caption{Simulation results for coefficients $\bbeta^*_{1:10}$ in a Logistic regression under Settings A and B, based on sample sizes of $n = 200$ and $n = 500$, using the ADS-GAN synthetic model and $m=n$. Results are based on 500 repetitions.} \label{tab:LogisCity1:10}
	\end{table}

	\begin{table}[ht]
		\spacingset{1}
		\centering
		\resizebox{\textwidth}{!}{%
			\begin{tabular}{l|ccccc|ccccc}
				\hline
				$n = 20{,}000$	& \multicolumn{5}{c|}{$\bbeta^*_1=0.5$} & \multicolumn{5}{c}{$\bbeta^*_2=0.9$} \\
				m & 200 & 500 & 1000 & 5000 & 20000 & 200 & 500 & 1000 & 5000 & 20000 \\
				\hline
				MLE 
				&  &  & 0.499 &  &  
				&  &  & 0.899 &  &  \\
				Syn-novel 
				& 0.499 & 0.499 & 0.499 & 0.499 & 0.499 
				& 0.915 & 0.905 & 0.903 & 0.900 & 0.900 \\
				Syn-naive 
				& 0.437 & 0.480 & 0.493 & 0.507 & 0.512 
				& 0.808 & 0.809 & 0.808 & 0.819 & 0.822 \\
				MLE-SE 
				&  &  & 0.015 &  &  
				&  &  & 0.026 &  &  \\
				Syn-novel-emp-SE 
				& 0.020 & 0.017 & 0.016 & 0.015 & 0.015 
				& 0.047 & 0.036 & 0.032 & 0.027 & 0.026 \\
				Syn-novel-est-SE 
				& 0.020 & 0.017 & 0.016 & 0.015 & 0.015 
				& 0.046 & 0.036 & 0.032 & 0.028 & 0.027 \\
				Syn-naive-SE 
				& 0.173 & 0.110 & 0.074 & 0.034 & 0.024 
				& 0.287 & 0.171 & 0.125 & 0.061 & 0.042 \\
				Coverage-MLE & & &0.954 & & & & &0.960  & & \\
				Coverage 
				& 0.952 & 0.946 & 0.944 & 0.942 & 0.956 
				& 0.950 & 0.962 & 0.954 & 0.962 & 0.960 \\
				\hline
				\multicolumn{11}{c}{} \\ 
				\hline
				$n = 20{,}000$	& \multicolumn{5}{c|}{$\bbeta^*_3=-0.7$} & \multicolumn{5}{c}{$\bbeta^*_4=-0.5$} \\
				m & 200 & 500 & 1000 & 5000 & 20000 & 200 & 500 & 1000 & 5000 & 20000 \\
				\hline
				MLE 
				&  &  & -0.701 &  &  
				&  &  & -0.501 &  &  \\
				Syn-novel 
				& -0.711 & -0.705 & -0.703 & -0.701 & -0.701 
				& -0.507 & -0.503 & -0.502 & -0.501 & -0.501 \\
				Syn-naive 
				& -0.654 & -0.654 & -0.650 & -0.653 & -0.653 
				& -0.428 & -0.432 & -0.439 & -0.430 & -0.432 \\
				MLE-SE 
				&  &  & 0.025 &  &  
				&  &  & 0.026 &  &  \\
				Syn-novel-emp-SE 
				& 0.043 & 0.035 & 0.031 & 0.026 & 0.026 
				& 0.045 & 0.035 & 0.031 & 0.026 & 0.026 \\
				Syn-novel-est-SE 
				& 0.043 & 0.034 & 0.031 & 0.027 & 0.027 
				& 0.043 & 0.034 & 0.031 & 0.027 & 0.027 \\
				Syn-naive-SE 
				& 0.300 & 0.177 & 0.131 & 0.063 & 0.043 
				& 0.295 & 0.187 & 0.130 & 0.064 & 0.044 \\
				Coverage-MLE & & &0.960  & & & & &0.954  & & \\
				Coverage 
				& 0.940 & 0.934 & 0.952 & 0.954 & 0.956 
				& 0.914 & 0.954 & 0.950 & 0.964 & 0.960 \\
				\hline
			\end{tabular}%
		}
		\caption{Simulation results for the coefficients $\bbeta^*_{1:4}$ in a Poisson regression under Setting A, with a fixed sample size of $n = 20{,}000$ and varying $m$ values, using the \texttt{synthpop} synthetic model. Results are based on 500 repetitions. Since the MLE is invariant to $m$, only a single value is shown.}
		\label{tab:res20k1:4}
	\end{table}

	\begin{landscape}  
		\begin{figure}[htbp]
			\centering
			\resizebox{1.25\textwidth}{!}{  
				\includegraphics{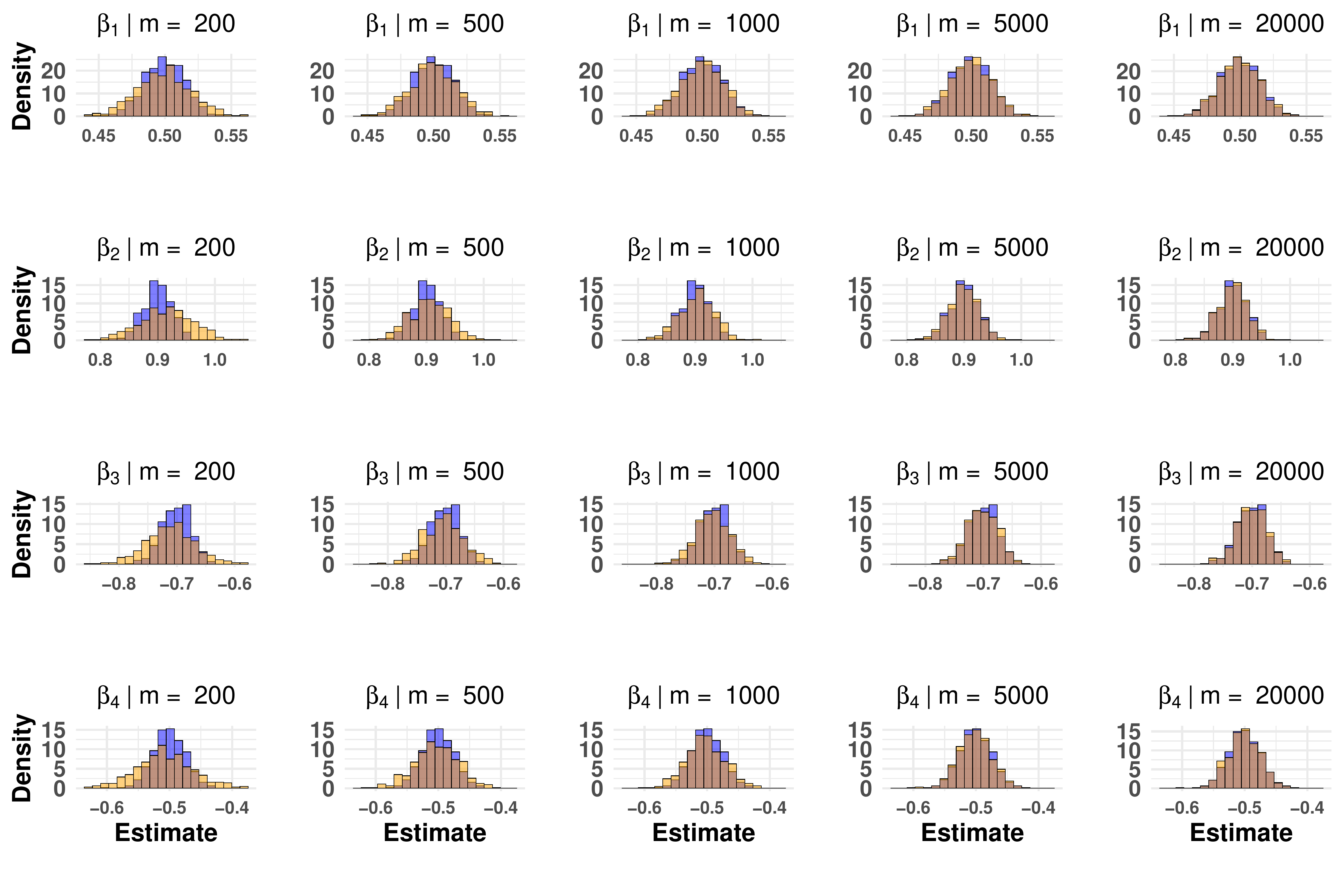}
			}
			\caption{Overlaid density histograms comparing the MLE (\textcolor{blue}{blue}) and Syn-novel (\textcolor{orange}{orange}) methods for estimating the coefficients $\bbeta^*_{1:4}$ in a Poisson regression under Setting A. The results are based on 500 repetitions with $n = 20{,}000$ and varying $m$ values. 	\label{fig:Histplots}}
			
		\end{figure}
	\end{landscape}

	\begin{figure}[h]
		\centering
		\begin{subfigure}[b]{0.35\textwidth}  
			\includegraphics[width=\textwidth]{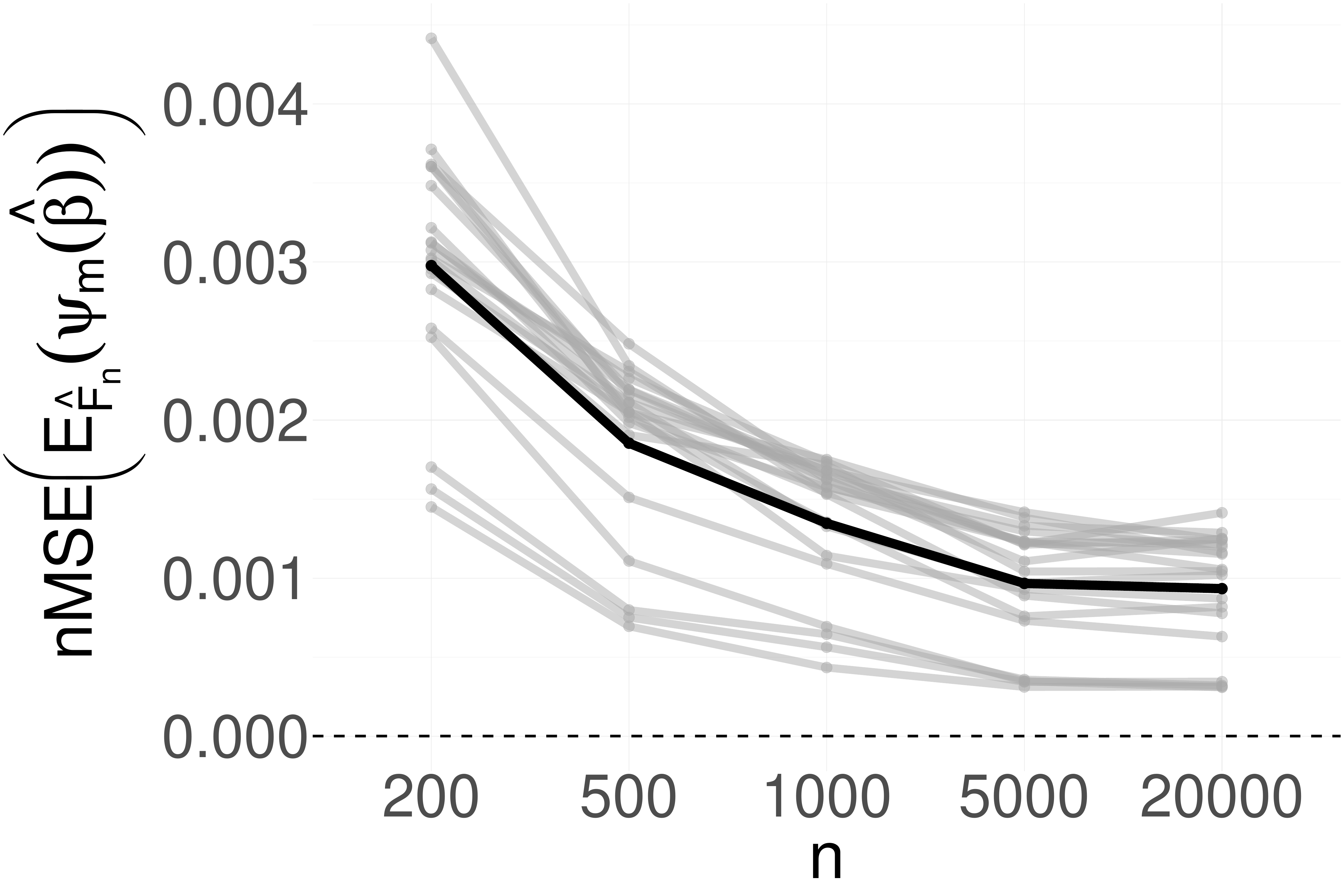}
			\caption{Poisson Regression--A}
		\end{subfigure}
		\begin{subfigure}[b]{0.31\textwidth}  
			\includegraphics[width=\textwidth, clip=TRUE, trim=8cm 0cm 0cm 0cm]{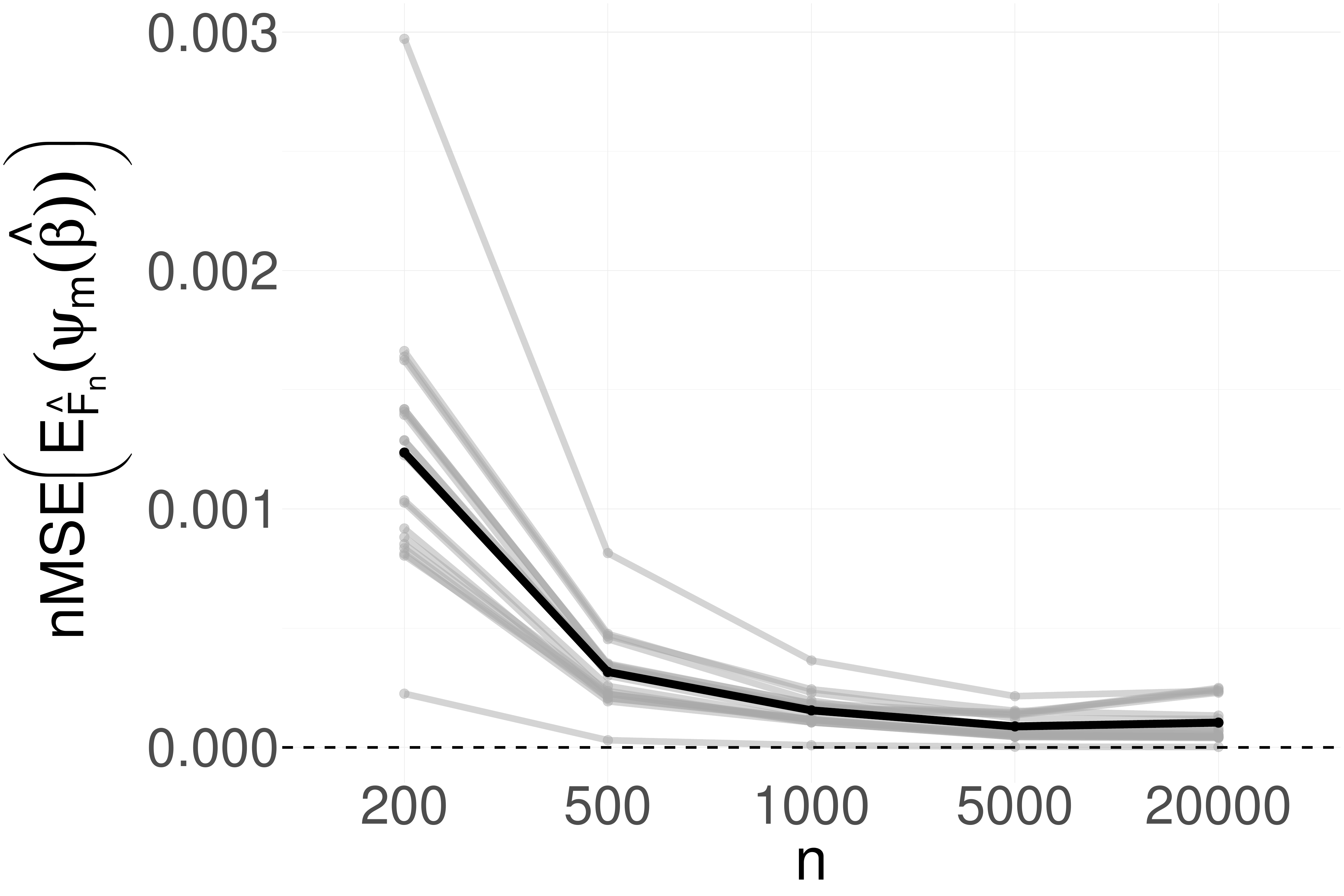}
			\caption{Logistic Regression--B}
		\end{subfigure}
		\begin{subfigure}[b]{0.31\textwidth}
			\includegraphics[width=\textwidth, clip=TRUE, trim=8cm 0cm 0cm 0cm]{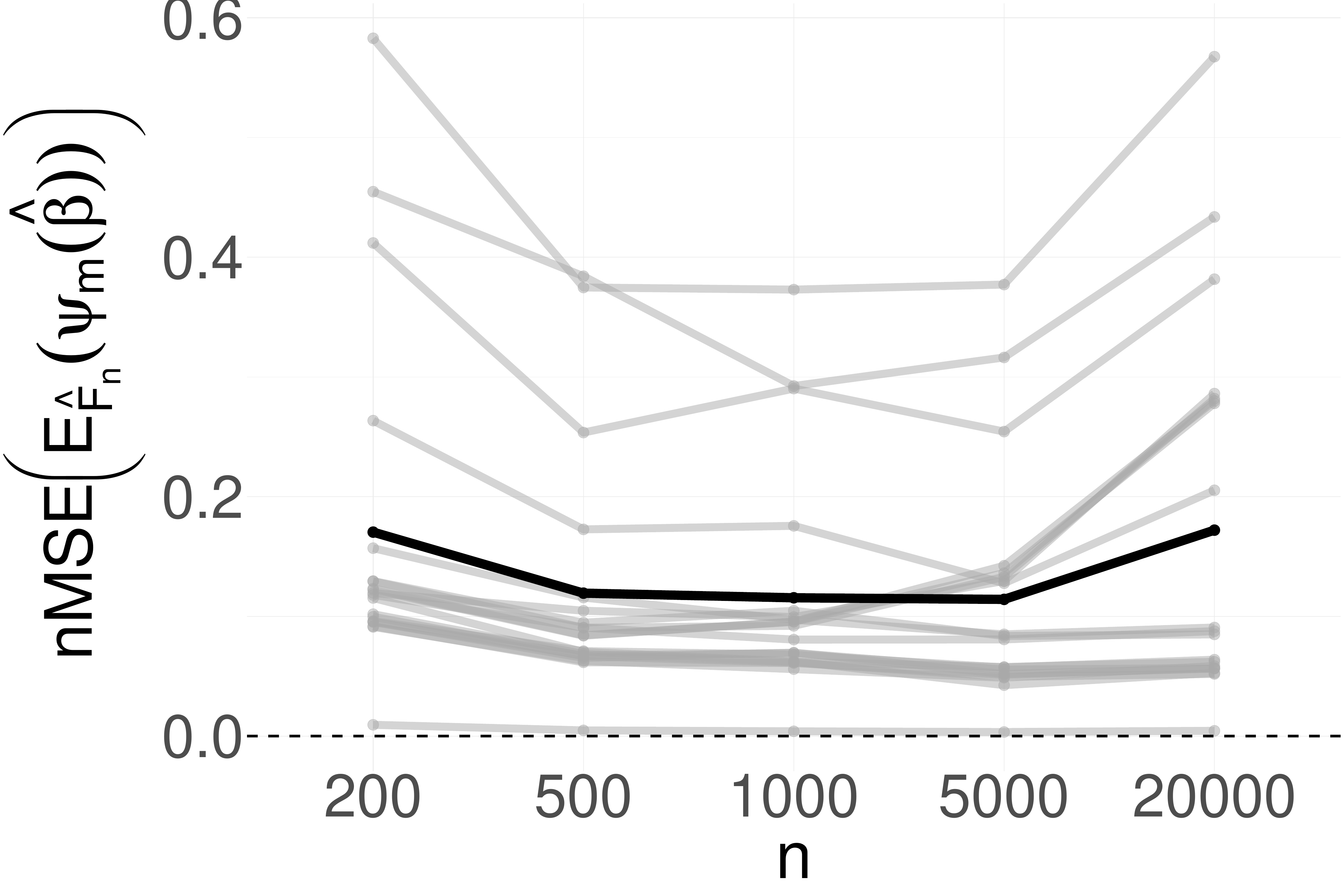}
			\caption{Poisson Regression--B}
		\end{subfigure}
		
		\caption{Empirical scaled MSE plots for the ``estimating equation's bias," as defined in Assumption~\ref{as:biasNorm}, across different values of $n$, based on 500 repetitions. Grey lines represent the MSE of the 20 covariates, while the black line denotes the average MSE across all covariates.		\label{fig:MSEplots}}
	\end{figure}

	\FloatBarrier
	
	\section{Real Data Analysis} \label{sec:dataAnalysis}
	
	To further illustrate our method, we use the 2015 U.S.\ Linked Birth/Infant Death Cohort Data \citep{nchs2015linked}, which includes all live births in the U.S.\ that year, linked to death records indicating one-year infant mortality. The dataset contains de-identified demographic, socioeconomic, and medical information on infants and their parents. We analyze the association between socio-demographic factors and one-year infant mortality under a hypothetical synthetic data release, evaluating whether inference aligns with that based on the original public-use file.

	We consider the following covariates in the analysis. Maternal health variables include maternal age, body mass index (BMI), diabetes, chronic hypertension, pregnancy-associated hypertension, eclampsia, and smoking before and during pregnancy. Delivery-related covariates include facility type (hospital vs. other), birth attendant (doctor, midwife, other), labor induction, mode of delivery (forceps or vacuum), breech presentation, birth order, and plurality (singleton, twins, three or more). Socio-demographic variables include maternal and paternal race (white, black, other), education (less than high school, high school, more than high school), marital status, and participation in the Women, Infants, and Children (WIC) nutrition program. Infant sex is also included. Postnatal characteristics such as birth weight, Apgar score, and congenital anomalies such as Down syndrome, spina bifida, and anencephaly are excluded, as they may lie on the causal pathway and explain away part of the socio-demographic effects of interest.
	
	Tables~\ref{tab:realTable1Part1}--\ref{tab:realTable1Part2} in Appendix~\ref{appen:real} report descriptive statistics of the original and synthetic datasets, stratified by one-year mortality, where ADS-GAN was used for synthetic data generation. For some variables, the original and synthetic data align in sample proportions (for categorical variables) or means and standard deviations (for continuous variables), while for others, larger discrepancies are observed. Notably, death was not treated differently during synthetic data generation and was modeled like any other variable.
	
	We note several limitations in the synthetic data generation process. First, the original dataset contained about 4 million records, too large for ADS-GAN to process, so we subsampled to $n = 500{,}000$, keeping all deaths and sampling uniformly among survivors. This distorts the marginal prevalence of death but affects only the intercept, which we did not adjust for in synthetic data generation. We treat this subsample as the effective original data. Second, missing values were imputed using the \texttt{mice} package in \texttt{R} \citep{van2011mice}, treating the completed data as ground truth. Incorporating uncertainty via multiple imputation is beyond the scope of this paper. Third, the dataset was inherently messy. Due to staggered adoption of the 2003 revision of the U.S. Standard Certificate of Live Birth, some variables appeared only in revised states, while others had multiple forms, such as continuous versus grouped maternal age or detailed versus broad race categories. We consulted the public-use data guide and retained a harmonized subset of broadly available variables, excluding revision-specific and redundant encodings. In practice, such preprocessing (imputation, deduplication, harmonization) should precede synthetic data generation, as using raw or inconsistently coded sources risks embedding structural artifacts.
	
	Figure~\ref{fig:forestplot} presents logistic regression results for one-year infant mortality using the 2015 U.S.\ cohort. Estimates are shown for the original data, Syn-novel with $m = n = 500{,}000$ and $m = 50{,}000$, and Syn-naive. Unadjusted confidence intervals are based on estimated SEs. For Syn-naive, these SEs naively treat the synthetic data as real and typically underestimate true variability. Point estimates and SEs for all coefficients are provided in Table~\ref{tab:realAnalysis} in Appendix~\ref{appen:real}. Syn-novel closely matches the original MLE at both sample sizes, whereas Syn-naive exhibits clear bias, with some coefficients heavily attenuated or even reversed in sign.
	
	Substantively, several social and demographic factors show strong associations with one-year infant mortality. Lower parental education, unmarried status, and maternal Black race are each linked to elevated risk, even after adjusting for a broad set of clinical covariates. Paternal Black race is also associated with increased risk, though to a lesser extent, likely reflecting shared structural and environmental disadvantages. Participation in WIC and hospital birth appear protective. Delivery by a midwife is associated with lower mortality compared to doctor-attended births, likely reflecting case selection since midwives generally attend lower-risk deliveries. Clinical factors such as breech presentation, multiple births, chronic hypertension, diabetes, and eclampsia also show strong associations with mortality. These findings align with known disparities in maternal and infant health and underscore the importance of considering both social and clinical context in analyzing infant mortality.
	
	\begin{figure}[ht]
		\centering
		\includegraphics[width=0.8\textwidth]{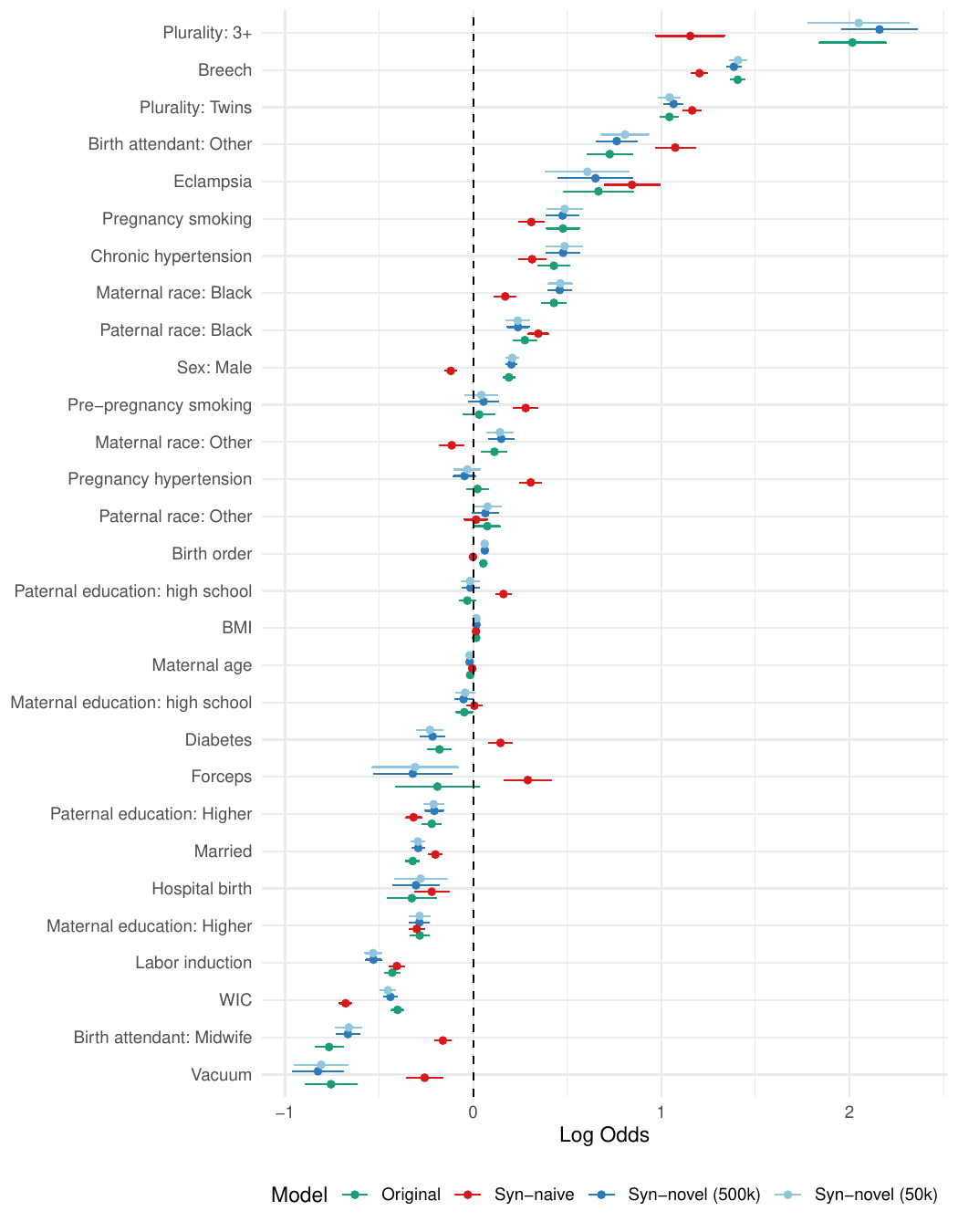}
		\caption{Forest plot comparing infant mortality logistic regression estimates across the different models. }
		\label{fig:forestplot}
	\end{figure}

	\section{Discussion} \label{sec:discussion}
	
	We introduced a novel framework for efficient inference using data generated from flexible generative models, leveraging additional summary statistics from the original dataset. When the number of observations $n$ greatly exceeds the number of covariates $p$, these statistics reveal limited information, keeping re-identification generally impossible. Augmenting the synthetic dataset with these summaries improves efficiency, allowing a smaller synthetic sample size, which reduces memory demands and lowers privacy risk.
	
	The proposed approach can also be applied to data compression beyond synthetic data generation. Instead of transmitting a full dataset of size $n$, one can transmit a subsample of size $m \ll n$ together with summary statistics from the full sample, enabling efficient inference for GLMs on the subsample. Since the subsample preserves the original data distribution, our theoretical results extend naturally to this setting.
	
	Our method currently assumes that analysis is conducted on the original variable scale defined by the data center and encoded in the Gram matrix. Future work should explore covariate transformations, interactions, extensions to additional models, or alternative summary statistics. For these directions, the data-fusion formulation in Eq.~(\ref{eq:dataFusion}) may provide a useful foundation. Additional directions include handling missing data in the original dataset during synthetic data generation and developing methods to account for the resulting variability in inference. Another avenue is to enhance privacy guarantees by adding noise to the released summary statistics to satisfy formal criteria such as differential privacy. Another limitation is that if the data center trains a single model for multiple users, it must cover a broad set of variables, which can reduce accuracy for researchers interested in only a subset.
	
	Finally, adapting our methodology to high-dimensional settings remains a significant challenge. The current framework assumes $n \gg p$, whereas many modern applications, such as genome-wide association studies, involve $p$ comparable to or exceeding $n$. In such settings, new theoretical developments would be required to establish valid inference, and practical concerns become more severe. The Gram matrix becomes large, making it computationally burdensome to share and potentially more revealing of sensitive information.

	\FloatBarrier
	\bibliographystyle{chicago}
	\bibliography{library}

	\clearpage
	\appendix
	\renewcommand{\thetable}{A\arabic{table}}
	\setcounter{table}{0}
	\renewcommand{\thefigure}{A\arabic{figure}}
	\setcounter{figure}{0}
	\renewcommand{\thealgorithm}{A\arabic{algorithm}}
	\setcounter{algorithm}{0}
	
	\section{Proofs} \label{appen:proofs}
	
	Proposition 1 concerns the asymptotic independence of the synthetic data, conditionally on $\mathcal{D}_n$ when applying the ``whiten-recolor" transformation. The proposition is derived for a general matrix $\bX$. 
	
	\begin{proposition} \label{prop:choleskyindep}
		Let $\bX\in\mathbb{R}^{n\times p}$ be a full-rank data matrix with i.i.d.\ rows, and let $\mathbf{R}_X$ be the upper-triangular Cholesky factor of $\bX^\top\bX$, so that $\bX^\top\bX = \mathbf{R}_X^\top \mathbf{R}_X$. Assume that $n^{-1}\bX^\top\bX$ converges in probability to $\boldsymbol{\Sigma}_X$. Then, as $n\to\infty$, the transformed observations $\mathbf{A} = \bX\,\mathbf{R}_X^{-1}$ are asymptotically independent, with the dependence vanishing at $O(1/n)$ rate.
	\end{proposition}

	\begin{proof}
		Since $n^{-1}\bX^\top\bX$ converges in probability to $\boldsymbol{\Sigma}_X$, we write
		\begin{equation} \label{eq:XtXSigmax}
			\bX^\top\bX=n\boldsymbol{\Sigma}_X+\Delta_X,\quad \Delta_X=O_P(\sqrt{n}).
		\end{equation}
		Recall that the Cholesky decomposition is a mapping $f\colon \mathcal{S}_{++}^{p}\to \mathcal{T}_p$ defined by $f(\mathbf{M})=\operatorname{chol}(\mathbf{M})$ where $\mathcal{S}_{++}^{p}$ is the set of $p\times p$ symmetric positive‐definite matrices and $\mathcal{T}_p$ the set of upper–triangular matrices with positive diagonal entries. By properties of Cholesky decomposition, we have
		
		$$
		\operatorname{chol}(n\boldsymbol{\Sigma}_X+\Delta_X)=\operatorname{chol}\left\{n\left(\boldsymbol{\Sigma}_X+\frac{1}{n}\Delta_X\right)\right\} = \sqrt{n}\operatorname{chol}\left(\boldsymbol{\Sigma}_X+\frac{1}{n}\Delta_X\right) \, ,
		$$
		and by Fréchet-differentiablity of the Cholesky mapping \cite[Lemma 1]{lutkepohl1989note} we obtain
		\begin{equation} \label{eq:cholTaylor}
			\operatorname{chol}\Bigl(\boldsymbol{\Sigma}_X+\frac{1}{n}\Delta_X\Bigr)
			=\operatorname{chol}(\boldsymbol{\Sigma}_X)+D\,\operatorname{chol}(\boldsymbol{\Sigma}_X)\Bigl[\frac{1}{n}\Delta_X\Bigr]+o_P\Bigl(\frac{1}{n}\|\Delta_X\|\Bigr)
			\, ,
		\end{equation}
		where $D$ is the derivative operator.
		Multiplying Eq.~\eqref{eq:cholTaylor} by $\sqrt{n}$ and plugging in Eq.~\eqref{eq:XtXSigmax} yields
		$$
		\operatorname{chol}(\bX^\top\bX)=\sqrt{n}\,\operatorname{chol}(\boldsymbol{\Sigma}_X)
		+\frac{1}{\sqrt{n}}\,D\,\operatorname{chol}(\boldsymbol{\Sigma}_X)[\Delta_X]+o_P\Bigl(\frac{1}{\sqrt{n}}\|\Delta_X\|\Bigr).
		$$
		
		Thus, denoting $\mathbf{R}_0=\operatorname{chol}(\boldsymbol{\Sigma}_X)$, the linear approximation becomes
		$$
		\mathbf{R}_X=\sqrt{n}\,\mathbf{R}_0+\mathbf{E}_n = \sqrt{n}\left(\mathbf{R}_0+\frac{1}{\sqrt{n}}\mathbf{E}_n\right),\quad\text{with }\mathbf{E}_n=\frac{1}{\sqrt{n}}\,D\,\operatorname{chol}(\boldsymbol{\Sigma}_X)[\Delta_X],
		$$
		
		Since $E_n=O_P(1)$, the perturbation $\frac{1}{\sqrt{n}}\mathbf{E}_n$ is small, and a first–order Taylor expansion for matrix inversion gives
		$$
		\mathbf{R}_X^{-1}=\frac{1}{\sqrt{n}}\mathbf{R}_0^{-1}-\frac{1}{n}\,\mathbf{R}_0^{-1}\mathbf{E}_n\,\mathbf{R}_0^{-1}+O_P\Bigl(\frac{1}{n^{3/2}}\Bigr),
		$$
		so that if we denote
		$$
		\Delta_n=-\frac{1}{n}\,\mathbf{R}_0^{-1}\mathbf{E}_n\,\mathbf{R}_0^{-1}+O_P\Bigl(\frac{1}{n^{3/2}}\Bigr),
		$$
		then
		$$
		\mathbf{R}_X^{-1}=\frac{1}{\sqrt{n}}\mathbf{R}_0^{-1}+\Delta_n,\quad \|\Delta_n\|=O_P\Bigl(\frac{1}{n}\Bigr).
		$$
		Defining the whitened data by
		$$
		\mathbf{A}=\bX\,\mathbf{R}_X^{-1}=\frac{1}{\sqrt{n}}\bX\,\mathbf{R}_0^{-1}+\bX\,\Delta_n,
		$$
		we note that the term $\frac{1}{\sqrt{n}}\bX\,\mathbf{R}_0^{-1}$ has i.i.d. rows (since $\bX$ has i.i.d. rows and $\mathbf{R}_0^{-1}$ is deterministic), whereas for a typical row the term $\Delta_n\bx$ is controlled by $\|\Delta_n\bx\|\le \|\bx\|\|\Delta_n\|$. Under standard moment conditions, $\|\bx\|$ grows at most like $O_P(1)$, so that overall the dependence of the transformed observations in $\mathbf{A}$ vanishes at $O_P(1/n)$ rate. 
	\end{proof}
	
	Proposition~\ref{prop:choleskyindep} holds in our case as a direct consequence of Assumption~\ref{as:regularityx}. Since the dependence diminishes at a rate of $O(1/n)$, asymptotic independence remains valid even after applying the $\sqrt{n}$ scaling. Therefore, in the following analysis, we will treat the synthetic observations as i.i.d., conditional on $\mathcal{D}_n$, without compromising the integrity of the theoretical results.
	
	Before proving Theorem~\ref{thm:normality}, let us present Lemma \ref{lemma:condConsis}, which establishes the consistency of the synthetic-data-based estimator towards the original data MLE $\hat{\bbeta}$, in the probability space conditional on $\mathcal{D}_n$.
	
	\begin{lemma} \label{lemma:condConsis}
		Conditionally on $\mathcal{D}_n$, and under Assumptions \ref{as:reularityParam},\ref{as:regularityx},\ref{as:biasNorm}, as $(m,n) \to \infty$, 	
		
		$$\|\tilde{\bbeta} - \hat{\bbeta}\| = o_{P|\mathcal{D}_n}(1) \, ,$$
		where $P|\mathcal{D}_n$ stands for the conditional probability measure given $\mathcal{D}_n$.
	\end{lemma}
	\begin{proof}[Proof of Lemma \ref{lemma:condConsis}]
		
		First, let us show that
		\begin{equation} \label{eq:uniformConvPsi}
			\sup_{\bbeta\in\mathcal{B}}\left\|\boldsymbol{\psi}_m(\bbeta)- \int \left(\mu(\bx^\top\bbeta) - \mu(\bx^\top\hat{\bbeta})\right)\bx d\hat{F}_n(\bx)\right\|\xrightarrow{P|\mathcal{D}_n}0 \, .
		\end{equation}
		Adding and subtracting $\int \bx^\top\hat{\btheta}\bx d\hat{F}_n(\bx)$, and recalling the definition of $\boldsymbol{\psi}_m$ from Eq.(\ref{eq:ScoreCorrected}), and by the triangle inequality, we get
		\begin{eqnarray*}
			& & \sup_{\bbeta\in\mathcal{B}}\left\|\boldsymbol{\psi}_m(\bbeta)- \int \left(\mu(\bx^\top\bbeta) - \mu(\bx^\top\hat{\bbeta})\right)\bx d\hat{F}_n(\bx)\right\| \\
			&\le& \sup_{\bbeta\in\mathcal{B}}\left\|\int \left(\mu(\bx^\top\bbeta) - \bx^\top\hat{\btheta}\right)\bx d\left(\tilde{F}_m(\bx)-\hat{F}_n(\bx)\right)\right\| \\
			&+& \left\|\int \left(\mu(\bx^\top\hat{\bbeta}) - \bx^\top\hat{\btheta}\right)\bx d\hat{F}_n(\bx)\right\| \, .
		\end{eqnarray*}
		
		By Assumption~\ref{as:biasNorm}, the second addend vanishes as $n\to\infty$. Recall that $\boldsymbol{\phi}(\bx, \bbeta) = (\mu(\bx^\top \bbeta) - \bx^\top \hat{\btheta})\bx$ and observe that the expression inside the norm takes the form  
		$$
		\frac{1}{m} \sum_{i=1}^m \left[\boldsymbol{\phi}(\tilde{\bx}_{n,i}, \bbeta) - \E_{\hat{F}_n} \left\{ \boldsymbol{\phi}(\bx, \bbeta) \right\} \right].
		$$  
		
		Due to Assumption~\ref{as:regularityx}(a) and by Chebyshev’s inequality, the expression converges to zero in conditional probability $P | \mathcal{D}_n$, demonstrating pointwise convergence. In order to establish uniform convergence, let us decompose the expression into three terms by introducing an $\epsilon$-net for $\mathcal{B}$. Since $\mathcal{B}$ is compact by Assumption~\ref{as:reularityParam}, for any $\epsilon > 0$, there exists a finite $\epsilon$-net $\{\bbeta_1, \dots, \bbeta_K\}$ such that for every $\bbeta \in \mathcal{B}$, there is some $\bbeta_k$ satisfying $\|\bbeta - \bbeta_k\| < \epsilon$. This allows us to write
		$$
		\sup_{\bbeta \in \mathcal{B}} \left\| \frac{1}{m} \sum_{i=1}^{m} \boldsymbol{\phi}(\tilde{\bx}_{n,i}, \bbeta) - \E_{\hat{F}_n}\left\{\boldsymbol{\phi}(\bx, \bbeta)\right\} \right\|
		\leq \sup_{\bbeta_k} \left\| \frac{1}{m} \sum_{i=1}^{m} \boldsymbol{\phi}(\tilde{\bx}_{n,i}, \bbeta_k) - \E_{\hat{F}_n}\left\{\boldsymbol{\phi}(\bx, \bbeta_k)\right\} \right\|
		$$
		$$
		+ \sup_{\bbeta \in \mathcal{B}} \left\| \frac{1}{m} \sum_{i=1}^{m} \left\{ \boldsymbol{\phi}(\tilde{\bx}_{n,i}, \bbeta) - \boldsymbol{\phi}(\tilde{\bx}_{n,i}, \bbeta_k) \right\} \right\|
		+ \sup_{\bbeta \in \mathcal{B}} \left\| \E_{\hat{F}_n}\left\{\boldsymbol{\phi}(\bx, \bbeta)\right\} - \E_{\hat{F}_n}\left\{\boldsymbol{\phi}(\bx, \bbeta_k)\right\} \right\|.
		$$
		
		We now show that each of these terms vanishes as $(m,n) \to \infty$, in conditional probability $P | \mathcal{D}_n$.
		First, the term involving the supremum over $\bbeta_k$ is finite, as there are only finitely many $\bbeta_k$. For each fixed $\bbeta_k$, by the pointwise convergence, we have
		$$
		\frac{1}{m} \sum_{i=1}^{m} \boldsymbol{\phi}(\tilde{\bx}_{n,i}, \bbeta_k) - \E_{\hat{F}_n}[\boldsymbol{\phi}(\bx, \bbeta_k)] \xrightarrow{P|\mathcal{D}_n} 0 \, .
		$$
		
		Next, consider the second term, which involves $\boldsymbol{\phi}(\tilde{\bx}_{n,i}, \bbeta) - \boldsymbol{\phi}(\tilde{\bx}_{n,i}, \bbeta_k)$. 
		For two parameters $\bbeta, \bbeta' \in \mathcal{B}$, by Assumption~\ref{as:regularityx}(b) we get
		$$
		\| \boldsymbol{\phi}(\bx, \bbeta) - \boldsymbol{\phi}(\bx, \bbeta') \| = |\mu(\bx^\top \bbeta) - \mu(\bx^\top \bbeta')| \|\bx \| \leq L(\bx)\|\bx\| \|\bbeta - \bbeta'\| \, .
		$$
		Hence, the second term is bounded by:
		$$
		\sup_{\bbeta \in \mathcal{B}} \frac{1}{m} \sum_{i=1}^{m} L(\tilde{\bx}_{n,i})\|\tilde{\bx}_{n,i}\| \|\bbeta - \bbeta_k\|.
		$$
		Since $\sup_{\bbeta \in \mathcal{B}} \|\bbeta - \bbeta_k\| < \epsilon$, we obtain:
		$$
		\sup_{\bbeta \in \mathcal{B}} \frac{1}{m} \sum_{i=1}^{m} L(\tilde{\bx}_{n,i})\|\tilde{\bx}_{n,i}\| \|\bbeta - \bbeta_k\| \leq \epsilon  \left[\frac{1}{m} \sum_{i=1}^{m} L(\tilde{\bx}_{n,i})\|\tilde{\bx}_{n,i}\| - \E_{\hat{F}_n}\{L(\bx)\|\tilde{\bx}_{n,i}\|\}\right] + \epsilon\E_{\hat{F}_n}\{L(\bx)\|\tilde{\bx}_{n,i}\|\}.
		$$
		Due to Assumption~\ref{as:regularityx}(b) Chebyshev's inequality implies that the first factor in the second term vanishes. Choosing $\epsilon$ sufficiently small ensures that the second factor also tends to zero. Thus, the entire second term vanishes in conditional probability.
		
		Finally, for the third term, we use the same bound as before
		$$
		\sup_{\bbeta \in \mathcal{B}} \left\| \E_{\hat{F}_n}\left\{\boldsymbol{\phi}(\bx, \bbeta)\right\} - \E_{\hat{F}_n}\left\{\boldsymbol{\phi}(\bx, \bbeta_k)\right\} \right\|
		\leq \E_{\hat{F}_n}\left\{L(\bx)\|\bx\|\right\} \sup_{\bbeta \in \mathcal{B}} \|\bbeta - \bbeta_k\|.
		$$
		Since $\sup_{\bbeta \in \mathcal{B}} \|\bbeta - \bbeta_k\| < \epsilon$ and $\E_{\hat{F}_n}\left\{L(\bx)\|\bx\|\right\} < \infty$, we obtain
		$$
		\sup_{\bbeta \in \mathcal{B}} \left\| \E_{\hat{F}_n}\left\{\boldsymbol{\phi}(\bx_i, \bbeta)\right\} - \E_{\hat{F}_n}\left\{\boldsymbol{\phi}(\bx_i, \bbeta_k)\right\} \right\| \leq \epsilon  \E_{\hat{F}_n}\left\{L(\bx)\|\bx\|\right\}.
		$$
		As $\epsilon \to 0$, this term vanishes in conditional probability.
		
		Since all three terms vanish, we conclude that:
		$$
		\sup_{\bbeta \in \mathcal{B}} \left\| \frac{1}{m} \sum_{i=1}^{m} \boldsymbol{\phi}(\tilde{\bx}_{n,i}, \bbeta) - \E_{\hat{F}_n}\left\{\boldsymbol{\phi}(\bx, \bbeta)\right\} \right\| \xrightarrow{P|\mathcal{D}_n} 0
		$$
		as $(m,n) \to \infty$, and we have thus established (\ref{eq:uniformConvPsi}).
		
		To invoke Theorem 5.9 of \cite{van2000asymptotic}, in addition to condition \eqref{eq:uniformConvPsi} we also must establish that $\hat{\bbeta}$ is the unique root of  
		$$
		\int \left\{\mu(\bx^\top\bbeta) - \mu(\bx^\top\hat{\bbeta})\right\} \bx \, d\hat{F}_n(\bx).
		$$
		By construction, $\hat{\bbeta}$ is a root. To show uniqueness, it suffices to prove that $\E_{\hat{F}_n} \left[\mu(\bx^\top\bbeta) \bx \right]$ is monotone in $\bbeta$. Since differentiation under expectation is justified here, we have  
		$$
		\frac{\partial}{\partial \bbeta} \E_{\hat{F}_n} \left\{\mu(\bx^\top\bbeta) \bx \right\} = \E_{\hat{F}_n} \left\{\mu'(\bx^\top\bbeta) \bx \bx^\top \right\} \, ,
		$$
		which is a definite matrix by Assumption~\ref{as:regularityx}(d), establishing monotonicity.
		Thus, we proved that $\tilde{\bbeta}$ converges to $\hat{\bbeta}$ in conditional probability $P|\mathcal{D}_n$ as $(m,n)\rightarrow\infty$.  
	\end{proof}
	
	In Lemma \ref{lemma:condNorm}, we establish the asymptotic normality of our estimator relative to the original-data MLE $\hat{\bbeta}$, conditional on $\{\hat{\bbeta},\hat{\btheta}\}$.
	
	\begin{lemma} \label{lemma:condNorm}
		Under Assumptions \ref{as:reularityParam}--\ref{as:infoMat}, as $(m,n)\rightarrow\infty$, 	
		$$\sqrt{n}(\tilde{\bbeta} - \hat{\bbeta}) \big| \hat{\bbeta}, \hat{\btheta} \xrightarrow{D} \mathcal{N}\left\{\mathbf{0},\boldsymbol{\Sigma}^{-1}(\alpha^{-1}\boldsymbol{\mathcal{A}}_1 + \boldsymbol{\mathcal{A}}_2)\boldsymbol{\Sigma}^{-1}\right\} \, .$$
	\end{lemma}
	
	\begin{proof}[Proof of Lemma \ref{lemma:condNorm}]
		
		From Lemma \ref{lemma:condConsis}, we know that the estimator $\tilde{\bbeta}$ is consistent for $\hat{\bbeta}$ conditionally on $\mathcal{D}_n$. We now demonstrate that this consistency also holds conditionally on $\{\hat{\bbeta}, \hat{\btheta}\}$. Specifically, we start by writing
		
		$$
		\lim_{(m,n)\rightarrow\infty}\Pr\left(\|\tilde{\bbeta} - \hat{\bbeta}\| > \epsilon \big| \hat{\bbeta}, \hat{\btheta} \right)
		=\lim_{(m,n)\rightarrow\infty}\E\Biggl\{\Pr\Bigl(\|\tilde{\bbeta} - \hat{\bbeta}\| > \epsilon \,\Big|\,\mathcal{D}_n\Bigr) \Big| \hat{\bbeta}, \hat{\btheta} \Biggr\}\, ,
		$$
		where we have used that $\sigma(\{\hat{\bbeta},\hat{\btheta}\})\subseteq\sigma(\mathcal{D}_n)$, where $\sigma(\cdot)$ denotes the generated sigma-algebra.
		Since the conditional probability is bounded (indeed, it is a random variable taking values in $[0,1]$), the dominated convergence theorem allows us to exchange the limit and the expectation. Thus, we have
		\begin{equation} \label{eq:ConsisUncond}
			\lim_{(m,n)\rightarrow\infty}\Pr\bigl(\|\tilde{\bbeta} - \hat{\bbeta}\| > \epsilon \big| \hat{\bbeta}, \hat{\btheta}\bigr)
			=\E\Biggl\{\lim_{(m,n)\rightarrow\infty}\Pr\Bigl(\|\tilde{\bbeta} - \hat{\bbeta}\| > \epsilon \,\Big|\,\mathcal{D}_n\Bigr)\Big| \hat{\bbeta}, \hat{\btheta}\Biggr\}=0\,,
		\end{equation}
		which establishes the consistency of $\tilde{\bbeta}$ for $\hat{\bbeta}$, conditionally on $\{\hat{\bbeta}, \hat{\btheta}\}$.
		
		Now, for each $\mu(\tilde{\bx}_{n,i}^\top\bbeta)$ we perform a Taylor expansion about $\hat{\bbeta}$:
		\begin{equation} \label{eq:PsiTaylor}
			0 = \boldsymbol{\psi}_m(\tilde{\bbeta}) = \boldsymbol{\psi}_m(\hat{\bbeta}) + \frac{1}{m}\sum_{i=1}^{m} \mu'(\tilde{\bx}_{n,i}^\top \hat{\bbeta})\, \tilde{\bx}_{n,i}\tilde{\bx}_{n,i}^\top (\tilde{\bbeta} - \hat{\bbeta}) + R_{n,m}\, ,
		\end{equation}
		
		where the remainder term is given by
		
		$$
		R_{n,m} = \frac{1}{m}\sum_{i=1}^{m} \frac{1}{2}\mu''(\tilde{\bx}_{n,i}^\top\boldsymbol{\xi}_{i})\left\{\tilde{\bx}_{n,i}^\top (\tilde{\bbeta} - \hat{\bbeta})\right\}^2\tilde{\bx}_{n,i}\, ,
		$$
		
		with each $\boldsymbol{\xi}_{i}$ lying between $\tilde{\bbeta}$ and $\hat{\bbeta}$.
		To bound the residual, note that
		$$\|R_{n,m}\| \le \frac{1}{2m}\sum_{i=1}^m \left|\mu''(\tilde{\bx}_{n,i}^\top\boldsymbol{\xi}_{i})\right| \left\{\tilde{\bx}_{n,i}^\top (\tilde{\bbeta} - \hat{\bbeta})\right\}^2\|\tilde{\bx}_{n,i}\|\, .$$
		
		Applying the Cauchy–Schwarz inequality, we obtain
		$$\left\{\tilde{\bx}_{n,i}^\top (\tilde{\bbeta} - \hat{\bbeta})\right\}^2\|\tilde{\bx}_{n,i}\| \le \|\tilde{\bx}_{n,i}\|^2\|\tilde{\bbeta} - \hat{\bbeta}\|^2\|\tilde{\bx}_{n,i}\| = \|\tilde{\bx}_{n,i}\|^3\|\tilde{\bbeta} - \hat{\bbeta}\|^2\, ,$$
		and therefore
		$$
		\|R_{n,m}\| \le \frac{\|\tilde{\bbeta} - \hat{\bbeta}\|^2}{2m}\sum_{i=1}^m \left|\mu''(\tilde{\bx}_{n,i}^\top\boldsymbol{\xi}_{i})\right||\tilde{\bx}_{n,i}\|^3 \, .
		$$ 
		Invoking Assumption~\ref{as:regularityx}(c), 
		and by Eq.~(\ref{eq:ConsisUncond}) we have   
		\begin{equation} \label{eq:RemainderBounded}
			\|R_{n,m}\| = O_{P|\hat{\bbeta}, \hat{\btheta}}\left(\|\tilde{\bbeta} - \hat{\bbeta}\|^2\right) = o_{P|\hat{\bbeta}, \hat{\btheta}}\left(\|\tilde{\bbeta} - \hat{\bbeta}\|\right) \, .
		\end{equation}
		
		Next, we wish to show that
		\begin{equation} \label{eq:PsiNorm}
			\sqrt{n}\,\boldsymbol{\psi}_m(\hat{\bbeta}) \big| \hat{\bbeta}, \hat{\btheta} \xrightarrow{D} \mathcal{N}\Big(\mathbf{0},\,\alpha^{-1}\boldsymbol{\mathcal{A}}_1 + \boldsymbol{\mathcal{A}}_2\Big) \, .
		\end{equation}
		To begin, we decompose $\boldsymbol{\psi}_m(\hat{\bbeta})$ by adding and subtracting its conditional expectation,
		$$
		\boldsymbol{\psi}_m(\hat{\bbeta}) = \underbrace{\Big\{\boldsymbol{\psi}_m(\hat{\bbeta}) - \E_{\hat{F}_n}\big(\boldsymbol{\psi}_m(\hat{\bbeta})\big)\Big\}}_{\text{Fluctuation term}} + \underbrace{\E_{\hat{F}_n}\big(\boldsymbol{\psi}_m(\hat{\bbeta})\big)}_{\text{Bias term}} \, .
		$$
		Here, the fluctuation term has mean zero by construction, while the bias term represents the systematic deviation away from zero due to the synthetic data distribution.
		
		Since the conditional distribution $\hat{F}_n$ depends on $n$, we invoke the Lindeberg–Feller central limit theorem for triangular arrays to study the fluctuation term. Define
		$$
		\mathbf{z}_{n,i} = \boldsymbol{\phi}(\tilde{\bx}_{n,i}, \hat{\bbeta}) - \E_{\hat{F}_n}\Big\{\boldsymbol{\phi}(\bx, \hat{\bbeta})\Big\}, \quad i=1,\ldots,m \, ,
		$$
		which, conditionally on $\mathcal{D}_n$, are i.i.d. with variance $\text{Var}(\mathbf{z}_{n,i} \mid \mathcal{D}_n) = \boldsymbol{\mathcal{A}}_1^{(n)}$, where, by Assumption~\ref{as:varMat}, $\boldsymbol{\mathcal{A}}_1^{(n)}$ converges to a positive definite matrix $\boldsymbol{\mathcal{A}}_1$.
		
		To apply the Lindeberg–Feller theorem, we must verify Lindeberg’s condition. Specifically, for any $\epsilon > 0$ we need to show that
		$$
		\E_{\hat{F}_n}\Big\{\|\mathbf{z}_{n,i}\|^2\, \mathbbm{1}\big(\|\mathbf{z}_{n,i}\| > \epsilon\sqrt{m}\big)\Big\} \le \frac{\E_{\hat{F}_n}\Big\{\|\mathbf{z}_{n,i}\|^{2+\delta}\Big\}}{(\epsilon\sqrt{m})^\delta} \rightarrow 0 \, .
		$$
		This inequality follows from a standard application of the Markov inequality (see, e.g., \citet[pg. 21]{van2000asymptotic}). Consequently, it is sufficient to show that the $(2+\delta)$'th moment,
		$$
		\E_{\hat{F}_n}\Big\{\|\mathbf{z}_{n,i}\|^{2+\delta}\Big\} \, ,
		$$
		is almost surely uniformly bounded. We first apply the inequality
		$$
		\|\mathbf{a} - \mathbf{b}\|^{2+\delta} \le 2^{1+\delta}\Big(\|\mathbf{a}\|^{2+\delta} + \|\mathbf{b}\|^{2+\delta}\Big),
		$$
		with $\mathbf{a} = \boldsymbol{\phi}(\tilde{\bx}_{n,i}, \hat{\bbeta})$ and $\mathbf{b} = \E_{\hat{F}_n}\{\boldsymbol{\phi}(\bx, \hat{\bbeta})\}$. This yields
		$$
		\|\mathbf{z}_{n,i}\|^{2+\delta} \le 2^{1+\delta}\Big(\|\boldsymbol{\phi}(\tilde{\bx}_{n,i}, \hat{\bbeta})\|^{2+\delta} + \Big\|\E_{\hat{F}_n}\{\boldsymbol{\phi}(\bx, \hat{\bbeta})\}\Big\|^{2+\delta}\Big).
		$$
		Jensen’s inequality implies that
		$$
		\Big\|\E_{\hat{F}_n}\{\boldsymbol{\phi}(\bx, \hat{\bbeta})\}\Big\|^{2+\delta} \le \E_{\hat{F}_n}\Big\{\|\boldsymbol{\phi}(\bx, \hat{\bbeta})\|^{2+\delta}\Big\} \, ,
		$$
		so we obtain
		$$
		\E_{\hat{F}_n}\Big\{\|\mathbf{z}_{n,i}\|^{2+\delta}\Big\} \le 2^{1+\delta}\Big( \E_{\hat{F}_n}\Big\{\|\boldsymbol{\phi}(\bx, \hat{\bbeta})\|^{2+\delta}\Big\} + \E_{\hat{F}_n}\Big\{\|\boldsymbol{\phi}(\bx, \hat{\bbeta})\|^{2+\delta}\Big\} \Big) = 2^{2+\delta}\, \E_{\hat{F}_n}\Big\{\|\boldsymbol{\phi}(\bx, \hat{\bbeta})\|^{2+\delta}\Big\}\,.
		$$
		which is uniformly bounded due to Assumption~\ref{as:regularityx}(a), so the Lindeberg condition is verified. For convenience, define
		$$
		\pi_{m,n}(\mathbf{t}) = \Pr\left( n^{1/2}\alpha^{1/2}\boldsymbol{\mathcal{A}}_1^{-1/2}\Big\{\boldsymbol{\psi}_m(\hat{\bbeta}) - \E_{\hat{F}_n}\big(\boldsymbol{\psi}_m(\hat{\bbeta})\big)\Big\} \le \mathbf{t} \,\Big|\, \mathcal{D}_n \right).
		$$
		Then, for any $\mathbf{t} \in \mathbb{R}^p$, we have
		$$
		\lim_{(m,n)\to\infty} \pi_{m,n}(\mathbf{t}) = \Phi(\mathbf{t}),
		$$
		where $\Phi(\cdot)$ denotes the multivariate normal distribution function. By the dominated convergence theorem, the conditional probability on $\{\hat{\bbeta}, \hat{\btheta}\}$ converges to the same limit:
		$$
		\E\big\{\pi_{m,n}(\mathbf{t})\big| \hat{\bbeta}, \hat{\btheta}\big\} \to \Phi(\mathbf{t}).
		$$
		
		Assumption~\ref{as:biasNorm} already guarantees the normality of the bias term conditionally on $\{\hat{\bbeta}, \hat{\btheta}\}$. It remains to show that this bias term is asymptotically independent of the fluctuation term, ensuring that their sum is normally distributed. Notably, upon conditioning on $\mathcal{D}_n$, the bias term becomes deterministic. Define
		\begin{align*}
			P^* \coloneq\, &\Pr\Bigg( n^{1/2}\alpha^{1/2}\boldsymbol{\mathcal{A}}_1^{-1/2}\Big\{\boldsymbol{\psi}_m(\hat{\bbeta}) - \E_{\hat{F}_n}\big(\boldsymbol{\psi}_m(\hat{\bbeta})\big)\Big\} \le \mathbf{t},\\
			&\quad\quad\quad n^{1/2}\boldsymbol{\mathcal{A}}_2^{-1/2}\E_{\hat{F}_n}\big(\boldsymbol{\psi}_m(\hat{\bbeta})\big) \le \mathbf{s}\Big| \hat{\bbeta}, \hat{\btheta} \Bigg).
		\end{align*}
		Conditioning on $\mathcal{D}_n$, we can write
		$$
		P^* = \E\Bigg[\pi_{m,n}(\mathbf{t})\,\mathbbm{1}\left\{ n^{1/2}\boldsymbol{\mathcal{A}}_2^{-1/2}\E_{\hat{F}_n}\big(\boldsymbol{\psi}_m(\hat{\bbeta})\big) \le \mathbf{s} \right\}\Big| \hat{\bbeta}, \hat{\btheta}\Bigg].
		$$
		Since $\E\big\{\pi_{m,n}(\mathbf{t})\mid \hat{\bbeta}, \hat{\btheta}\big\}$ converges almost surely to $\Phi(\mathbf{t})$, we decompose the expectation as
		
		$$
		\begin{aligned}
			P^* = \; & \E\Biggl[\Phi(\mathbf{t})\,\mathbbm{1}\Bigl\{ n^{1/2}\,\boldsymbol{\mathcal{A}}_2^{-1/2}\,\E_{\hat{F}_n}\bigl(\boldsymbol{\psi}_m(\hat{\bbeta})\bigr) \le \mathbf{s} \Bigr\}\,\Big|\, \hat{\bbeta}, \hat{\btheta}\Biggr] \\
			& {}+ \E\Biggl[\Bigl(\pi_{m,n}(\mathbf{t}) - \Phi(\mathbf{t})\Bigr)\,\mathbbm{1}\Bigl\{ n^{1/2}\,\boldsymbol{\mathcal{A}}_2^{-1/2}\,\E_{\hat{F}_n}\bigl(\boldsymbol{\psi}_m(\hat{\bbeta})\bigr) \le \mathbf{s} \Bigr\}\,\Big|\, \hat{\bbeta}, \hat{\btheta}\Biggr].
		\end{aligned}
		$$
		Thus,
		$$
		\Big| P^* - \Phi(\mathbf{t})\,\Pr\Big(n^{1/2}\boldsymbol{\mathcal{A}}_2^{-1/2}\E_{\hat{F}_n}\big(\boldsymbol{\psi}_m(\hat{\bbeta})\big) \le \mathbf{s}\big| \hat{\bbeta}, \hat{\btheta}\Big) \Big| \le \Big|\E\Big\{ \pi_{m,n}(\mathbf{t}) - \Phi(\mathbf{t})\big| \hat{\bbeta}, \hat{\btheta} \Big\}\Big|.
		$$
		Since the right-hand side vanishes as $(m,n)\to\infty$, it follows that
		$$
		P^* = \Phi(\mathbf{t})\,\Phi(\mathbf{s}) + o(1).
		$$
		This establishes the asymptotic independence between the fluctuation term and the bias term, implying Eq.(\ref{eq:PsiNorm}).
		
		Finally, combining Eq.'s (\ref{eq:PsiTaylor})–(\ref{eq:RemainderBounded}), we obtain
		$$
		\sqrt{n}\big(\tilde{\bbeta} - \hat{\bbeta}\big) = -\sqrt{n}\,\tilde{\boldsymbol{\cI}}^{-1}(\hat{\bbeta})\,\boldsymbol{\psi}_m(\hat{\bbeta}) + o_P(1) \, ,
		$$
		so Lemma \ref{lemma:condNorm} follows by invoking Assumption~\ref{as:infoMat}, Eq.(\ref{eq:PsiNorm}) and Slutsky’s theorem.
		
	\end{proof}
	
	Finally, we can prove Theorem \ref{thm:normality}.
	\begin{proof}[Proof of Theorem \ref{thm:normality}]
		We have
		$$
		\sqrt{n}\bigl(\tilde{\bbeta}-\bbeta^*\bigr)
		=\sqrt{n}\bigl(\tilde{\bbeta}-\hat{\bbeta}\bigr)
		+\sqrt{n}\bigl(\hat{\bbeta}-\bbeta^*\bigr).
		$$
		By Lemma~\ref{lemma:condNorm}, the first term is asymptotically normal conditionally on 
		$\{\hat{\bbeta},\hat{\btheta}\}$, while standard GLM theory implies that the second term is asymptotically normal with zero mean and asymptotic variance $\boldsymbol{\Sigma}^{-1}$. It remains to show that these two terms are asymptotically independent, noting that the second term is deterministic given $\hat{\bbeta}$.
		
		Denote
		$$
		\mathbf{A}=\boldsymbol{\Sigma}^{-1}\Bigl(\alpha^{-1}\boldsymbol{\mathcal{A}}_1
		+\boldsymbol{\mathcal{A}}_2\Bigr)\boldsymbol{\Sigma}^{-1}.
		$$
		Then,
		\begin{align*}
			\Pr\Bigl(\sqrt{n}&\bigl(\tilde{\bbeta}-\hat{\bbeta}\bigr)\le \mathbf{t},\,
			\sqrt{n}\bigl(\hat{\bbeta}-\bbeta^*\bigr)\le \mathbf{s}\Bigr)\\[1mm]
			&=\E\!\Biggl\{\Pr\Bigl(\sqrt{n}\bigl(\tilde{\bbeta}-\hat{\bbeta}\bigr)\le \mathbf{t},\,
			\sqrt{n}\bigl(\hat{\bbeta}-\bbeta^*\bigr)\le \mathbf{s}\,\Bigm|
			\hat{\bbeta},\hat{\btheta}\Bigr)\Biggr\}\\[1mm]
			&=\E\!\Biggl[\Pr\Bigl(\sqrt{n}\bigl(\tilde{\bbeta}-\hat{\bbeta}\bigr)\le \mathbf{t}\,\Bigm|
			\hat{\bbeta},\hat{\btheta}\Bigr)
			\mathbbm{1}\!\Bigl\{\sqrt{n}\bigl(\hat{\bbeta}-\bbeta^*\bigr)\le \mathbf{s}\Bigr\}\Biggr]\\[1mm]
			&=\E\!\Biggl[\Phi\Bigl(\mathbf{A}^{-1/2}\mathbf{t}\Bigr)
			\mathbbm{1}\!\Bigl\{\sqrt{n}\bigl(\hat{\bbeta}-\bbeta^*\bigr)\le \mathbf{s}\Bigr\}\Biggr]\\[1mm]
			&\quad+\E\!\Biggl[\Bigl\{\Pr\Bigl(\sqrt{n}\bigl(\tilde{\bbeta}-\hat{\bbeta}\bigr)
			\le \mathbf{t}\,\Bigm|\,\hat{\bbeta},\hat{\btheta}\Bigr)
			-\Phi\Bigl(\mathbf{A}^{-1/2}\mathbf{t}\Bigr)\Bigr\}\mathbbm{1}\!\Bigl\{\sqrt{n}\bigl(\hat{\bbeta}-\bbeta^*\bigr)
			\le \mathbf{s}\Bigr\}\Biggr].
		\end{align*}
		Thus,
		\begin{align*}
			\Bigl|\Pr\Bigl(&\sqrt{n}\bigl(\tilde{\bbeta}-\hat{\bbeta}\bigr)\le \mathbf{t},\,
			\sqrt{n}\bigl(\hat{\bbeta}-\bbeta^*\bigr)\le \mathbf{s}\Bigr)
			-\E\!\Bigl[\Phi\Bigl(\mathbf{A}^{-1/2}\mathbf{t}\Bigr)
			\mathbbm{1}\!\Bigl\{\sqrt{n}\bigl(\hat{\bbeta}-\bbeta^*\bigr)
			\le \mathbf{s}\Bigr\}\Bigr]\Bigr|\\[1mm]
			&\le \Biggl|\E\!\Bigl[\Pr\Bigl(\sqrt{n}\bigl(\tilde{\bbeta}-\hat{\bbeta}\bigr)
			\le \mathbf{t}\,\Bigm|\,\hat{\bbeta},\hat{\btheta}\Bigr)
			-\Phi\Bigl(\mathbf{A}^{-1/2}\mathbf{t}\Bigr)\Bigr]\Biggr|.
		\end{align*}
		
		Since the right-hand side vanishes as $(m,n)\to\infty$, we obtain
		$$
		\Pr\Bigl(\sqrt{n}\bigl(\tilde{\bbeta}-\hat{\bbeta}\bigr)\le \mathbf{t},\,
		\sqrt{n}\bigl(\hat{\bbeta}-\bbeta^*\bigr)\le \mathbf{s}\Bigr)
		\longrightarrow \Phi\Bigl(\mathbf{A}^{-1/2}\mathbf{t}\Bigr)
		\Phi\Bigl(\boldsymbol{\Sigma}^{1/2}\mathbf{s}\Bigr),
		$$
		which establishes the asymptotic independence and completes the proof.
	\end{proof}
	
	\section{Technical Conditions for Common $\mu$ Functions} \label{appen:mu}
	
	In this section we provide sufficient conditions for Assumption~\ref{as:regularityx} to hold for some commonly used $\mu$ functions. 
	
	First, we note that in conjunction with Assumption~\ref{as:regularityx}(e), a sufficient condition for Assumption~\ref{as:regularityx}(d) to hold is that $\mu$ is strictly monotone. Let $\mathbf{v}$ be a non-zero vector, then  
	$$
	\mathbf{v}^\top \E_{\hat{F}_n} \left\{\mu'(\bx^\top\bbeta) \bx \bx^\top \right\} \mathbf{v} = \E_{\hat{F}_n} \left\{\mathbf{v}^\top \mu'(\bx^\top\bbeta) \bx \bx^\top \mathbf{v} \right\} = \E_{\hat{F}_n} \left\{\mu'(\bx^\top\bbeta) (v^\top\bx)^2 \right\}.
	$$
	Assume without loss of generality that $\mu$ is strictly increasing, then $\mu'(\bx^\top\bbeta) > 0$, and the expression is strictly positive unless $\mathbf{v}^\top \bx = 0$ almost surely. However, that would imply $\mathbf{v}^\top \E_{\hat{F}_n}(\bx \bx^\top) \mathbf{v} = 0$ almost surely, contradicting Assumption~\ref{as:regularityx}(e). All $\mu$ functions discussed below are strictly monotone, so we need not further discuss Assumption~\ref{as:regularityx}(d).
	
	\subsection{$\mu(x) = \exp(x)$}
	
	This corresponds to the log-link function, which is the canonical link for Poisson regression. It is also frequently used as a non-canonical link in Gamma regression, including exponential regression as a special case. Since the function $\exp(x)$ is unbounded, and grows exponentially, we need a strong tail control for the distribution of $\tilde{\bx}_n$ as given in the following lemma, where we assume a sub-exponential distribution.
	
	\begin{lemma} \label{lemma:sufficient_cond_exp(x)}
		Assumption~\ref{as:regularityx}(a)--(c) holds if there exist constants $\rho_1 > 0$ and $\rho_2 > (2 + \delta)M$ (where $M$ is defined in Assumption~\ref{as:reularityParam}), and an index $n^*$ such that for all $n > n^*$, the following holds almost surely:
		$$
		\Pr(\|\tilde{\bx}_{n}\| \ge r | \mathcal{D}_n) \le \rho_1 e^{-\rho_2 r}.
		$$
	\end{lemma}
	\begin{proof}[Proof of Lemma~\ref{lemma:sufficient_cond_exp(x)}]
		We begin by verifying part (a) of Assumption~\ref{as:regularityx}. Since
		$\boldsymbol{\phi}(\bx, \bbeta) = \left\{ \bx^\top \hat{\btheta} - \mu(\bx^\top \bbeta) \right\} \bx$,
		we apply the inequality $(a + b)^r \le 2^{r - 1}(a^r + b^r)$ for $r = 2 + \delta$ to get:
		\begin{align*}
			\|\boldsymbol{\phi}(\bx, \bbeta)\|^{2+\delta}
			&= \left| \bx^\top \hat{\btheta} - \mu(\bx^\top \bbeta) \right|^{2+\delta}  \|\bx\|^{2+\delta} \\
			&\le 2^{1+\delta} \left( |\bx^\top \hat{\btheta}|^{2+\delta} + |\mu(\bx^\top \bbeta)|^{2+\delta} \right)  \|\bx\|^{2+\delta}.
		\end{align*}
		
		Assumption~\ref{as:reularityParam} with Cauchy-Schwarz inequality imply that $|\mu(\bx^\top\bbeta)| \le \exp(M\|\bx\|)$, so
		$$
		|\mu(\bx^\top\bbeta)|^{2+\delta}\|\bx\|^{2+\delta} \le e^{(2+\delta)M\|\bx\|} \|\bx\|^{2+\delta}.
		$$
		We now upper bound
		$$
		\E_{\hat{F}_n}\left\{e^{(2+\delta)M\|\bx\|} \|\bx\|^{2+\delta}\right\}.
		$$
		Define the function $f(r) := e^{(2+\delta)M r} r^{2+\delta}$. Since $f$ is differentiable with $f(0) = 0$ and nonnegative on $[0,\infty)$, we may use integration-by-parts to produce the identity
		$$
		\E\left\{f(X)\right\} = \int_0^\infty f'(r) \, \Pr(X > r) \, dr
		$$
		to obtain
		$$
		\E_{\hat{F}_n}\left\{f(\|\bx\|)\right\}
		= \int_0^\infty f'(r) \, \Pr\left({\|\tilde{\bx}_n\| > r} \mid \mathcal{D}_n\right) \, dr.
		$$
		By assumption, for all $n > n^*$,
		$$
		\Pr\left({\|\tilde{\bx}_n\| > r} \mid \mathcal{D}_n\right) \le \rho_1 e^{-\rho_2 r},
		$$
		so
		$$
		\E_{\hat{F}_n}\left(e^{(2+\delta)M\|\bx\|} \|\bx\|^{2+\delta}\right)
		\le \rho_1 \int_0^\infty f'(r) e^{-\rho_2 r} \, dr.
		$$
		
		We now compute $f'(r)$:
		$$
		f'(r) = \frac{d}{dr} \left\{e^{(2+\delta)M r} r^{2+\delta}\right\}
		= e^{(2+\delta)M r} \left\{(2+\delta)M r^{2+\delta} + (2+\delta) r^{1+\delta}\right\}.
		$$
		Therefore,
		$$
		f'(r) e^{-\rho_2 r}
		= e^{-\epsilon r} \left\{(2+\delta)M r^{2+\delta} + (2+\delta) r^{1+\delta}\right\},
		\quad \text{where } \epsilon := \rho_2 - (2+\delta)M > 0.
		$$
		
		This gives the bound
		$$
		\E_{\hat{F}_n}\left(e^{(2+\delta)M\|\bx\|} \|\bx\|^{2+\delta}\right)
		\le \rho_1 \left\{(2+\delta)M \int_0^\infty r^{2+\delta} e^{-\epsilon r} dr
		+ (2+\delta) \int_0^\infty r^{1+\delta} e^{-\epsilon r} dr\right\}.
		$$
		Using the identity $\int_0^\infty r^k e^{-\epsilon r} dr = \Gamma(k+1)/\epsilon^{k+1}$, this becomes
		$$
		\E_{\hat{F}_n}\left(e^{(2+\delta)M\|\bx\|} \|\bx\|^{2+\delta}\right)
		\le \rho_1 \left\{(2+\delta)M  \frac{\Gamma(3+\delta)}{\epsilon^{3+\delta}} + (2+\delta)  \frac{\Gamma(2+\delta)}{\epsilon^{2+\delta}}\right\} := C_1'.
		$$
		
		Next, we bound the second term $\|\bx^\top \hat{\btheta} \bx\|^{2+\delta}$. Since $\|\hat{\btheta}\| \le M$ by Assumption~\ref{as:reularityParam}, we have
		$$
		\|\bx^\top \hat{\btheta} \bx\|^{2+\delta} \le \|\hat{\btheta}\|^{2+\delta} \|\bx\|^{2(2+\delta)} \le M^{2+\delta} \|\bx\|^{2(2+\delta)}.
		$$
		Hence,
		$$
		\E_{\hat{F}_n}\left(\|\bx^\top \hat{\btheta} \bx\|^{2+\delta}\right)
		\le M^{2+\delta}  \E_{\hat{F}_n}\left\{\|\bx\|^{2(2+\delta)}\right\}.
		$$
		
		To bound the remaining expectation, define $g(r) := r^{2(2+\delta)}$. Then,
		$$
		\E_{\hat{F}_n}\left\{g(\|\bx\|)\right\}
		= \int_0^\infty g'(r) \, \Pr\left({\|\tilde{\bx}_n\| > r} \mid \mathcal{D}_n\right) \, dr
		\le \rho_1 \int_0^\infty g'(r) e^{-\rho_2 r} dr.
		$$
		We compute $g'(r) = 2(2+\delta) r^{2(2+\delta)-1}$, so
		$$
		\E_{\hat{F}_n}\left\{\|\bx\|^{2(2+\delta)}\right\}
		\le \rho_1  2(2+\delta) \int_0^\infty r^{2(2+\delta)-1} e^{-\rho_2 r} dr
		= \frac{2(2+\delta)\rho_1 \Gamma\left({2(2+\delta)}\right)}{\rho_2^{2(2+\delta)}} := C_1''.
		$$
		
		Putting both bounds together, we obtain
		$$
		\E_{\hat{F}_n} \left\{\|\phi(\bx, \bbeta)\|^{2+\delta}\right\}
		\le C_1' + M^{2+\delta} C_1'' := C_1,
		$$
		which holds almost surely for all $n > n^*$, and thus verifies Assumption~\ref{as:regularityx}(a).
		We aim now to verify Assumption~\ref{as:regularityx}(b), and to that end
		we apply the mean value theorem to the function $\mu(t)$ with $t_1 := \bx^\top \bbeta_1$ and $t_2 := \bx^\top \bbeta_2$. Then there exists $\xi$ between $t_1$ and $t_2$ such that
		$$
		|\mu(t_1) - \mu(t_2)| = |\mu'(\xi)| \cdot |t_1 - t_2| = e^{\xi}  |\bx^\top(\bbeta_1 - \bbeta_2)|.
		$$
		By the Cauchy--Schwarz inequality,
		$$
		|\bx^\top(\bbeta_1 - \bbeta_2)| \le \|\bx\| \cdot \|\bbeta_1 - \bbeta_2\|.
		$$
		Moreover, since $\bbeta_1, \bbeta_2 \in \mathcal{B}$ and $\|\bbeta\| \le M$ for all $\bbeta \in \mathcal{B}$ by Assumption~\ref{as:reularityParam}, we have
		$$
		|\xi| \le \max\left({|\bx^\top \bbeta_1|, |\bx^\top \bbeta_2|}\right) \le M \|\bx\|.
		$$
		Hence,
		$$
		|\mu(\bx^\top \bbeta_1) - \mu(\bx^\top \bbeta_2)| \le e^{M\|\bx\|}  \|\bx\| \cdot \|\bbeta_1 - \bbeta_2\|.
		$$
		We identify the Lipschitz constant as
		$$
		L(\bx) := e^{M\|\bx\|}  \|\bx\|.
		$$
		To complete the proof, we verify that
		$$
		\E_{\hat{F}_n}\left\{L(\bx)^2  \|\bx\|^2\right\} = \E_{\hat{F}_n}\left({e^{2M\|\bx\|}  \|\bx\|^4}\right) < C_2,
		$$
		almost surely for $n > n^*$. As in part (a), we use the tail bound assumption
		$$
		\Pr\left({\|\tilde{\bx}_n\| > r} \mid \mathcal{D}_n \right) \le \rho_1 e^{-\rho_2 r},
		$$
		with $\rho_2 > (2+\delta)M > 2M$. Then applying the tail integral identity:
		$$
		\E_{\hat{F}_n}\left({e^{2M\|\bx\|}  \|\bx\|^4}\right)
		= \int_0^\infty f'(r) \Pr\left({\|\tilde{\bx}_n\| > r} \mid \mathcal{D}_n\right) dr,
		$$
		where $f(r) := e^{2M r} r^4$. Then $f'(r) = e^{2M r} \left({4r^3 + 2M r^4}\right)$, so
		$$
		\E_{\hat{F}_n}\left({e^{2M\|\bx\|}  \|\bx\|^4}\right)
		\le \rho_1 \int_0^\infty \left({4r^3 + 2M r^4}\right) e^{-(\rho_2 - 2M)r} dr.
		$$
		We compute
		$$
		\int_0^\infty r^3 e^{-\epsilon r} dr = \frac{6}{\epsilon^4}, \quad
		\int_0^\infty r^4 e^{-\epsilon r} dr = \frac{24}{\epsilon^5}, \quad \text{with } \epsilon := \rho_2 - 2M > 0,
		$$
		so
		$$
		\E_{\hat{F}_n}\left({e^{2M\|\bx\|}  \|\bx\|^4}\right)
		\le \rho_1 \left({\frac{24}{\epsilon^4} + \frac{48M}{\epsilon^5}}\right) := C_2.
		$$
		
		Thus, the moment condition is satisfied with explicit constant
		$$
		C_2 = \rho_1 \left\{{\frac{24}{\left({\rho_2 - 2M}\right)^4} + \frac{48M}{\left({\rho_2 - 2M}\right)^5}}\right\}.
		$$
		This completes the verification of Assumption~\ref{as:regularityx}(b).
		
		We now address part (c). Let $\bbeta_1, \dots, \bbeta_m \in \mathcal{B}$. Since $\mu''(t) = e^t$ and $\|\bbeta_i\| \le M$, define
		$$
		\left| \mu''\left( \tilde{\bx}_{n,i}^\top \bbeta_i \right) \right| \cdot \|\tilde{\bx}_{n,i}\|^3
		\le e^{M \|\tilde{\bx}_{n,i}\|}  \|\tilde{\bx}_{n,i}\|^3 := Z_{n,i}.
		$$
		
		To show
		$$
		\frac{1}{m} \sum_{i=1}^m Z_{n,i} = O_{P \mid \mathcal{D}_n}(1),
		$$
		it suffices to prove that $\E_{\hat{F}_n}\left( Z_{n,i} \right) \le C_3 < \infty$ almost surely, and apply Markov’s inequality.
		
		Let $f(r) = e^{M r} r^3$, so $Z_{n,i} = f\left( \|\tilde{\bx}_{n,i}\| \right)$. Then
		$$
		f'(r) = \frac{d}{dr} \left( e^{M r} r^3 \right) = e^{M r} \left( 3 r^2 + M r^3 \right),
		$$
		and by the tail bound assumption,
		$$
		\E_{\hat{F}_n}\left( Z_{n,i} \right)
		\le \rho_1 \int_0^\infty f'(r) e^{-\rho_2 r} \, dr
		= \rho_1 \int_0^\infty \left( 3 r^2 + M r^3 \right) e^{-(\rho_2 - M) r} \, dr.
		$$
		
		Let $\epsilon := \rho_2 - M > 0$. Then
		$$
		\int_0^\infty r^2 e^{-\epsilon r} \, dr = \frac{2}{\epsilon^3}, \qquad
		\int_0^\infty r^3 e^{-\epsilon r} \, dr = \frac{6}{\epsilon^4},
		$$
		so
		$$
		\E_{\hat{F}_n}\left( Z_{n,i} \right)
		\le \rho_1 \left( 3  \frac{2}{\epsilon^3} + M  \frac{6}{\epsilon^4} \right)
		= \rho_1 \left( \frac{6}{(\rho_2 - M)^3} + \frac{6M}{(\rho_2 - M)^4} \right) := C_3.
		$$
		
		Now let $\varepsilon > 0$ be arbitrary. By Markov’s inequality, for any $\gamma > 0$,
		$$
		\Pr \left( \frac{1}{m} \sum_{i=1}^m Z_{n,i} > \gamma  \middle| \, \mathcal{D}_n  \right) 
		\le \frac{1}{\gamma}  \E_{\hat{F}_n} \left( \frac{1}{m} \sum_{i=1}^m Z_{n,i} \right)
		\le \frac{C_3}{\gamma}.
		$$
		
		Choosing $\gamma := C_3 / \varepsilon$ gives
		$$
		\Pr \left( \frac{1}{m} \sum_{i=1}^m Z_{n,i} > \gamma  \, \middle| \, \mathcal{D}_n \right) \le \varepsilon,
		$$
		which shows
		$$
		\frac{1}{m} \sum_{i=1}^m \left| \mu''\left( \tilde{\bx}_{n,i}^\top \bbeta_i \right) \right| \cdot \|\tilde{\bx}_{n,i}\|^3 = O_{P \mid \mathcal{D}_n}(1).
		$$

	\end{proof}

	\subsection{$\mu$ is an absolutely continuous and bounded function with bounded first and second derivatives}
	
	This category contains link functions commonly used in binary regression, such as the canonical logistic quantile function, the normal quantile function (Probit regression), and the Cauchy quantile function. Since $\mu$ is bounded, the required conditions are relatively mild -- a uniformly bounded fourth moment.

	\begin{lemma} \label{lemma:sufficient_cond_cdf}
		Assumption~\ref{as:regularityx}(a)--(c) holds if there exists $\delta > 0$, a constant $\rho < \infty$, and an index $n^*$ such that for all $n > n^*$,
		$$
		\E_{\hat{F}_n} \left( \|\tilde{\bx}_n\|^{4+\delta} \right) \le \rho
		\quad \text{almost surely}.
		$$
	\end{lemma}
	
	\begin{proof}
		We begin by verifying part (a) of Assumption~\ref{as:regularityx}.  
		Since $\mu$ is bounded, we have $|\mu(t)| \le C_\mu$ for all $t \in \mathbb{R}$. Moreover, Assumption~\ref{as:reularityParam} implies that $\|\hat{\btheta}\| \le M$, yielding
		$$
		|\bx^\top \hat{\btheta}| \le M \|\bx\|, \qquad |\mu(\bx^\top \bbeta)| \le C_\mu.
		$$
		
		Applying the inequality $(a + b)^r \le 2^{r-1}(a^r + b^r)$ with $r = 2+\delta$, we obtain
		$$
		\left| \bx^\top \hat{\btheta} - \mu(\bx^\top \bbeta) \right|^{2+\delta}
		\le 2^{1+\delta} \left( M^{2+\delta} \|\bx\|^{2+\delta} + C^{2+\delta}_\mu \right).
		$$
		Therefore,
		$$
		\|\boldsymbol{\phi}(\bx, \bbeta)\|^{2+\delta}
		= \left| \bx^\top \hat{\btheta} - \mu(\bx^\top \bbeta) \right|^{2+\delta}  \|\bx\|^{2+\delta}
		\le 2^{1+\delta} \left( M^{2+\delta} \|\bx\|^{4+\delta} + C^{2+\delta}_\mu\|\bx\|^{2+\delta} \right).
		$$
		Since $\|\bx\|^{2+\delta} \le 1 + \|\bx\|^{4+\delta}$ for all $\bx$, we have
		$$
		\|\boldsymbol{\phi}(\bx, \bbeta)\|^{2+\delta}
		\le 2^{1+\delta} \left\{ \left(M^{2+\delta} + C^{2+\delta}_\mu\right)\|\bx\|^{4+\delta} + C^{2+\delta}_\mu \right\}.
		$$
		Taking expectation under $\hat{F}_n$ and using the assumption,
		$$
		\E_{\hat{F}_n} \left\{ \|\boldsymbol{\phi}(\bx, \bbeta)\|^{2+\delta} \right\}
		\le 2^{1+\delta} \left(M^{2+\delta} + C^{2+\delta}_\mu\right) \E_{\hat{F}_n} \|\bx\|^{4+\delta} + 2^{1+\delta}C^{2+\delta}_\mu
		\le C_1,
		$$
		where $C_1 := 2^{1+\delta} \left\{(M^{2+\delta} + C^{2+\delta}_\mu)\rho + C^{2+\delta}_\mu\right\} < \infty$. This verifies Assumption~\ref{as:regularityx}(a).
		
		We now verify part (b) of Assumption~\ref{as:regularityx}. Fix $\bbeta_1, \bbeta_2 \in \mathcal{B}$ and define $t_1 := \bx^\top \bbeta_1$, $t_2 := \bx^\top \bbeta_2$. By the mean value theorem, there exists $\xi$ between $t_1$ and $t_2$ such that
		$$
		\mu(t_1) - \mu(t_2) = \mu'(\xi)(t_1 - t_2) = \mu'(\xi)  \bx^\top(\bbeta_1 - \bbeta_2).
		$$
		Taking absolute values and applying the Cauchy--Schwarz inequality,
		$$
		|\mu(\bx^\top \bbeta_1) - \mu(\bx^\top \bbeta_2)|
		\le |\mu'(\xi)| \cdot \|\bx\| \cdot \|\bbeta_1 - \bbeta_2\|.
		$$
		Since $\mu'$ is bounded, say $|\mu'(t)| \le C_{\mu'}$ for all $t \in \mathbb{R}$, we obtain
		$$
		|\mu(\bx^\top \bbeta_1) - \mu(\bx^\top \bbeta_2)|
		\le C_{\mu'}  \|\bx\| \cdot \|\bbeta_1 - \bbeta_2\|.
		$$
		Thus, the Lipschitz constant can be taken as
		$$
		L(\bx) := C_{\mu'}  \|\bx\| \, ,
		$$
		and we obtain that
		$$
		\E_{\hat{F}_n} \left\{{L(\bx)^2  \|\bx\|^2}\right\}
		= C_{\mu'}^2  \E_{\hat{F}_n} \left( \|\bx\|^4 \right) < C_2,
		$$
		where $C_2 = C_{\mu'}^2  \rho$, almost surely for all $n > n^*$, verifying part (b).
		
		We now verify part (c) of Assumption~\ref{as:regularityx}. Let $\bbeta_1, \dots, \bbeta_m \in \mathcal{B}$, and define the quantity
		$$
		Z_{n,i} := \left| \mu''(\tilde{\bx}_{n,i}^\top \bbeta_i) \right| \cdot \|\tilde{\bx}_{n,i}\|^3.
		$$
		Since $\mu''$ is bounded by assumption, say $|\mu''(t)| \le C_{\mu''}$ for all $t \in \mathbb{R}$, we have
		$$
		Z_{n,i} \le C_{\mu''} \|\tilde{\bx}_{n,i}\|^3.
		$$
		To establish the condition
		$$
		\frac{1}{m} \sum_{i=1}^m Z_{n,i} = O_{P \mid \mathcal{D}_n}(1),
		$$
		it suffices to show that $\E_{\hat{F}_n}(Z_{n,i}) \le C_3 < \infty$ for all $n > n^*$, since then we may apply Markov's inequality.
		We bound:
		$$
		\E_{\hat{F}_n} (Z_{n,i})
		\le C_{\mu''} \E_{\hat{F}_n} \left( \|\bx\|^3 \right)
		\le C_{\mu''} \E_{\hat{F}_n} \left( \|\bx	\|^4 \right)
		\le C_{\mu''} \rho := C_3.
		$$

		Applying Markov's inequality to the average,
		$$
		\Pr\left({\frac{1}{m} \sum_{i=1}^m Z_{n,i} > \gamma \,\Big|\, \mathcal{D}_n}\right)
		\le \frac{1}{\gamma} \E_{\hat{F}_n} \left( \frac{1}{m} \sum_{i=1}^m Z_{n,i} \right)
		\le \frac{C_{\mu''} \rho}{\gamma}.
		$$
		
		Now set $\gamma := C_{\mu''} \rho / \varepsilon$ to get:
		$$
		\Pr\left({\frac{1}{m} \sum_{i=1}^m Z_{n,i} > \gamma \,\Big|\, \mathcal{D}_n}\right) \le \varepsilon,
		$$
		which shows
		$$
		\frac{1}{m} \sum_{i=1}^m \left| \mu''(\tilde{\bx}_{n,i}^\top \bbeta_i) \right| \cdot \|\tilde{\bx}_{n,i}\|^3 = O_{P \mid \mathcal{D}_n}(1).
		$$	
	\end{proof}
	
	\subsection{$\mu(x) = 1/x$}
	
	This corresponds to the canonical link function for gamma (and exponential) regression. Despite being canonical, it is not commonly used in practice due to the constraint $\bx^\top\bbeta > 0$, which ensures the distribution’s parameter remains within its domain (i.e., strictly positive). For the proof below, we will assume a slightly stronger condition: there exists a constant $b > 0$ such that $\bx^\top\bbeta > b$.

	\begin{lemma}\label{lemma:sufficient_cond_inv}
		Assumption~\ref{as:regularityx}(a)--(c) holds if there exist constants $b > 0$, $\delta > 0$, $\rho<\infty$ and an index $n^*$ such that for all $n > n^*$, almost surely,
		$$
		\inf_{\bbeta\in\mathcal{B}} \left| \tilde{\bx}_n^{\!\top} \bbeta \right| \ge b \quad \text{and} \quad \E_{\hat{F}_n}\|\bx\|^{4+2\delta} < \rho.
		$$
	\end{lemma}
	
	\begin{proof}
		
		Note that
		$$
		\boldsymbol{\phi}(\bx,\bbeta)=\left( \bx^{\!\top}\hat{\btheta}- \frac{1}{\bx^{\!\top}\bbeta} \right) \bx.
		$$
		Since $|\tilde{\bx}_n^{\top}\hat{\btheta}|\le M\|\tilde{\bx}_n\|$ and $|\tilde{\bx}_n^{\top}\bbeta|\ge b$, we have
		$$
		\|\boldsymbol{\phi}(\tilde{\bx}_n,\bbeta)\|\le|\tilde{\bx}_n^{\top}\hat{\btheta}|\cdot\|\tilde{\bx}_n\|+\frac{\|\tilde{\bx}_n\|}{b}\le M\|\tilde{\bx}_n\|^2+\frac{\|\tilde{\bx}_n\|}{b}.
		$$
		Hence
		$$
		\|\boldsymbol{\phi}(\tilde{\bx}_n,\bbeta)\|^{2+\delta}\le2^{1+\delta}\!\left(M^{2+\delta}\|\tilde{\bx}_n\|^{4+2\delta}+b^{-(2+\delta)}\|\tilde{\bx}_n\|^{2+\delta}\right),
		$$
		Taking expectations and applying Jensen's inequality to the second term,
		$$
		\E_{\hat{F}_n}\|\bx\|^{2+\delta} \le \left(\E_{\hat{F}_n}\|\bx\|^{4+2\delta}\right)^{(2+\delta)/(4+2\delta)} \le \rho^{(2+\delta)/(4+2\delta)} \, ,
		$$
		giving the bound
		$$
		C_1 = 2^{1+\delta} \left( M^{2+\delta} \rho + b^{-(2+\delta)} \rho^{(2+\delta)/(4+2\delta)} \right)
		.
		$$
		
		Now, for part (b), note that
		$$
		|\mu(\tilde{\bx}_n^\top \bbeta_1) - \mu(\tilde{\bx}_n^\top \bbeta_2)| = \left|\frac{1}{\tilde{\bx}_n^{\!\top}\bbeta_1}-\frac{1}{\tilde{\bx}_n^{\!\top}\bbeta_2}\right|
		\le\frac{|\tilde{\bx}_n^{\!\top}(\bbeta_1-\bbeta_2)|}{b^2}
		\le\frac{\|\tilde{\bx}_n\|}{b^2}\,\|\bbeta_1-\bbeta_2\|,
		$$
		thus,
		$$
		L(\bx)=\frac{\|\bx\|}{b^2},
		$$
		and $\E_{\hat{F}_n}\left\{L(\bx)^2\|\bx\|^2\right\}< \rho/b^4$ by the assumption.
		
		Finally, for part (c), 
		Since $\mu''(x)=2/x^3$, for any $\bbeta_i\in\mathcal{B},$
		
		$$
		\left|\mu''(\tilde{\bx}_{n,i}^{\!\top}\bbeta_i)\right|\cdot|\tilde{\bx}_{n,i}\|^3\le\frac{2\|\tilde{\bx}_{n,i}\|^3}{b^3}
		$$
		so
		$$
		\frac{1}{m}\sum_{i=1}^m\left|\mu''(\tilde{\bx}_{n,i}^{\!\top}\bbeta_i)\right|\cdot\|\tilde{\bx}_{n,i}\|^3
		\le\frac{2}{b^3}\;\frac{1}{m}\sum_{i=1}^m\|\tilde{\bx}_{n,i}\|^3
		=O_{P|\mathcal D_n}(1),
		$$
		where the last equality follows by the assumption and Markov's inequality.

	\end{proof}
	
	
	\section{Beyond Canonical GLMs} \label{appen:extensions}
	
	\subsection{Inference without Canonical GLMs} \label{appen:sec:inference}
	The main paper assumes a correctly-specified GLM with a canonical link. In this case, the solution to Eq.~(\ref{eq:ScoreMLEOrig}) in the main paper,
	$$
	\int \left\{y - \mu(\bx^\top\bbeta)\right\} \bx \, dF_n(\bx, y) \equiv \mathbf{0},
	$$  
	is the MLE, whose asymptotic variance is $\boldsymbol{\Sigma}^{-1}$, where $\boldsymbol{\Sigma} = \E\left\{\mu'(\bx^\top\bbeta^*) \bx\bx^\top\right\}$. However, if the link function is not canonical, the estimating equation above no longer corresponds to the MLE, generally leading to some efficiency loss, and the asymptotic variance is no longer simply $\boldsymbol{\Sigma}^{-1}$.
	Under model misspecification, or when only the conditional mean $\E(Y|\bx) = \mu(\bx^\top\bbeta)$ is posited without specifying the variance, a robust (sandwich) variance estimator is appropriate. Denoting  
	$$
	\mathbf{K}_{\beta\beta} = \text{Var}\left[\left\{Y - \mu\left(\bx^\top\bbeta^*\right)\right\} \bx\right],
	$$  
	and using standard M-estimation results, the asymptotic variance of $\hat{\bbeta}$ satisfies  
	$$
	\text{Var}
	\left(\sqrt{n}\hat{\bbeta}\right) = \boldsymbol{\Sigma}^{-1} \mathbf{K}_{\beta\beta} \boldsymbol{\Sigma}^{-1}.
	$$
	Therefore, the asymptotic variance of our estimator $\tilde{\bbeta}$ becomes  
	\begin{equation} \label{eq:varNonCanonical}
		\text{Var}\left(\sqrt{n}\tilde{\bbeta}\right) = \boldsymbol{\Sigma}^{-1} \left( \alpha^{-1} \boldsymbol{\mathcal{A}}_1 + \boldsymbol{\mathcal{A}}_2 \right) \boldsymbol{\Sigma}^{-1} + \boldsymbol{\Sigma}^{-1} \mathbf{K}_{\beta\beta} \boldsymbol{\Sigma}^{-1}.
	\end{equation}
	
	In a GLM with a non-canonical link, the conditional variance $\text{Var}(Y|\bx)$ is known through the model’s variance function. This allows a model-based estimator for $\mathbf{K}_{\beta\beta}$, which is valid under correct model specification. For robustness under possible misspecification, we consider instead the original-data-based empirical estimator
	$$
	\hat{\mathbf{K}}_{\beta\beta} = \frac{1}{n} \sum_{i=1}^n \left\{ Y_i - \mu(\bx_i^\top \hat{\bbeta}) \right\}^2 \bx_i \bx_i^\top \, ,
	$$
	and the corresponding synthetic-data-based estimator 
	\begin{equation*}
		\tilde{\mathbf{K}}_{\beta\beta} = \frac{1}{m} \sum_{i=1}^m \left\{ \tilde{Y}_{n,i} - \mu\left(\tilde{\bx}_{n,i}^\top \tilde{\bbeta}\right) \right\}^2 \tilde{\bx}_{n,i} \tilde{\bx}_{n,i}^\top.
	\end{equation*}
	Unlike the estimation of $\tilde{\bbeta}$, which does not use $\tilde{\by}_n$, robust estimation of $\mathbf{K}_{\beta\beta}$ does. To improve the quality of this variance estimation, we propose including $\tilde{\by}_n$ in the whitening–recoloring step, such that the Cholesky decomposition is applied to the full Gram matrix $\mathcal{X}^\top \mathcal{X}$, rather than just $\bX^\top\bX$. 
	Interestingly, if we include $\tilde{\by}_n$ in whitening–recoloring and then solve  
	$$
	\int \left\{ y - \mu\left(\bx^\top\bbeta\right) \right\} \bx \, d\tilde{F}_m(\bx, y) \equiv \mathbf{0},
	$$
	using the transformed synthetic data, we obtain the same $\tilde{\bbeta}$ as the solution to Eq.~(\ref{eq:ScoreCorrected}),
	\begin{equation*}
	\int \left\{ \bx^\top \hat{\btheta} - \mu\big(\bx^\top \bbeta\big) \right\} \bx \, d\tilde{F}_m(\bx) \equiv \mathbf{0} \, ,
\end{equation*} 	
 even though the latter does not actually depend on $\tilde{\by}_n$. This is because whitening–recoloring aligns the Gram matrices of the original and synthetic data, so the resulting linear-regression coefficients are the same.
	Still, we believe that presenting the estimator without using $\tilde{\by}_n$, as we did in Section~\ref{sec:method} of the main paper, offers clearer intuition for why this method works. Moreover, the whitening–recoloring step is not essential for good performance, so we do not wish to define our estimator in terms of the transformed $\tilde{\by}_n$.
	
To facilitate inference beyond correctly-specified canonical GLMs, we modify the sampling of $\hat{\btheta}^{(b)}$ and $\hat{\bbeta}^{(b)}$ in step (i) of the bootstrap procedure in Algorithm~\ref{alg:boot} of the main paper. This bootstrap procedure is necessary for estimating the quantity $\left( \alpha^{-1} \boldsymbol{\mathcal{A}}_1 + \boldsymbol{\mathcal{A}}_2 \right)$ in Eq.~(\ref{eq:varNonCanonical}). 
Under correct specification of a canonical GLM, the asymptotic joint distribution of $\hat{\btheta}$ and $\hat{\bbeta}$ was given in Eq.~(\ref{CovBetaTheta}) in the main paper. To adapt this step to the more general setting considered in this section, we instead use the asymptotic joint distribution
\begin{equation*}  
	\sqrt{n}  
	\begin{pmatrix}  
		\hat{\bbeta} - \bbeta^* \\  
		\hat{\btheta} - \btheta^*  
	\end{pmatrix}  
	\sim \mathcal{N} \left( 
	\mathbf{0},  
	\begin{pmatrix}  
		\mathbf{V}_{\beta\beta} & \mathbf{V}_{\beta\theta}\\  
		\mathbf{V}_{\beta\theta}  & \mathbf{V}_{\theta\theta}  
	\end{pmatrix}  
	\right) ,
\end{equation*}
where the variance blocks are given by the sandwich forms:
	\begin{align*}
		\mathbf{V}_{\beta\beta} &= \boldsymbol{\Sigma}^{-1} \mathbf{K}_{\beta\beta} \boldsymbol{\Sigma}^{-1}, \\
		\mathbf{V}_{\theta\theta} &= \E\left(\bx\bx^\top\right)^{-1}  \mathbf{K}_{\theta\theta} \E\left(\bx\bx^\top\right)^{-1}, \\
		\mathbf{V}_{\beta\theta} &= \boldsymbol{\Sigma}^{-1} \mathbf{K}_{\beta\theta} \E\left(\bx\bx^\top\right)^{-1},
	\end{align*}
	and,
	\begin{align*}
		\mathbf{K}_{\theta\theta} &= \text{Var}\left\{ \left(Y - \bx^\top \btheta^*\right) \bx \right\}, \\
		\mathbf{K}_{\beta\theta} &= \text{Cov}\left[ \left\{Y - \mu\left(\bx^\top \bbeta^*\right)\right\} \bx,\, \left(Y - \bx^\top \btheta^*\right) \bx \right].
	\end{align*}
	
	In practice, we sample realizations of $\hat{\btheta}^{(b)}$ and $\hat{\bbeta}^{(b)}$ from
	
		\begin{equation*} 
		\mathcal{N} \left\{  
		\begin{pmatrix}  
			\tilde{\bbeta} \\  
			\hat{\btheta}  
		\end{pmatrix},  
		\begin{pmatrix}  
			\tilde{\mathbf{V}}_{\beta\beta} & \tilde{\mathbf{V}}_{\beta\theta} \\  
		\tilde{\mathbf{V}}_{\beta\theta} & 	\tilde{\mathbf{V}}_{\theta\theta}
		\end{pmatrix}  
		\right\} \, ,
	\end{equation*}
where	
	\begin{align*}
		\tilde{\mathbf{V}}_{\beta\beta} &= \tilde{\boldsymbol{\cI}}_n(\tilde{\bbeta})^{-1} \, \tilde{\mathbf{K}}_{\beta\beta} \, \tilde{\boldsymbol{\cI}}_n(\tilde{\bbeta})^{-1}, \\
		\tilde{\mathbf{V}}_{\theta\theta} &= \left( \frac{1}{m} \sum_{i=1}^m \tilde{\bx}_{n,i} \tilde{\bx}_{n,i}^\top \right)^{-1}
		\tilde{\mathbf{K}}_{\theta\theta}
		\left( \frac{1}{m} \sum_{i=1}^m \tilde{\bx}_{n,i} \tilde{\bx}_{n,i}^\top \right)^{-1}, \\
		\tilde{\mathbf{V}}_{\beta\theta} &= \tilde{\boldsymbol{\cI}}_n(\tilde{\bbeta})^{-1} \, \tilde{\mathbf{K}}_{\beta\theta} \left( \frac{1}{m} \sum_{i=1}^m \tilde{\bx}_{n,i} \tilde{\bx}_{n,i}^\top \right)^{-1} \, ,
	\end{align*}
and,	
	\begin{align*}
		\tilde{\mathbf{K}}_{\theta\theta} &= \frac{1}{m} \sum_{i=1}^m \left( \tilde{Y}_{n,i} - \tilde{\bx}_{n,i}^\top \hat{\btheta} \right)^2 \tilde{\bx}_{n,i} \tilde{\bx}_{n,i}^\top, \\
		\tilde{\mathbf{K}}_{\beta\theta} &= \frac{1}{m} \sum_{i=1}^m \left\{ \tilde{Y}_{n,i} - \mu\left(\tilde{\bx}_{n,i}^\top \tilde{\bbeta}\right) \right\} \left( \tilde{Y}_{n,i} - \tilde{\bx}_{n,i}^\top \hat{\btheta} \right) \tilde{\bx}_{n,i} \tilde{\bx}_{n,i}^\top \, .
	\end{align*}
The other bootstrap steps in Algorithm~\ref{alg:boot} remain the same. 
	
	\subsection{Simulation Study}
We conducted a simulation study to evaluate the estimator defined by Eq.~(\ref{eq:ScoreCorrected}) under deviations from the canonical GLM link assumed in the main paper, and to assess the modified variance estimation procedure described in Section~\ref{appen:sec:inference}. 
Two simulation settings were considered, both involving synthetic data generated using ADS-GAN with $n = m = 500$.

	The first setting involves exponential regression with a log link, where $Y|\bx \sim \exp(\bx^\top\bbeta)$; the covariates were generated according to Setting~B in Section~\ref{sec:sims} of the main paper. Since the log link is not the canonical link function for exponential regression, the resulting estimator is generally inefficient relative to the MLE.
	
	Tables~\ref{tab:ExpCity1:10}--\ref{tab:ExpCity11:20} summarize the results: we report the MLE, the suboptimal estimator (`SubOpt') from Eq.~(\ref{eq:ScoreMLEOrig}) applied to the original data, and both the naive synthetic `MLE' (`Syn-naive') and our method applied to synthetic data (`Syn-novel'). We also report empirical SEs, mean estimated SEs from our variance procedures (based on 500 replications with 1,000 bootstraps each), and coverage rates for the original-data MLE and for Syn-novel.
	
	The second setting sampled $Y = \arctan(\bx^\top\bbeta) + \varepsilon$, with $\varepsilon \sim \mathcal{N}(0, 2)$, and covariates generated according to Setting~A from Section~\ref{sec:sims} of the main paper. Under the assumption of homoskedasticity, the estimator defined by Eq.~(\ref{eq:ScoreMLEOrig}) is efficient. If Gaussianity is additionally assumed, further gains in efficiency could be achieved by deriving the MLE, although we did not pursue this direction. Tables~\ref{tab:atanCity1:10}--\ref{tab:atanCity11:20} present the simulation results, where `Original' refers to the estimator obtained by applying Eq.~(\ref{eq:ScoreMLEOrig}) to the original data. The Syn-naive estimator was numerically unstable and occasionally produced extreme values, so the median is reported instead of the mean. Its standard error is estimated using $1.4826 \times \operatorname{MAD}$ (median absolute deviation), which is a robust approximation to the standard deviation under normality.

	To assess the impact of the ``whitening-recoloring'' procedure, in both settings we include results with and without applying this procedure to the synthetic data before solving estimating equation~(\ref{eq:ScoreCorrected}). Note that in both settings, a whitening-recoloring step was used within the bootstrap procedure to account for variability in the original data, after perturbing the Gram matrix, as described in Algorithm~\ref{alg:boot}.
	
	In both settings, Syn-novel closely matches the corresponding estimator obtained by applying Eq.~(\ref{eq:ScoreMLEOrig}) to the original data. In the exponential regression setting, where this estimating equation does not yield the MLE, some loss in efficiency is observed, as expected. The whitening-recoloring step consistently improves performance, both by reducing the SEs and by enhancing the accuracy of variance estimation, as evidenced by the close agreement between empirical and estimated SEs, and by the resulting coverage rates. Overall, the results indicate that the proposed variance estimation procedure is effective when applied to synthetic data, yielding reliable uncertainty quantification.

	\begin{table}[ht]
		\centering
		\begin{tabular}{lcccccccccc}
			\hline
			$\bbeta^*$ & 0.150 & 0.100 & 0.150 & 0.200 & -0.100 & 0.100 & 0.200 & -0.100 & -0.100 & 0.100 \\ 
			\hline
			MLE & 0.156 & 0.103 & 0.152 & 0.201 & -0.103 & 0.100 & 0.201 & -0.104 & -0.103 & 0.101 \\ 
			SubOpt & 0.157 & 0.104 & 0.153 & 0.202 & -0.102 & 0.101 & 0.200 & -0.105 & -0.103 & 0.101 \\ 
			Syn-naive & 0.354 & 0.021 & 0.067 & 0.112 & -0.081 & 0.076 & 0.097 & -0.041 & -0.046 & 0.095 \\ 
			MLE-SE & 0.133 & 0.058 & 0.075 & 0.049 & 0.045 & 0.047 & 0.064 & 0.048 & 0.048 & 0.047 \\ 
			SubOpt-SE & 0.138 & 0.059 & 0.078 & 0.052 & 0.047 & 0.050 & 0.068 & 0.049 & 0.050 & 0.049 \\ 
			Syn-naive-SE & 0.225 & 0.072 & 0.080 & 0.074 & 0.077 & 0.079 & 0.091 & 0.073 & 0.071 & 0.064 \\ 
			Coverage-MLE & 0.942 & 0.922 & 0.920 & 0.914 & 0.950 & 0.946 & 0.942 & 0.940 & 0.930 & 0.930 \\ 
			\hline
			\multicolumn{11}{c}{With ``whitening-recoloring"} \\
			Syn-novel & 0.160 & 0.106 & 0.151 & 0.198 & -0.102 & 0.102 & 0.202 & -0.106 & -0.103 & 0.105 \\ 
			Syn-novel-emp-SE & 0.140 & 0.062 & 0.087 & 0.052 & 0.049 & 0.051 & 0.070 & 0.051 & 0.050 & 0.055 \\ 
			Syn-novel-est-SE & 0.142 & 0.061 & 0.092 & 0.065 & 0.049 & 0.049 & 0.066 & 0.049 & 0.048 & 0.051 \\ 
			Coverage-novel & 0.946 & 0.952 & 0.960 & 0.976 & 0.948 & 0.924 & 0.930 & 0.948 & 0.946 & 0.938 \\ 
			\hline
			\multicolumn{11}{c}{Without ``whitening-recoloring"} \\
			Syn-novel & 0.152 & 0.109 & 0.150 & 0.198 & -0.102 & 0.102 & 0.203 & -0.107 & -0.103 & 0.107 \\ 
			Syn-novel-emp-SE & 0.146 & 0.065 & 0.092 & 0.055 & 0.051 & 0.053 & 0.074 & 0.052 & 0.051 & 0.058 \\ 
			Syn-novel-est-SE& 0.164 & 0.067 & 0.081 & 0.074 & 0.057 & 0.056 & 0.071 & 0.052 & 0.052 & 0.050 \\ 			
			Coverage-novel & 0.966 & 0.950 & 0.906 & 0.980 & 0.966 & 0.960 & 0.938 & 0.950 & 0.940 & 0.890 \\ 
			\hline
		\end{tabular}
		\caption{Simulation results for coefficients $\bbeta^*_{1:10}$ in an exponential regression with log-link, under Setting~B in Section~\ref{sec:sims} of the main paper, based on sample size of $n = 500$, using the ADS-GAN synthetic model and $m=n$. Results are based on 500 repetitions.} \label{tab:ExpCity1:10}
	\end{table}
	
	\begin{table}[ht]
		\centering
		\begin{tabular}{lcccccccccc}
			\hline
			$\bbeta^*$ & -0.100 & 0.000 & 0.000 & 0.000 & 0.000 & 0.000 & 0.000 & 0.000 & 0.000 & 0.000 \\ 
			\hline
			MLE & -0.099 & 0.000 & -0.000 & 0.001 & -0.001 & 0.001 & 0.003 & 0.005 & 0.000 & 0.001 \\ 
			SubOpt & -0.099 & 0.001 & 0.001 & 0.001 & -0.000 & -0.000 & 0.003 & 0.005 & 0.001 & 0.001 \\
			Syn-naive & -0.045 & 0.001 & 0.002 & 0.005 & 0.001 & 0.010 & -0.006 & -0.004 & -0.002 & -0.004 \\ 
			MLE-SE & 0.047 & 0.048 & 0.047 & 0.044 & 0.045 & 0.046 & 0.047 & 0.047 & 0.045 & 0.046 \\ 
			SubOpt-SE & 0.048 & 0.050 & 0.049 & 0.046 & 0.046 & 0.047 & 0.049 & 0.048 & 0.047 & 0.048 \\ 
			Syn-naive-SE & 0.073 & 0.079 & 0.075 & 0.079 & 0.076 & 0.078 & 0.070 & 0.066 & 0.070 & 0.070 \\ 
			Coverage-MLE & 0.944 & 0.942 & 0.930 & 0.958 & 0.946 & 0.946 & 0.938 & 0.932 & 0.940 & 0.950 \\
			\hline
			\multicolumn{11}{c}{With ``whitening-recoloring"} \\
			Syn-novel & -0.099 & 0.001 & -0.000 & 0.001 & -0.000 & 0.000 & 0.003 & 0.005 & 0.000 & 0.000 \\ 
			Syn-novel-emp-SE & 0.049 & 0.051 & 0.050 & 0.046 & 0.047 & 0.048 & 0.050 & 0.050 & 0.049 & 0.048 \\  
			Syn-novel-est-SE & 0.049 & 0.049 & 0.049 & 0.049 & 0.048 & 0.048 & 0.049 & 0.050 & 0.049 & 0.049 \\ 
			Coverage-novel & 0.956 & 0.936 & 0.950 & 0.946 & 0.956 & 0.962 & 0.950 & 0.942 & 0.934 & 0.938 \\ 
			\hline
			\multicolumn{11}{c}{Without ``whitening-recoloring"} \\
			Syn-novel & -0.100 & 0.001 & -0.000 & 0.001 & -0.000 & -0.000 & 0.003 & 0.004 & 0.000 & 0.000 \\ 
			Syn-novel-emp-SE & 0.052 & 0.052 & 0.052 & 0.047 & 0.046 & 0.048 & 0.051 & 0.050 & 0.049 & 0.049 \\ 
			Syn-novel-est-SE & 0.054 & 0.056 & 0.056 & 0.056 & 0.055 & 0.056 & 0.057 & 0.057 & 0.057 & 0.056 \\ 
			Coverage-novel & 0.958 & 0.962 & 0.976 & 0.974 & 0.976 & 0.980 & 0.964 & 0.966 & 0.964 & 0.966 \\ 
			\hline
		\end{tabular}
		\caption{Simulation results for coefficients $\bbeta^*_{11:20}$ in an exponential regression with log-link, under Setting~B in Section~\ref{sec:sims} of the main paper, based on sample size of $n = 500$, using the ADS-GAN synthetic model and $m=n$. Results are based on 500 repetitions.} \label{tab:ExpCity11:20}
	\end{table}

	
	\begin{table}[ht]
		\centering
		\begin{tabular}{lcccccccccc}
			\hline
			$\bbeta^*$ & 0.500 & 0.900 & -0.700 & -0.500 & 1.000 & -0.700 & -0.200 & 0.000 & 0.000 & 0.000 \\ 
			\hline
			Original & 0.560 & 0.966 & -0.749 & -0.536 & 1.113 & -0.767 & -0.240 & 0.002 & 0.005 & 0.008 \\ 
			Syn-naive & 0.584 & 0.654 & -0.546 & -0.386 & 0.787 & -0.534 & -0.141 & -0.060 & -0.030 & 0.059 \\ 
			Original-SE & 0.229 & 0.456 & 0.405 & 0.436 & 0.439 & 0.453 & 0.427 & 0.415 & 0.424 & 0.408 \\ 
			Syn-naive-SE & 0.681 & 0.755 & 0.714 & 0.729 & 0.870 & 0.758 & 0.654 & 0.710 & 0.746 & 0.711 \\ 
			Coverage-original & 0.950 & 0.926 & 0.968 & 0.954 & 0.964 & 0.936 & 0.942 & 0.938 & 0.932 & 0.956 \\ 
			\hline
			\multicolumn{11}{c}{With ``whitening-recoloring"} \\		
			Syn-novel & 0.562 & 0.973 & -0.756 & -0.541 & 1.124 & -0.774 & -0.243 & 0.000 & 0.008 & 0.004 \\ 
			Syn-novel-emp-SE & 0.231 & 0.462 & 0.408 & 0.442 & 0.443 & 0.459 & 0.435 & 0.421 & 0.423 & 0.412 \\ 
			Syn-novel-est-SE & 0.229 & 0.435 & 0.424 & 0.420 & 0.445 & 0.425 & 0.414 & 0.411 & 0.410 & 0.410 \\ 
			Coverage-novel & 0.960 & 0.944 & 0.970 & 0.958 & 0.968 & 0.950 & 0.944 & 0.944 & 0.942 & 0.952 \\ 
			\hline
			\multicolumn{11}{c}{Without ``whitening-recoloring"} \\		
			Syn-novel & 0.568 & 1.002 & -0.782 & -0.555 & 1.161 & -0.800 & -0.249 & 0.005 & 0.014 & 0.005 \\ 
			Syn-novel-emp-SE & 0.237 & 0.500 & 0.449 & 0.460 & 0.508 & 0.507 & 0.452 & 0.440 & 0.457 & 0.429 \\
			Syn-novel-est-SE & 0.319 & 0.492 & 0.474 & 0.461 & 0.509 & 0.484 & 0.457 & 0.449 & 0.456 & 0.455 \\ 
			Coverage-novel & 0.968 & 0.944 & 0.970 & 0.952 & 0.962 & 0.932 & 0.944 & 0.944 & 0.938 & 0.952 \\ 
			\hline
		\end{tabular}
		\caption{Simulation results for coefficients $\bbeta^*_{1:10}$ with $\mu(x)=\arctan(x)$, under Setting~A in Section~\ref{sec:sims} of the main paper, based on sample size of $n = 500$, using the ADS-GAN synthetic model and $m=n$. Results are based on 500 repetitions.} \label{tab:atanCity1:10}
	\end{table}
	
	\begin{table}[ht]
		\centering
		\begin{tabular}{lcccccccccc}
			\hline
			$\bbeta^*$ & 0.000 & 0.000 & 0.000 & 0.000 & 0.000 & 0.000 & 0.000 & 0.000 & 0.000 & 0.000 \\ 
			\hline
			Original & 0.006 & -0.020 & 0.008 & -0.029 & 0.003 & 0.006 & 0.011 & -0.003 & 0.000 & -0.024 \\ 
			Syn-naive & 0.016 & 0.003 & -0.078 & 0.002 & -0.056 & 0.063 & 0.018 & -0.089 & 0.038 & 0.017 \\ 
			Original-SE & 0.419 & 0.398 & 0.404 & 0.409 & 0.403 & 0.413 & 0.404 & 0.425 & 0.422 & 0.407 \\ 
			Syn-naive-SE & 0.675 & 0.654 & 0.706 & 0.721 & 0.712 & 0.698 & 0.677 & 0.672 & 0.666 & 0.650 \\ 
			Coverage-original & 0.948 & 0.964 & 0.952 & 0.960 & 0.962 & 0.960 & 0.964 & 0.948 & 0.948 & 0.964 \\ 
			\hline
			\multicolumn{11}{c}{With ``whitening-recoloring"} \\
			Syn-novel & 0.005 & -0.019 & 0.007 & -0.029 & 0.004 & 0.008 & 0.010 & -0.002 & 0.004 & -0.024 \\ 
			Syn-novel-emp-SE & 0.418 & 0.402 & 0.406 & 0.413 & 0.408 & 0.414 & 0.406 & 0.430 & 0.424 & 0.409 \\ 
			Syn-novel-est-SE & 0.409 & 0.409 & 0.410 & 0.410 & 0.411 & 0.411 & 0.409 & 0.411 & 0.409 & 0.410 \\ 
			Coverage-novel & 0.956 & 0.962 & 0.960 & 0.974 & 0.962 & 0.970 & 0.968 & 0.948 & 0.950 & 0.962 \\ 
			\hline
			\multicolumn{11}{c}{Without ``whitening-recoloring"} \\		
			Syn-novel & 0.008 & -0.020 & 0.009 & -0.033 & 0.005 & 0.002 & 0.007 & -0.004 & 0.006 & -0.029 \\ 
			Syn-novel-emp-SE & 0.442 & 0.415 & 0.417 & 0.428 & 0.420 & 0.441 & 0.428 & 0.446 & 0.439 & 0.430 \\ 
			Syn-novel-est-SE & 0.454 & 0.450 & 0.450 & 0.451 & 0.447 & 0.452 & 0.455 & 0.451 & 0.453 & 0.457 \\ 
			Coverage-novel & 0.948 & 0.964 & 0.952 & 0.964 & 0.950 & 0.968 & 0.962 & 0.962 & 0.942 & 0.962 \\ 
			\hline
		\end{tabular}
		\caption{Simulation results for coefficients $\bbeta^*_{11:20}$ with $\mu(x)=\arctan(x)$, under Setting~A in Section~\ref{sec:sims} of the main paper, based on sample size of $n = 500$, using the ADS-GAN synthetic model and $m=n$. Results are based on 500 repetitions.} \label{tab:atanCity11:20}
	\end{table}

	\FloatBarrier
	
	\section{Additional Simulation Results}\label{appen:sim}
	
	The following tables contain additional results for the simulation study described in Section~\ref{sec:sims} in the main paper.
	
	\begin{table}[ht]
		\spacingset{1}
		\centering
		\begin{tabular}{lcccccccccc}
			\hline
			$\bbeta^*$ (Setting A)& 0.000 & 0.000 & 0.000 & 0.000 & 0.000 & 0.000 & 0.000 & 0.000 & 0.000 & 0.000 \\ 
			\hline
			$n = m = 200$ &&&&&&&&&&\\ 
			MLE & 0.025 & -0.009 & 0.006 & 0.004 & -0.012 & -0.027 & 0.004 & 0.002 & -0.001 & -0.015 \\ 
			Syn-novel & 0.028 & -0.010 & -0.000 & 0.002 & -0.012 & -0.025 & -0.002 & -0.005 & -0.005 & -0.014 \\ 
			Syn-naive & -0.009 & 0.027 & -0.039 & -0.002 & -0.027 & -0.013 & -0.016 & -0.016 & 0.001 & -0.015 \\ 
			MLE-SE & 0.285 & 0.291 & 0.281 & 0.281 & 0.281 & 0.283 & 0.276 & 0.281 & 0.291 & 0.282 \\ 
			Syn-novel-SE & 0.293 & 0.311 & 0.288 & 0.283 & 0.294 & 0.294 & 0.292 & 0.294 & 0.305 & 0.293 \\ 
			Syn-naive-SE & 0.393 & 0.410 & 0.410 & 0.404 & 0.391 & 0.382 & 0.369 & 0.364 & 0.390 & 0.386 \\ 
			Coverage-MLE & 0.942 & 0.952 & 0.948 & 0.954 & 0.950 & 0.958 & 0.954 & 0.948 & 0.962 & 0.958 \\ 
			Coverage-novel & 0.946 & 0.948 & 0.958 & 0.958 & 0.962 & 0.952 & 0.950 & 0.952 & 0.954 & 0.950 \\ 
			\hline
			$n = m = 500$ &&&&&&&&&&\\
			MLE & 0.025 & -0.009 & 0.006 & 0.004 & -0.012 & -0.027 & 0.004 & 0.002 & -0.001 & -0.015 \\ 
			Syn-novel & 0.028 & -0.010 & -0.000 & 0.002 & -0.012 & -0.025 & -0.002 & -0.005 & -0.005 & -0.014 \\ 
			Syn-naive & -0.009 & 0.027 & -0.039 & -0.002 & -0.027 & -0.013 & -0.016 & -0.016 & 0.001 & -0.015 \\ 
			MLE-SE & 0.285 & 0.291 & 0.281 & 0.281 & 0.281 & 0.283 & 0.276 & 0.281 & 0.291 & 0.282 \\ 
			Syn-novel-SE & 0.293 & 0.311 & 0.288 & 0.283 & 0.294 & 0.294 & 0.292 & 0.294 & 0.305 & 0.293 \\ 
			Syn-naive-SE & 0.393 & 0.410 & 0.410 & 0.404 & 0.391 & 0.382 & 0.369 & 0.364 & 0.390 & 0.386 \\ 
			Coverage-MLE & 0.948 & 0.944 & 0.950 & 0.938 & 0.960 & 0.958 & 0.956 & 0.950 & 0.952 & 0.962 \\ 
			Coverage-novel & 0.946 & 0.948 & 0.958 & 0.958 & 0.962 & 0.952 & 0.950 & 0.952 & 0.954 & 0.950 \\ 
			\hline
			$\bbeta^*$ (Setting B)& -0.100 & 0.000 & 0.000 & 0.000 & 0.000 & 0.000 & 0.000 & 0.000 & 0.000 & 0.000 \\ 
			\hline
			$n = m = 200$ &&&&&&&&&&\\
			MLE & -0.101 & -0.001 & -0.000 & -0.000 & 0.001 & -0.001 & -0.004 & -0.000 & -0.001 & -0.003 \\ 
			Syn-novel & -0.102 & -0.001 & -0.000 & -0.000 & 0.000 & -0.001 & -0.004 & 0.001 & -0.001 & -0.004 \\ 
			Syn-naive & -0.020 & 0.005 & -0.005 & -0.002 & -0.002 & -0.008 & -0.002 & 0.005 & 0.002 & -0.008 \\ 
			MLE-SE & 0.054 & 0.056 & 0.057 & 0.055 & 0.056 & 0.055 & 0.061 & 0.058 & 0.059 & 0.056 \\ 
			Syn-novel-SE & 0.057 & 0.060 & 0.060 & 0.058 & 0.059 & 0.058 & 0.063 & 0.061 & 0.061 & 0.059 \\ 
			Syn-naive-SE & 0.095 & 0.112 & 0.116 & 0.111 & 0.109 & 0.114 & 0.096 & 0.097 & 0.100 & 0.095 \\ 
			Coverage-MLE & 0.962 & 0.938 & 0.948 & 0.956 & 0.956 & 0.960 & 0.938 & 0.948 & 0.948 & 0.958 \\ 
			Coverage-novel & 0.962 & 0.948 & 0.944 & 0.946 & 0.962 & 0.962 & 0.950 & 0.946 & 0.958 & 0.956 \\ 
			\hline
			$n = m = 500$ &&&&&&&&&&\\
			MLE & -0.101 & -0.001 & 0.001 & -0.002 & 0.001 & 0.001 & 0.000 & -0.002 & -0.002 & -0.001 \\ 
			Syn-novel & -0.102 & -0.001 & 0.003 & -0.002 & 0.001 & 0.000 & 0.000 & -0.001 & -0.002 & -0.001 \\ 
			Syn-naive & -0.024 & -0.002 & -0.005 & 0.001 & -0.001 & 0.001 & 0.006 & 0.001 & 0.001 & 0.005 \\ 
			MLE-SE & 0.034 & 0.032 & 0.035 & 0.033 & 0.034 & 0.033 & 0.032 & 0.034 & 0.034 & 0.035 \\ 
			Syn-novel-SE & 0.036 & 0.034 & 0.036 & 0.035 & 0.035 & 0.034 & 0.035 & 0.036 & 0.035 & 0.036 \\ 
			Syn-naive-SE & 0.063 & 0.064 & 0.066 & 0.068 & 0.069 & 0.063 & 0.059 & 0.063 & 0.062 & 0.063 \\ 
			Coverage-MLE & 0.952 & 0.964 & 0.938 & 0.958 & 0.964 & 0.954 & 0.960 & 0.942 & 0.942 & 0.934 \\ 
			Coverage-novel & 0.948 & 0.954 & 0.932 & 0.956 & 0.958 & 0.962 & 0.956 & 0.932 & 0.952 & 0.946 \\ 
			\hline
		\end{tabular}
		\caption{Simulation results for coefficients $\bbeta^*_{11:20}$ in a Poisson regression under Settings A and B  in Section~\ref{sec:sims} of the main paper,   based on sample sizes of $n = 200$ and $n = 500$, using the ADS-GAN synthetic model and $m=n$. Results are based on 500 repetitions.} \label{tab:PoisCity11:20}
	\end{table}

	\begin{table}[ht]
		\spacingset{1}
		\centering
		\begin{tabular}{lcccccccccc}
			\hline
			$\bbeta^*$ & 0.500 & 0.900 & -0.700 & -0.500 & 1.000 & -0.700 & -0.200 & 0.000 & 0.000 & 0.000 \\  
			\hline
			$n = m = 200$ &&&&&&&&&&\\
			MLE & 0.469 & 0.881 & -0.693 & -0.511 & 1.012 & -0.687 & -0.181 & 0.017 & -0.001 & -0.001 \\ 
			Syn-novel & 0.467 & 0.899 & -0.703 & -0.520 & 1.024 & -0.698 & -0.188 & 0.013 & -0.002 & 0.007 \\ 
			Syn-naive & 0.466 & 0.508 & -0.447 & -0.286 & 0.641 & -0.435 & -0.085 & 0.021 & -0.000 & -0.024 \\ 
			MLE-SE & 0.148 & 0.279 & 0.275 & 0.282 & 0.286 & 0.274 & 0.262 & 0.280 & 0.280 & 0.283 \\ 
			Syn-novel-SE & 0.148 & 0.290 & 0.281 & 0.294 & 0.300 & 0.284 & 0.279 & 0.290 & 0.290 & 0.293 \\ 
			Syn-naive-SE & 0.242 & 0.474 & 0.465 & 0.434 & 0.470 & 0.438 & 0.362 & 0.365 & 0.373 & 0.374 \\ 
			Coverage-MLE & 0.968 & 0.964 & 0.954 & 0.952 & 0.948 & 0.956 & 0.972 & 0.952 & 0.948 & 0.948 \\ 
			Coverage-novel & 0.972 & 0.968 & 0.962 & 0.956 & 0.960 & 0.962 & 0.960 & 0.958 & 0.950 & 0.944 \\ 
			\hline
			$n = m = 500$ &&&&&&&&&&\\
			MLE & 0.489 & 0.902 & -0.701 & -0.496 & 1.010 & -0.711 & -0.198 & -0.009 & -0.003 & 0.004 \\ 
			Syn-novel & 0.489 & 0.908 & -0.705 & -0.497 & 1.019 & -0.717 & -0.200 & -0.010 & -0.004 & 0.004 \\ 
			Syn-naive & 0.501 & 0.629 & -0.517 & -0.308 & 0.775 & -0.527 & -0.097 & -0.011 & -0.008 & -0.000 \\ 
			MLE-SE & 0.096 & 0.174 & 0.166 & 0.164 & 0.166 & 0.177 & 0.169 & 0.162 & 0.167 & 0.173 \\ 
			Syn-novel-SE & 0.098 & 0.179 & 0.172 & 0.166 & 0.170 & 0.179 & 0.175 & 0.167 & 0.170 & 0.175 \\ 
			Syn-naive-SE & 0.144 & 0.300 & 0.305 & 0.250 & 0.307 & 0.310 & 0.229 & 0.220 & 0.218 & 0.224 \\ 
			Coverage-MLE & 0.950 & 0.954 & 0.944 & 0.956 & 0.964 & 0.936 & 0.960 & 0.966 & 0.948 & 0.948 \\ 
			Coverage-novel & 0.952 & 0.954 & 0.946 & 0.958 & 0.954 & 0.936 & 0.956 & 0.958 & 0.956 & 0.950 \\ 
			\hline
			$\bbeta^*$ & 0.000 & 0.000 & 0.000 & 0.000 & 0.000 & 0.000 & 0.000 & 0.000 & 0.000 & 0.000 \\ 
			\hline
			$n = m = 200$ &&&&&&&&&&\\
			MLE & 0.025 & -0.009 & 0.006 & 0.004 & -0.012 & -0.027 & 0.004 & 0.002 & -0.001 & -0.015 \\ 
			Syn-novel & 0.031 & -0.003 & 0.006 & 0.004 & -0.013 & -0.030 & 0.005 & -0.002 & -0.005 & -0.015 \\ 
			Syn-naive & -0.019 & -0.011 & 0.021 & 0.007 & -0.008 & -0.028 & -0.021 & -0.008 & -0.004 & 0.006 \\ 
			MLE-SE & 0.285 & 0.291 & 0.281 & 0.281 & 0.281 & 0.283 & 0.276 & 0.281 & 0.291 & 0.282 \\ 
			Syn-novel-SE & 0.291 & 0.303 & 0.287 & 0.285 & 0.292 & 0.294 & 0.287 & 0.287 & 0.300 & 0.292 \\ 
			Syn-naive-SE & 0.358 & 0.364 & 0.352 & 0.336 & 0.360 & 0.363 & 0.350 & 0.360 & 0.347 & 0.343 \\ 
			Coverage-MLE & 0.942 & 0.952 & 0.948 & 0.954 & 0.950 & 0.958 & 0.954 & 0.948 & 0.962 & 0.958 \\ 
			Coverage-novel & 0.948 & 0.954 & 0.966 & 0.966 & 0.956 & 0.956 & 0.958 & 0.948 & 0.958 & 0.958 \\ 
			\hline
			$n = m = 500$ &&&&&&&&&&\\
			MLE & 0.001 & -0.006 & 0.006 & 0.000 & 0.006 & 0.007 & -0.009 & -0.004 & -0.001 & -0.005 \\ 
			Syn-novel & -0.001 & -0.007 & 0.007 & 0.000 & 0.006 & 0.010 & -0.009 & -0.005 & -0.002 & -0.004 \\ 
			Syn-naive & -0.002 & -0.012 & -0.014 & -0.007 & -0.011 & -0.019 & -0.007 & -0.015 & -0.020 & -0.010 \\ 
			MLE-SE & 0.175 & 0.173 & 0.177 & 0.178 & 0.166 & 0.167 & 0.165 & 0.168 & 0.170 & 0.169 \\ 
			Syn-novel-SE & 0.179 & 0.176 & 0.183 & 0.182 & 0.170 & 0.173 & 0.173 & 0.171 & 0.173 & 0.175 \\ 
			Syn-naive-SE & 0.230 & 0.222 & 0.222 & 0.219 & 0.214 & 0.223 & 0.211 & 0.211 & 0.212 & 0.213 \\ 
			Coverage-MLE & 0.948 & 0.944 & 0.950 & 0.938 & 0.960 & 0.958 & 0.956 & 0.950 & 0.952 & 0.962 \\ 
			Coverage-novel & 0.952 & 0.948 & 0.952 & 0.942 & 0.944 & 0.954 & 0.956 & 0.956 & 0.956 & 0.942 \\ 
			\hline
		\end{tabular}
		\caption{Simulation results for Poisson regression under Setting A  in Section~\ref{sec:sims} of the main paper,   based on sample sizes of $n = 200$ and $n = 500$, using the \texttt{synthpop} synthetic model and $m=n$. Results are based on 500 repetitions.} \label{tab:PoisAPop}
	\end{table}

	\begin{table}[ht]
		\spacingset{1}
		\centering
		\begin{tabular}{lcccccccccc}
			\hline
			$\bbeta^*$ & 0.150 & 0.100 & 0.150 & 0.200 & -0.100 & 0.100 & 0.200 & -0.100 & -0.100 & 0.100 \\ 
			\hline
			$n = m = 200$ &&&&&&&&&&\\
			MLE & 0.137 & 0.095 & 0.144 & 0.204 & -0.102 & 0.101 & 0.199 & -0.098 & -0.099 & 0.101 \\ 
			Syn-novel & 0.131 & 0.094 & 0.147 & 0.209 & -0.104 & 0.104 & 0.202 & -0.096 & -0.101 & 0.102 \\ 
			Syn-naive & 0.239 & 0.017 & 0.095 & 0.161 & -0.051 & 0.048 & 0.052 & -0.027 & -0.028 & 0.070 \\ 
			MLE-SE & 0.168 & 0.065 & 0.087 & 0.051 & 0.055 & 0.057 & 0.070 & 0.057 & 0.056 & 0.056 \\ 
			Syn-novel-SE & 0.173 & 0.068 & 0.089 & 0.057 & 0.058 & 0.060 & 0.073 & 0.060 & 0.058 & 0.057 \\ 
			Syn-naive-SE & 0.229 & 0.086 & 0.122 & 0.095 & 0.083 & 0.084 & 0.098 & 0.076 & 0.073 & 0.084 \\ 
			Coverage-MLE & 0.948 & 0.948 & 0.954 & 0.932 & 0.950 & 0.942 & 0.962 & 0.944 & 0.950 & 0.956 \\ 
			Coverage-novel & 0.944 & 0.952 & 0.964 & 0.954 & 0.956 & 0.952 & 0.962 & 0.950 & 0.950 & 0.968 \\ 
			\hline
			$n = m = 500$ &&&&&&&&&&\\
			MLE & 0.144 & 0.102 & 0.147 & 0.199 & -0.100 & 0.100 & 0.201 & -0.101 & -0.102 & 0.102 \\ 
			Syn-novel & 0.147 & 0.100 & 0.148 & 0.202 & -0.101 & 0.100 & 0.202 & -0.100 & -0.102 & 0.103 \\ 
			Syn-naive & 0.219 & 0.017 & 0.100 & 0.172 & -0.061 & 0.047 & 0.055 & -0.033 & -0.027 & 0.081 \\ 
			MLE-SE & 0.094 & 0.043 & 0.050 & 0.029 & 0.035 & 0.035 & 0.044 & 0.033 & 0.034 & 0.034 \\ 
			Syn-novel-SE & 0.096 & 0.044 & 0.053 & 0.032 & 0.037 & 0.037 & 0.046 & 0.034 & 0.035 & 0.035 \\ 
			Syn-naive-SE & 0.138 & 0.051 & 0.078 & 0.053 & 0.056 & 0.054 & 0.064 & 0.046 & 0.047 & 0.056 \\ 
			Coverage-MLE & 0.958 & 0.930 & 0.954 & 0.942 & 0.940 & 0.952 & 0.962 & 0.954 & 0.952 & 0.944 \\ 
			Coverage-novel & 0.960 & 0.920 & 0.942 & 0.946 & 0.942 & 0.950 & 0.956 & 0.944 & 0.950 & 0.942 \\ 
			\hline
			$\bbeta^*$ & -0.100 & 0.000 & 0.000 & 0.000 & 0.000 & 0.000 & 0.000 & 0.000 & 0.000 & 0.000 \\ 
			\hline
			$n = m = 200$ &&&&&&&&&&\\
			MLE & -0.101 & -0.001 & -0.000 & -0.000 & 0.001 & -0.001 & -0.004 & -0.000 & -0.001 & -0.003 \\ 
			Syn-novel & -0.102 & -0.000 & 0.001 & 0.000 & 0.002 & 0.000 & -0.007 & -0.002 & -0.002 & -0.006 \\ 
			Syn-naive & -0.016 & -0.003 & 0.004 & 0.001 & 0.001 & -0.003 & 0.003 & 0.003 & -0.005 & -0.004 \\ 
			MLE-SE & 0.054 & 0.056 & 0.057 & 0.055 & 0.056 & 0.055 & 0.061 & 0.058 & 0.059 & 0.056 \\ 
			Syn-novel-SE & 0.055 & 0.060 & 0.059 & 0.058 & 0.059 & 0.057 & 0.060 & 0.059 & 0.059 & 0.060 \\ 
			Syn-naive-SE & 0.069 & 0.076 & 0.076 & 0.071 & 0.072 & 0.078 & 0.076 & 0.077 & 0.077 & 0.080 \\ 
			Coverage-MLE & 0.962 & 0.938 & 0.948 & 0.956 & 0.956 & 0.960 & 0.938 & 0.948 & 0.948 & 0.958 \\ 
			Coverage-novel & 0.962 & 0.950 & 0.952 & 0.952 & 0.950 & 0.962 & 0.952 & 0.952 & 0.960 & 0.960 \\ 
			\hline
			$n = m = 500$ &&&&&&&&&&\\
			MLE & -0.101 & -0.001 & 0.001 & -0.002 & 0.001 & 0.001 & 0.000 & -0.002 & -0.002 & -0.001 \\ 
			Syn-novel & -0.102 & -0.001 & 0.003 & -0.003 & 0.000 & 0.000 & -0.003 & -0.003 & -0.004 & -0.003 \\ 
			Syn-naive & -0.022 & 0.000 & -0.000 & -0.001 & 0.001 & 0.002 & 0.002 & 0.001 & 0.004 & 0.005 \\ 
			MLE-SE & 0.034 & 0.032 & 0.035 & 0.033 & 0.034 & 0.033 & 0.032 & 0.034 & 0.034 & 0.035 \\ 
			Syn-novel-SE & 0.035 & 0.034 & 0.036 & 0.034 & 0.035 & 0.034 & 0.034 & 0.035 & 0.036 & 0.036 \\ 
			Syn-naive-SE & 0.043 & 0.048 & 0.047 & 0.042 & 0.047 & 0.047 & 0.045 & 0.046 & 0.047 & 0.046 \\ 
			Coverage-MLE & 0.952 & 0.964 & 0.938 & 0.958 & 0.964 & 0.954 & 0.960 & 0.942 & 0.942 & 0.934 \\ 
			Coverage-novel & 0.948 & 0.962 & 0.942 & 0.950 & 0.954 & 0.960 & 0.962 & 0.952 & 0.946 & 0.940 \\ 
			\hline
		\end{tabular}
		\caption{Simulation results for Poisson regression under Setting B  in Section~\ref{sec:sims} of the main paper,   based on sample sizes of $n = 200$ and $n = 500$, using the \texttt{synthpop} synthetic model and $m=n$. Results are based on 500 repetitions.} \label{tab:PoisPop}
	\end{table}

	\begin{table}[ht]
		\spacingset{1}
		\centering
		\begin{tabular}{lcccccccccc}
			\hline
			$\bbeta^*$ (Setting A) & 0.000 & 0.000 & 0.000 & 0.000 & 0.000 & 0.000 & 0.000 & 0.000 & 0.000 & 0.000 \\ 
			\hline
			$n = m = 200$ &&&&&&&&&&\\
			MLE & -0.049 & 0.035 & 0.038 & -0.003 & 0.049 & 0.042 & -0.019 & 0.000 & 0.018 & 0.034 \\ 
			Syn-novel & -0.050 & 0.035 & 0.042 & 0.000 & 0.047 & 0.041 & -0.017 & 0.001 & 0.017 & 0.039 \\ 
			Syn-naive & -0.133 & 0.031 & 0.016 & -0.046 & -0.005 & 0.047 & 0.085 & -0.053 & 0.015 & 0.050 \\ 
			MLE-SE & 0.866 & 0.852 & 0.809 & 0.792 & 0.873 & 0.840 & 0.819 & 0.866 & 0.842 & 0.804 \\ 
			Syn-novel-SE & 0.872 & 0.864 & 0.822 & 0.789 & 0.882 & 0.848 & 0.830 & 0.876 & 0.855 & 0.815 \\ 
			Syn-naive-SE & 1.412 & 1.417 & 1.434 & 1.447 & 1.356 & 1.414 & 1.394 & 1.308 & 1.357 & 1.417 \\ 
			Coverage-MLE & 0.924 & 0.930 & 0.948 & 0.950 & 0.914 & 0.954 & 0.942 & 0.924 & 0.930 & 0.938 \\ 
			Coverage-novel & 0.926 & 0.930 & 0.952 & 0.958 & 0.930 & 0.960 & 0.946 & 0.930 & 0.934 & 0.948 \\ 
			\hline
			$n = m = 500$ &&&&&&&&&&\\
			MLE & -0.038 & 0.008 & -0.022 & 0.023 & 0.007 & -0.012 & -0.024 & 0.024 & -0.004 & 0.023 \\ 
			Syn-novel & -0.037 & 0.009 & -0.022 & 0.023 & 0.007 & -0.011 & -0.023 & 0.022 & -0.004 & 0.022 \\ 
			Syn-naive & -0.064 & 0.016 & 0.036 & -0.004 & -0.022 & -0.024 & -0.070 & 0.021 & 0.023 & 0.040 \\ 
			MLE-SE & 0.480 & 0.465 & 0.494 & 0.469 & 0.444 & 0.449 & 0.496 & 0.471 & 0.467 & 0.455 \\ 
			Syn-novel-SE & 0.484 & 0.468 & 0.496 & 0.471 & 0.445 & 0.450 & 0.498 & 0.471 & 0.469 & 0.458 \\ 
			Syn-naive-SE & 0.657 & 0.661 & 0.687 & 0.639 & 0.599 & 0.657 & 0.670 & 0.635 & 0.654 & 0.618 \\ 
			Coverage-MLE & 0.932 & 0.954 & 0.926 & 0.962 & 0.958 & 0.950 & 0.922 & 0.942 & 0.942 & 0.952 \\ 
			Coverage-novel & 0.932 & 0.954 & 0.930 & 0.962 & 0.960 & 0.950 & 0.924 & 0.942 & 0.942 & 0.952 \\ 
			\hline
			$\bbeta^*$ (Setting B) & -0.100 & 0.000 & 0.000 & 0.000 & 0.000 & 0.000 & 0.000 & 0.000 & 0.000 & 0.000 \\ 
			\hline
			$n = m = 200$ &&&&&&&&&&\\
			MLE & -0.122 & 0.006 & 0.007 & 0.007 & 0.004 & 0.002 & 0.015 & -0.011 & -0.002 & 0.010 \\ 
			Syn-novel & -0.122 & 0.005 & 0.008 & 0.008 & 0.004 & 0.001 & 0.013 & -0.013 & -0.002 & 0.009 \\ 
			Syn-naive & -0.080 & -0.025 & 0.013 & 0.017 & 0.003 & 0.029 & 0.032 & 0.001 & 0.023 & 0.010 \\ 
			MLE-SE & 0.171 & 0.168 & 0.182 & 0.176 & 0.181 & 0.172 & 0.187 & 0.190 & 0.192 & 0.173 \\ 
			Syn-novel-SE & 0.172 & 0.168 & 0.180 & 0.179 & 0.182 & 0.172 & 0.187 & 0.188 & 0.190 & 0.171 \\ 
			Syn-naive-SE & 0.448 & 0.534 & 0.505 & 0.511 & 0.529 & 0.528 & 0.427 & 0.458 & 0.429 & 0.467 \\ 
			Coverage-MLE & 0.946 & 0.968 & 0.936 & 0.954 & 0.940 & 0.954 & 0.944 & 0.926 & 0.954 & 0.942 \\ 
			Coverage-novel & 0.952 & 0.970 & 0.946 & 0.954 & 0.948 & 0.958 & 0.946 & 0.940 & 0.952 & 0.948 \\ 
			\hline
			$n = m = 500$ &&&&&&&&&&\\
			MLE & -0.102 & 0.003 & 0.001 & 0.009 & 0.002 & -0.002 & -0.000 & 0.001 & 0.008 & 0.002 \\ 
			Syn-novel & -0.102 & 0.003 & 0.001 & 0.009 & 0.002 & -0.002 & -0.001 & 0.001 & 0.007 & 0.001 \\ 
			Syn-naive & -0.055 & 0.012 & 0.009 & -0.003 & 0.007 & 0.004 & 0.017 & -0.001 & 0.009 & 0.002 \\ 
			MLE-SE & 0.099 & 0.098 & 0.099 & 0.100 & 0.104 & 0.100 & 0.098 & 0.107 & 0.105 & 0.102 \\ 
			Syn-novel-SE & 0.099 & 0.098 & 0.099 & 0.100 & 0.104 & 0.100 & 0.098 & 0.106 & 0.105 & 0.101 \\ 
			Syn-naive-SE & 0.174 & 0.201 & 0.198 & 0.205 & 0.212 & 0.211 & 0.185 & 0.178 & 0.184 & 0.177 \\ 
			Coverage-MLE & 0.944 & 0.960 & 0.954 & 0.942 & 0.950 & 0.948 & 0.964 & 0.938 & 0.946 & 0.958 \\ 
			Coverage-novel & 0.950 & 0.962 & 0.950 & 0.942 & 0.952 & 0.946 & 0.964 & 0.936 & 0.946 & 0.962 \\ 
			\hline
		\end{tabular}
		\caption{Simulation results for coefficients $\bbeta^*_{11:20}$ in a Logistic regression under Settings A and B  in Section~\ref{sec:sims} of the main paper,   based on sample sizes of $n = 200$ and $n = 500$, using the ADS-GAN synthetic model and $m=n$. Results are based on 500 repetitions.} \label{tab:LogisCity11:20}
	\end{table}

	\begin{table}[ht]
		\spacingset{1}
		\centering
		\begin{tabular}{lcccccccccc}
			\hline
			$\bbeta^*$ & 0.500 & 0.900 & -0.700 & -0.500 & 1.000 & -0.700 & -0.200 & 0.000 & 0.000 & 0.000 \\ 
			\hline
			$n = m = 200$ &&&&&&&&&&\\
			MLE & 0.563 & 1.018 & -0.777 & -0.598 & 1.093 & -0.804 & -0.262 & -0.012 & -0.014 & -0.022 \\ 
			Syn-novel & 0.565 & 1.020 & -0.782 & -0.600 & 1.092 & -0.808 & -0.261 & -0.010 & -0.017 & -0.026 \\ 
			Syn-naive & 0.571 & 0.610 & -0.538 & -0.292 & 0.651 & -0.523 & -0.160 & -0.065 & -0.091 & 0.002 \\ 
			MLE-SE & 0.426 & 0.783 & 0.784 & 0.836 & 0.889 & 0.793 & 0.787 & 0.819 & 0.837 & 0.805 \\ 
			Syn-novel-SE & 0.427 & 0.789 & 0.795 & 0.845 & 0.897 & 0.797 & 0.796 & 0.824 & 0.844 & 0.811 \\ 
			Syn-naive-SE & 0.745 & 1.183 & 1.283 & 1.227 & 1.201 & 1.236 & 1.170 & 1.227 & 1.111 & 1.198 \\ 
			Coverage-MLE & 0.956 & 0.956 & 0.948 & 0.944 & 0.916 & 0.946 & 0.960 & 0.930 & 0.944 & 0.948 \\ 
			Coverage-novel & 0.960 & 0.960 & 0.952 & 0.948 & 0.930 & 0.950 & 0.962 & 0.940 & 0.940 & 0.952 \\ 
			\hline
			$n = m = 500$ &&&&&&&&&&\\
			MLE & 0.519 & 0.939 & -0.753 & -0.513 & 1.058 & -0.719 & -0.233 & 0.021 & 0.010 & 0.005 \\ 
			Syn-novel & 0.519 & 0.942 & -0.754 & -0.515 & 1.059 & -0.721 & -0.234 & 0.021 & 0.011 & 0.003 \\ 
			Syn-naive & 0.487 & 0.671 & -0.561 & -0.371 & 0.741 & -0.497 & -0.095 & 0.047 & -0.017 & -0.001 \\ 
			MLE-SE & 0.263 & 0.485 & 0.462 & 0.487 & 0.439 & 0.493 & 0.461 & 0.455 & 0.459 & 0.464 \\ 
			Syn-novel-SE & 0.263 & 0.486 & 0.463 & 0.487 & 0.442 & 0.494 & 0.463 & 0.455 & 0.462 & 0.465 \\ 
			Syn-naive-SE & 0.421 & 0.720 & 0.672 & 0.691 & 0.744 & 0.690 & 0.641 & 0.639 & 0.655 & 0.652 \\ 
			Coverage-MLE & 0.940 & 0.934 & 0.930 & 0.936 & 0.952 & 0.926 & 0.952 & 0.948 & 0.946 & 0.934 \\ 
			Coverage-novel & 0.940 & 0.936 & 0.934 & 0.944 & 0.952 & 0.930 & 0.950 & 0.952 & 0.946 & 0.938 \\ 
			\hline
			$\bbeta^*$ & 0.000 & 0.000 & 0.000 & 0.000 & 0.000 & 0.000 & 0.000 & 0.000 & 0.000 & 0.000 \\ 
			\hline
			$n = m = 200$ &&&&&&&&&&\\
			MLE & -0.049 & 0.035 & 0.038 & -0.003 & 0.049 & 0.042 & -0.019 & 0.000 & 0.018 & 0.034 \\ 
			Syn-novel & -0.046 & 0.042 & 0.040 & -0.007 & 0.046 & 0.042 & -0.017 & 0.004 & 0.014 & 0.038 \\ 
			Syn-naive & 0.049 & -0.040 & 0.028 & -0.072 & 0.020 & 0.053 & 0.053 & 0.069 & 0.016 & 0.015 \\ 
			MLE-SE & 0.866 & 0.852 & 0.809 & 0.792 & 0.873 & 0.840 & 0.819 & 0.866 & 0.842 & 0.804 \\ 
			Syn-novel-SE & 0.865 & 0.858 & 0.818 & 0.793 & 0.880 & 0.840 & 0.819 & 0.865 & 0.849 & 0.807 \\ 
			Syn-naive-SE & 1.162 & 1.151 & 1.123 & 1.129 & 1.079 & 1.103 & 1.111 & 1.167 & 1.084 & 1.134 \\ 
			Coverage-MLE & 0.924 & 0.930 & 0.948 & 0.950 & 0.914 & 0.954 & 0.942 & 0.924 & 0.930 & 0.938 \\ 
			Coverage-novel & 0.930 & 0.930 & 0.952 & 0.954 & 0.930 & 0.956 & 0.952 & 0.930 & 0.932 & 0.948 \\ 
			\hline
			$n = m = 500$ &&&&&&&&&&\\
			MLE & -0.038 & 0.008 & -0.022 & 0.023 & 0.007 & -0.012 & -0.024 & 0.024 & -0.004 & 0.023 \\ 
			Syn-novel & -0.038 & 0.008 & -0.021 & 0.024 & 0.008 & -0.013 & -0.024 & 0.023 & -0.003 & 0.023 \\ 
			Syn-naive & -0.000 & 0.034 & -0.020 & -0.002 & -0.003 & -0.007 & -0.049 & 0.013 & -0.019 & -0.017 \\ 
			MLE-SE & 0.480 & 0.465 & 0.494 & 0.469 & 0.444 & 0.449 & 0.496 & 0.471 & 0.467 & 0.455 \\ 
			Syn-novel-SE & 0.481 & 0.466 & 0.496 & 0.470 & 0.445 & 0.450 & 0.498 & 0.470 & 0.468 & 0.456 \\ 
			Syn-naive-SE & 0.656 & 0.621 & 0.653 & 0.642 & 0.637 & 0.656 & 0.617 & 0.613 & 0.615 & 0.603 \\ 
			Coverage-MLE & 0.932 & 0.954 & 0.926 & 0.962 & 0.958 & 0.950 & 0.922 & 0.942 & 0.942 & 0.952 \\ 
			Coverage-novel & 0.934 & 0.954 & 0.930 & 0.966 & 0.960 & 0.948 & 0.924 & 0.946 & 0.942 & 0.952 \\ 
			\hline
		\end{tabular}
		\caption{Simulation results for Logistic regression under Setting A  in Section~\ref{sec:sims} of the main paper,   based on sample sizes of $n = 200$ and $n = 500$, using the \texttt{synthpop} synthetic model and $m=n$. Results are based on 500 repetitions.} \label{tab:LogisAPop}
	\end{table}

	\begin{table}[ht]
		\spacingset{1}
		\centering
		\begin{tabular}{lcccccccccc}
			\hline
			$\bbeta^*$ & 0.150 & 0.100 & 0.150 & 0.200 & -0.100 & 0.100 & 0.200 & -0.100 & -0.100 & 0.100 \\  
			\hline
			$n = m = 200$ &&&&&&&&&&\\
			MLE & 0.127 & 0.122 & 0.192 & 0.230 & -0.114 & 0.112 & 0.235 & -0.126 & -0.113 & 0.111 \\ 
			Syn-novel & 0.115 & 0.124 & 0.196 & 0.235 & -0.114 & 0.114 & 0.238 & -0.128 & -0.115 & 0.111 \\ 
			Syn-naive & 0.418 & 0.025 & 0.075 & 0.139 & -0.076 & 0.059 & 0.090 & -0.031 & -0.046 & 0.082 \\ 
			MLE-SE & 0.495 & 0.218 & 0.295 & 0.202 & 0.175 & 0.185 & 0.218 & 0.186 & 0.175 & 0.176 \\ 
			Syn-novel-SE & 0.508 & 0.225 & 0.305 & 0.209 & 0.175 & 0.188 & 0.223 & 0.189 & 0.179 & 0.180 \\ 
			Syn-naive-SE & 0.733 & 0.279 & 0.362 & 0.303 & 0.277 & 0.271 & 0.301 & 0.232 & 0.240 & 0.239 \\ 
			Coverage-MLE & 0.940 & 0.934 & 0.926 & 0.932 & 0.940 & 0.934 & 0.956 & 0.936 & 0.934 & 0.940 \\ 
			Coverage-novel & 0.940 & 0.938 & 0.938 & 0.956 & 0.952 & 0.944 & 0.960 & 0.940 & 0.946 & 0.950 \\ 
			\hline
			$n = m = 500$ &&&&&&&&&&\\
			MLE & 0.158 & 0.101 & 0.141 & 0.209 & -0.103 & 0.098 & 0.207 & -0.096 & -0.103 & 0.110 \\ 
			Syn-novel & 0.153 & 0.103 & 0.143 & 0.209 & -0.103 & 0.098 & 0.208 & -0.097 & -0.104 & 0.110 \\ 
			Syn-naive & 0.375 & 0.014 & 0.076 & 0.151 & -0.073 & 0.054 & 0.086 & -0.040 & -0.036 & 0.087 \\ 
			MLE-SE & 0.277 & 0.123 & 0.167 & 0.110 & 0.105 & 0.101 & 0.138 & 0.097 & 0.099 & 0.105 \\ 
			Syn-novel-SE & 0.279 & 0.125 & 0.169 & 0.110 & 0.106 & 0.101 & 0.139 & 0.098 & 0.099 & 0.106 \\ 
			Syn-naive-SE & 0.406 & 0.158 & 0.230 & 0.177 & 0.158 & 0.151 & 0.179 & 0.138 & 0.130 & 0.142 \\
			Coverage-MLE & 0.966 & 0.942 & 0.928 & 0.956 & 0.932 & 0.956 & 0.952 & 0.960 & 0.954 & 0.932 \\  
			Coverage-novel & 0.968 & 0.944 & 0.930 & 0.962 & 0.930 & 0.958 & 0.950 & 0.958 & 0.956 & 0.934 \\ 
			\hline
			$\bbeta^*$ & -0.100 & 0.000 & 0.000 & 0.000 & 0.000 & 0.000 & 0.000 & 0.000 & 0.000 & 0.000 \\ 
			\hline
			$n = m = 200$ &&&&&&&&&&\\
			MLE & -0.122 & 0.006 & 0.007 & 0.007 & 0.004 & 0.002 & 0.015 & -0.011 & -0.002 & 0.010 \\ 
			Syn-novel & -0.122 & 0.006 & 0.009 & 0.008 & 0.003 & 0.002 & 0.016 & -0.011 & 0.001 & 0.011 \\ 
			Syn-naive & -0.041 & 0.007 & 0.007 & 0.012 & 0.005 & -0.002 & -0.009 & -0.015 & -0.018 & 0.003 \\ 
			MLE-SE & 0.171 & 0.168 & 0.182 & 0.176 & 0.181 & 0.172 & 0.187 & 0.190 & 0.192 & 0.173 \\ 
			Syn-novel-SE & 0.173 & 0.168 & 0.183 & 0.179 & 0.184 & 0.176 & 0.191 & 0.193 & 0.194 & 0.172 \\ 
			Syn-naive-SE & 0.212 & 0.227 & 0.244 & 0.243 & 0.248 & 0.248 & 0.253 & 0.259 & 0.241 & 0.252 \\ 
			Coverage-MLE & 0.946 & 0.968 & 0.936 & 0.954 & 0.940 & 0.954 & 0.944 & 0.926 & 0.954 & 0.942 \\ 
			Coverage-novel & 0.952 & 0.970 & 0.942 & 0.952 & 0.946 & 0.956 & 0.952 & 0.938 & 0.958 & 0.948 \\ 
			\hline
			$n = m = 500$ &&&&&&&&&&\\
			MLE & -0.102 & 0.003 & 0.001 & 0.009 & 0.002 & -0.002 & -0.000 & 0.001 & 0.008 & 0.002 \\ 
			Syn-novel & -0.103 & 0.003 & 0.001 & 0.009 & 0.002 & -0.002 & 0.001 & 0.002 & 0.008 & 0.002 \\ 
			Syn-naive & -0.034 & -0.001 & 0.001 & 0.002 & 0.005 & 0.006 & -0.019 & -0.015 & -0.007 & -0.017 \\ 
			MLE-SE & 0.099 & 0.098 & 0.099 & 0.100 & 0.104 & 0.100 & 0.098 & 0.107 & 0.105 & 0.102 \\ 
			Syn-novel-SE & 0.099 & 0.098 & 0.099 & 0.100 & 0.104 & 0.100 & 0.098 & 0.108 & 0.105 & 0.103 \\ 
			Syn-naive-SE & 0.124 & 0.140 & 0.135 & 0.148 & 0.149 & 0.144 & 0.139 & 0.134 & 0.158 & 0.141 \\ 
			Coverage-MLE & 0.944 & 0.960 & 0.954 & 0.942 & 0.950 & 0.948 & 0.964 & 0.938 & 0.946 & 0.958 \\ 
			Coverage-novel & 0.950 & 0.964 & 0.954 & 0.942 & 0.952 & 0.954 & 0.966 & 0.938 & 0.940 & 0.962 \\ 
			\hline
		\end{tabular}
		\caption{Simulation results for Logistic regression under Setting B  in Section~\ref{sec:sims} of the main paper, based on sample sizes of $n = 200$ and $n = 500$, using the \texttt{synthpop} synthetic model and $m=n$. Results are based on 500 repetitions.} \label{tab:LogisBPop}
	\end{table}

	\begin{table}[ht]
		\spacingset{1}
		\centering
		\resizebox{\textwidth}{!}{%
			\begin{tabular}{l|ccccc|ccccc}
				\hline
				$n = 20{,}000$& \multicolumn{5}{c|}{$\beta^*_5=1.0$} & \multicolumn{5}{c}{$\beta^*_6=-0.7$} \\
				m & 200 & 500 & 1000 & 5000 & 20000 & 200 & 500 & 1000 & 5000 & 20000 \\
				\hline
				MLE 
				&  &  & 1.000 &  &  
				&  &  & -0.699 &  &  \\
				Syn-novel 
				& 1.016 & 1.008 & 1.004 & 1.001 & 1.000 
				& -0.710 & -0.703 & -0.701 & -0.699 & -0.699 \\
				Syn-naive 
				& 0.920 & 0.919 & 0.915 & 0.926 & 0.931 
				& -0.631 & -0.639 & -0.636 & -0.645 & -0.649 \\
				MLE-SE 
				&  &  & 0.026 &  &  
				&  &  & 0.026 &  &  \\
				Syn-novel-emp-SE 
				& 0.047 & 0.037 & 0.033 & 0.028 & 0.027 
				& 0.047 & 0.036 & 0.032 & 0.028 & 0.027 \\
				Syn-novel-est-SE 
				& 0.047 & 0.036 & 0.032 & 0.028 & 0.027 
				& 0.044 & 0.034 & 0.031 & 0.027 & 0.027 \\
				Syn-naive-SE 
				& 0.287 & 0.175 & 0.123 & 0.061 & 0.040 
				& 0.291 & 0.177 & 0.125 & 0.064 & 0.043 \\
				Coverage-MLE & & &0.952  & & & & &0.962  & & \\
				Coverage-novel 
				& 0.924 & 0.944 & 0.938 & 0.956 & 0.946 
				& 0.928 & 0.944 & 0.936 & 0.936 & 0.950 \\
				\hline
				\multicolumn{11}{c}{} \\ 
				\hline
				$n = 20{,}000$& \multicolumn{5}{c|}{$\beta^*_7=-0.2$} & \multicolumn{5}{c}{$\beta^*_8=0.0$} \\
				m & 200 & 500 & 1000 & 5000 & 20000 & 200 & 500 & 1000 & 5000 & 20000 \\
				\hline
				MLE 
				&  &  & -0.197 &  &  
				&  &  & 0.001 &  &  \\
				Syn-novel 
				& -0.204 & -0.201 & -0.199 & -0.198 & -0.197 
				& 0.002 & 0.001 & 0.001 & 0.001 & 0.001 \\
				Syn-naive 
				& -0.127 & -0.133 & -0.128 & -0.123 & -0.125 
				& -0.005 & -0.002 & 0.002 & -0.006 & -0.004 \\
				MLE-SE 
				&  &  & 0.025 &  &  
				&  &  & 0.026 &  &  \\
				Syn-novel-emp-SE 
				& 0.044 & 0.032 & 0.029 & 0.026 & 0.025 
				& 0.044 & 0.034 & 0.030 & 0.027 & 0.026 \\
				Syn-novel-est-SE 
				& 0.042 & 0.034 & 0.030 & 0.027 & 0.026 
				& 0.042 & 0.034 & 0.030 & 0.027 & 0.026 \\
				Syn-naive-SE 
				& 0.285 & 0.176 & 0.127 & 0.062 & 0.040 
				& 0.307 & 0.171 & 0.125 & 0.056 & 0.034 \\
				Coverage-MLE & & &0.968  & & & & &0.952  & & \\
				Coverage-novel 
				& 0.944 & 0.966 & 0.958 & 0.962 & 0.958 
				& 0.918 & 0.944 & 0.958 & 0.946 & 0.944 \\
				\hline
			\end{tabular}%
		}
		\caption{Simulation results for the coefficients $\bbeta^*_{5:8}$ in a Poisson regression under Setting A  in Section~\ref{sec:sims} of the main paper, with a fixed sample size of $n = 20{,}000$ and varying $m$ values, using the \texttt{synthpop} synthetic model. Results are based on 500 repetitions. Since the MLE is invariant to $m$, only a single value is shown.}
		\label{tab:res20k5:8}
	\end{table}

	\begin{table}[ht]
		\spacingset{1}
		\centering
		\resizebox{\textwidth}{!}{%
			\begin{tabular}{l|ccccc|ccccc}
				\hline
				$n = 20{,}000$& \multicolumn{5}{c|}{$\beta^*_{9}=0$} & \multicolumn{5}{c}{$\beta^*_{10}=0$} \\
				$m$ & 200 & 500 & 1000 & 5000 & 20000 & 200 & 500 & 1000 & 5000 & 20000 \\
				\hline
				MLE &  &  & -0.003 &  &  &  &  & 0.000 &  &  \\
				Syn-novel & -0.001 & -0.001 & -0.002 & -0.002 & -0.003 & -0.000 & -0.001 & -0.000 & -0.001 & 0.000 \\
				Syn-naive & -0.004 & -0.011 & 0.000 & -0.006 & -0.008 & 0.013 & 0.013 & 0.003 & -0.008 & -0.006 \\
				MLE-SE &  &  & 0.025 &  &  &  &  & 0.028 &  &  \\
				Syn-novel-emp-SE & 0.043 & 0.033 & 0.029 & 0.027 & 0.026 & 0.043 & 0.035 & 0.032 & 0.028 & 0.028 \\
				Syn-novel-est-SE & 0.042 & 0.034 & 0.030 & 0.027 & 0.026 & 0.042 & 0.034 & 0.030 & 0.027 & 0.026 \\
				Syn-naive-SE & 0.258 & 0.165 & 0.117 & 0.063 & 0.036 & 0.288 & 0.176 & 0.124 & 0.058 & 0.036 \\
				Coverage-MLE & & &0.958   & & & & &0.928  & & \\
				Coverage-novel  & 0.940 & 0.952 & 0.954 & 0.952 & 0.952 & 0.948 & 0.934 & 0.936 & 0.940 & 0.930 \\
				\hline
				\multicolumn{11}{c}{} \\
				\hline
				$n = 20{,}000$& \multicolumn{5}{c|}{$\beta^*_{11}=0$} & \multicolumn{5}{c}{$\beta^*_{12}=0$} \\
				$m$ & 200 & 500 & 1000 & 5000 & 20000 & 200 & 500 & 1000 & 5000 & 20000 \\
				\hline
				MLE &  &  & 0.000 &  &  &  &  & 0.002 &  &  \\
				Syn-novel & -0.002 & 0.000 & -0.000 & -0.000 & -0.000 & 0.006 & 0.005 & 0.003 & 0.003 & 0.003 \\
				Syn-naive & 0.002 & -0.015 & -0.007 & -0.011 & -0.008 & -0.009 & -0.004 & -0.009 & -0.008 & -0.007 \\
				MLE-SE &  &  & 0.027 &  &  &  &  & 0.025 &  &  \\
				Syn-novel-emp-SE & 0.042 & 0.034 & 0.031 & 0.028 & 0.027 & 0.042 & 0.034 & 0.030 & 0.026 & 0.025 \\
				Syn-novel-est-SE & 0.042 & 0.034 & 0.030 & 0.027 & 0.026 & 0.042 & 0.034 & 0.030 & 0.027 & 0.026 \\
				Syn-naive-SE & 0.287 & 0.170 & 0.125 & 0.057 & 0.37 & 0.295 & 0.172 & 0.119 & 0.056 & 0.034 \\
				Coverage-MLE & & &0.956   & & & & &0.960  & & \\
				Coverage-novel & 0.938 & 0.948 & 0.940 & 0.942 & 0.942 & 0.960 & 0.946 & 0.948 & 0.956 & 0.962 \\
				\hline
			\end{tabular}
		}
		\caption{Simulation results for the coefficients $\bbeta^*_{9:12}$ in a Poisson regression under Setting A  in Section~\ref{sec:sims} of the main paper, with a fixed sample size of $n = 20{,}000$ and varying $m$ values, using the \texttt{synthpop} synthetic model. Results are based on 500 repetitions. Since the MLE is invariant to $m$, only a single value is shown.}
		\label{tab:res20k9:12}
	\end{table}

	\begin{table}[ht]
		\spacingset{1}
		\centering
		\resizebox{\textwidth}{!}{%
			\begin{tabular}{l|ccccc|ccccc}
				\hline
				$n = 20{,}000$& \multicolumn{5}{c|}{$\beta^*_{13}=0$} & \multicolumn{5}{c}{$\beta^*_{14}=0$} \\
				$m$ & 200 & 500 & 1000 & 5000 & 20000 & 200 & 500 & 1000 & 5000 & 20000 \\
				\hline
				MLE              &  &   & 0.002  &   &   &  &  & -0.002 &  &  \\
				Syn-novel        & 0.002  & 0.001  & 0.003  & 0.002  & 0.002  & -0.002 & -0.002 & -0.002 & -0.002 & -0.002 \\
				Syn-naive        & -0.008 & -0.002 & -0.007 & -0.008 & -0.006 & -0.008 & -0.002 & -0.001 & -0.004 & -0.007 \\
				MLE-SE           &   &   & 0.025  &   &   &   &   & 0.027  &   &   \\
				Syn-novel-emp-SE & 0.043  & 0.033  & 0.030  & 0.026  & 0.025  & 0.044  & 0.036  & 0.031  & 0.028  & 0.028  \\
				Syn-novel-est-SE  & 0.042  & 0.034  & 0.030  & 0.027  & 0.026  & 0.042  & 0.034  & 0.030  & 0.027  & 0.026  \\
				Syn-naive-SE     & 0.289  & 0.173  & 0.118  & 0.057  & 0.034  & 0.282  & 0.171  & 0.119  & 0.056  & 0.035  \\
				Coverage-MLE & & &0.958   & & & & &0.936  & & \\
				Coverage-novel   & 0.938  & 0.948  & 0.952  & 0.954  & 0.964  & 0.932  & 0.930  & 0.936  & 0.930  & 0.930  \\
				\hline
				\multicolumn{11}{c}{} \\
				\hline
				$n = 20{,}000$& \multicolumn{5}{c|}{$\beta^*_{15}=0$} & \multicolumn{5}{c}{$\beta^*_{16}=0$} \\
				$m$ & 200 & 500 & 1000 & 5000 & 20000 & 200 & 500 & 1000 & 5000 & 20000 \\
				\hline
				MLE              &  & & -0.001 &  &  &  &  & -0.000 &  &  \\
				Syn-novel        & -0.003 & -0.002 & -0.002 & -0.001 & -0.002 & -0.000 &  0.000 & -0.000 & -0.001 & -0.000 \\
				Syn-naive        & -0.009 & -0.007 & -0.009 & -0.009 & -0.008 & -0.019 & -0.004 & -0.006 & -0.002 & -0.007 \\
				MLE-SE           &  &   & 0.026  &   &   &   &   & 0.028  &   &   \\
				Syn-novel-emp-SE  & 0.046  & 0.034  & 0.030  & 0.027  & 0.026  & 0.043  & 0.035  & 0.032  & 0.029  & 0.028  \\
				Syn-novel-est-SE  & 0.042  & 0.034  & 0.030  & 0.027  & 0.026  & 0.042  & 0.034  & 0.030  & 0.027  & 0.026  \\
				Syn-naive-SE     & 0.299  & 0.174  & 0.123  & 0.056  & 0.035  & 0.283  & 0.167  & 0.116  & 0.059  & 0.035  \\
				Coverage-MLE & & &0.948   & & & & &0.932  & & \\
				Coverage-novel  & 0.920  & 0.936  & 0.948  & 0.940  & 0.952  & 0.942  & 0.936  & 0.940  & 0.932  & 0.936  \\
				\hline
			\end{tabular}
		}
		\caption{Simulation results for the coefficients $\bbeta^*_{13:16}$ in a Poisson regression under Setting A  in Section~\ref{sec:sims} of the main paper, with a fixed sample size of $n = 20{,}000$ and varying $m$ values, using the \texttt{synthpop} synthetic model. Results are based on 500 repetitions. Since the MLE is invariant to $m$, only a single value is shown.}
		\label{tab:res20k13:16}
	\end{table}

	\begin{table}[ht]
		\spacingset{1}
		\centering
		\resizebox{\textwidth}{!}{%
			\begin{tabular}{l|ccccc|ccccc}
				\hline
				$n = 20{,}000$& \multicolumn{5}{c|}{$\beta^*_{17}=0$} & \multicolumn{5}{c}{$\beta^*_{18}=0$} \\
				$m$ & 200 & 500 & 1000 & 5000 & 20000 & 200 & 500 & 1000 & 5000 & 20000 \\
				\hline
				MLE              & & & -0.001 &  & & &  & -0.001 &  &  \\
				Syn-novel        & -0.001 & -0.002 & -0.001 & -0.000 & -0.001 & -0.000 & -0.001 & -0.001 & -0.001 & -0.001 \\
				Syn-naive        & -0.006 & -0.006 & -0.006 & -0.003 & -0.007 & -0.000 &  0.001 & -0.004 & -0.007 & -0.006 \\
				MLE-SE           &   &   & 0.025  &   &   &   &   & 0.027  &   &   \\
				Syn-novel-emp-SE  & 0.045  & 0.033  & 0.029  & 0.026  & 0.026  & 0.042  & 0.034  & 0.031  & 0.029  & 0.028  \\
				Syn-novel-est-SE  & 0.042  & 0.034  & 0.030  & 0.027  & 0.026  & 0.042  & 0.034  & 0.030  & 0.027  & 0.026  \\
				Syn-naive-SE     & 0.294  & 0.183  & 0.124  & 0.060  & 0.035  & 0.285  & 0.179  & 0.124  & 0.059  & 0.035  \\
				Coverage-MLE & & &0.964   & & & & &0.934  & & \\
				Coverage-novel  & 0.932  & 0.938  & 0.960  & 0.952  & 0.958  & 0.938  & 0.946  & 0.946  & 0.926  & 0.932  \\
				\hline
				\multicolumn{11}{c}{} \\
				\hline
				$n = 20{,}000$& \multicolumn{5}{c|}{$\beta^*_{19}=0$} & \multicolumn{5}{c}{$\beta^*_{20}=0$} \\
				$m$ & 200 & 500 & 1000 & 5000 & 20000 & 200 & 500 & 1000 & 5000 & 20000 \\
				\hline
				MLE              &   &   & 0.001  & &   &   &   & 0.000  &   &   \\
				Syn-novel        & -0.000 & -0.001 & -0.000 & 0.000  & 0.001  & 0.000  & -0.000 & 0.000  & 0.000  & 0.000  \\
				Syn-naive        & -0.021 & -0.013 & -0.010 & -0.007 & -0.007 & 0.018  & -0.002 & -0.001 & -0.003 & -0.005 \\
				MLE-SE           &   &  & 0.026  &   &   &   &   & 0.026  &   &   \\
				Syn-novel-emp-SE  & 0.043  & 0.034  & 0.031  & 0.027  & 0.026  & 0.042  & 0.034  & 0.031  & 0.027  & 0.026  \\
				Syn-novel-est-SE  & 0.041  & 0.034  & 0.030  & 0.027  & 0.026  & 0.041  & 0.034  & 0.030  & 0.027  & 0.026  \\
				Syn-naive-SE     & 0.286  & 0.163  & 0.122  & 0.062  & 0.037  & 0.266  & 0.167  & 0.123  & 0.057  & 0.033  \\
				Coverage-MLE & & &0.960   & & & &  &0.962  & & \\
				Coverage-novel  & 0.926  & 0.936  & 0.944  & 0.956  & 0.956  & 0.948  & 0.942  & 0.946  & 0.952  & 0.960  \\
				\hline
			\end{tabular}
		}
		\caption{Simulation results for the coefficients $\bbeta^*_{17:20}$ in a Poisson regression under Setting A  in Section~\ref{sec:sims} of the main paper, with a fixed sample size of $n = 20{,}000$ and varying $m$ values, using the \texttt{synthpop} synthetic model. Results are based on 500 repetitions. Since the MLE is invariant to $m$, only a single value is shown.}
		\label{tab:res20k17:20}
	\end{table}

	\begin{figure}[h]
		\centering
		\resizebox{0.9\textwidth}{!}{ 
			\begin{minipage}{\textwidth}
				\begin{subfigure}[b]{0.48\textwidth}  
					\includegraphics[width=\textwidth]{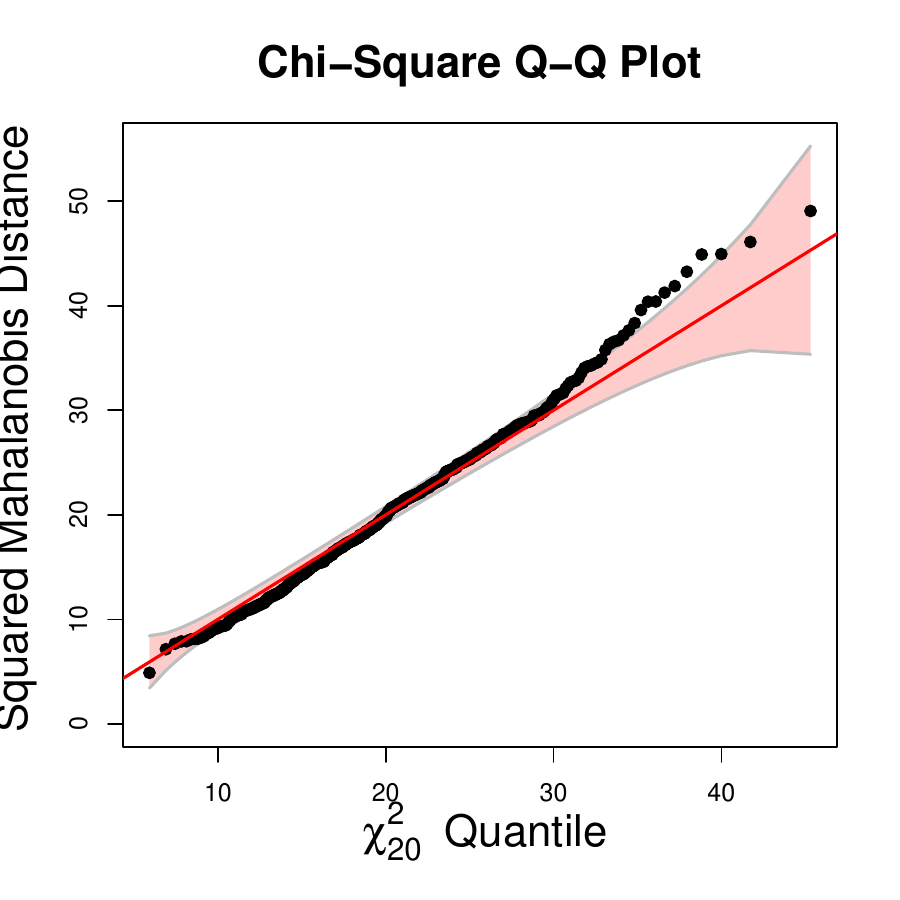}
					\caption{Poisson: Setting A, $n=5,000$}
				\end{subfigure}
				\begin{subfigure}[b]{0.48\textwidth}  
					\includegraphics[width=\textwidth]{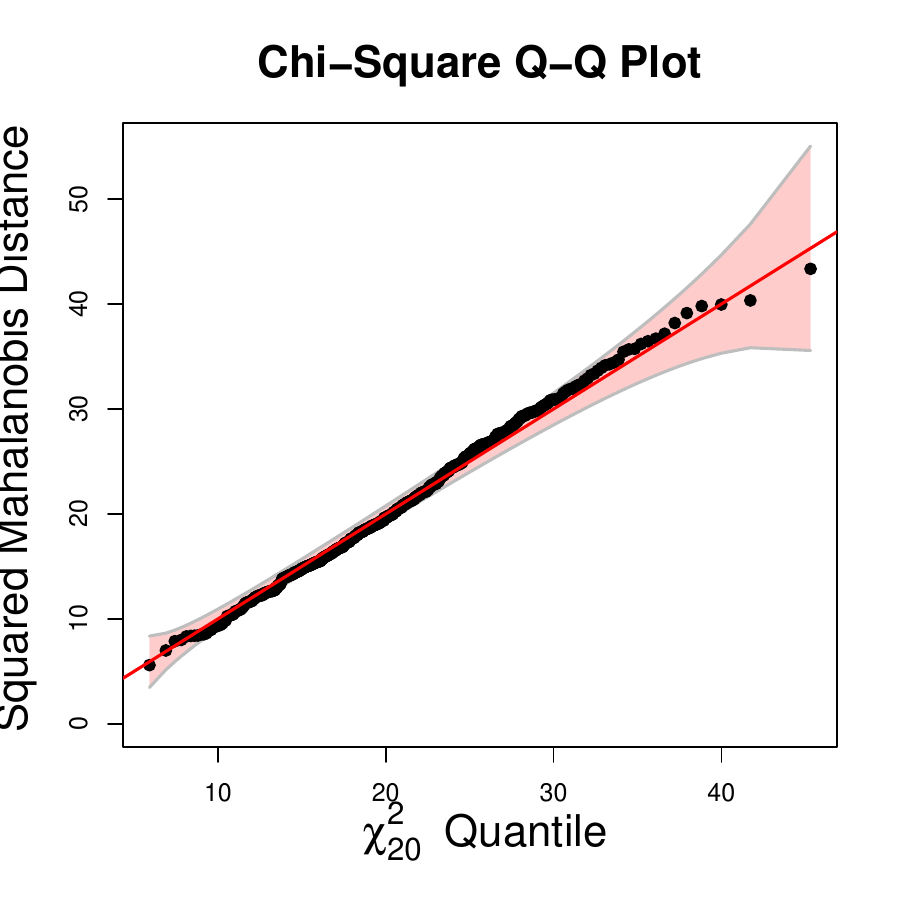}
					\caption{Poisson: Setting A, $n=20,000$}
				\end{subfigure}
				
				\begin{subfigure}[b]{0.48\textwidth}  
					\includegraphics[width=\textwidth]{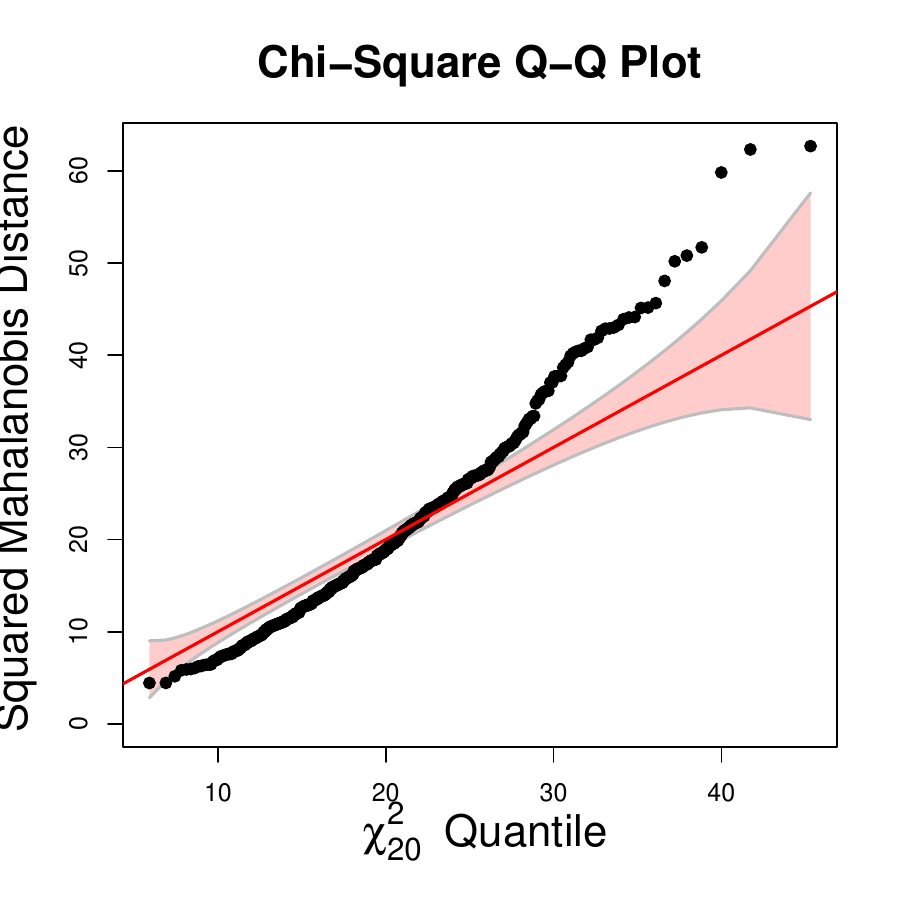}
					\caption{Logistic: Setting B, $n=5,000$}
				\end{subfigure}
				\begin{subfigure}[b]{0.48\textwidth}  
					\includegraphics[width=\textwidth]{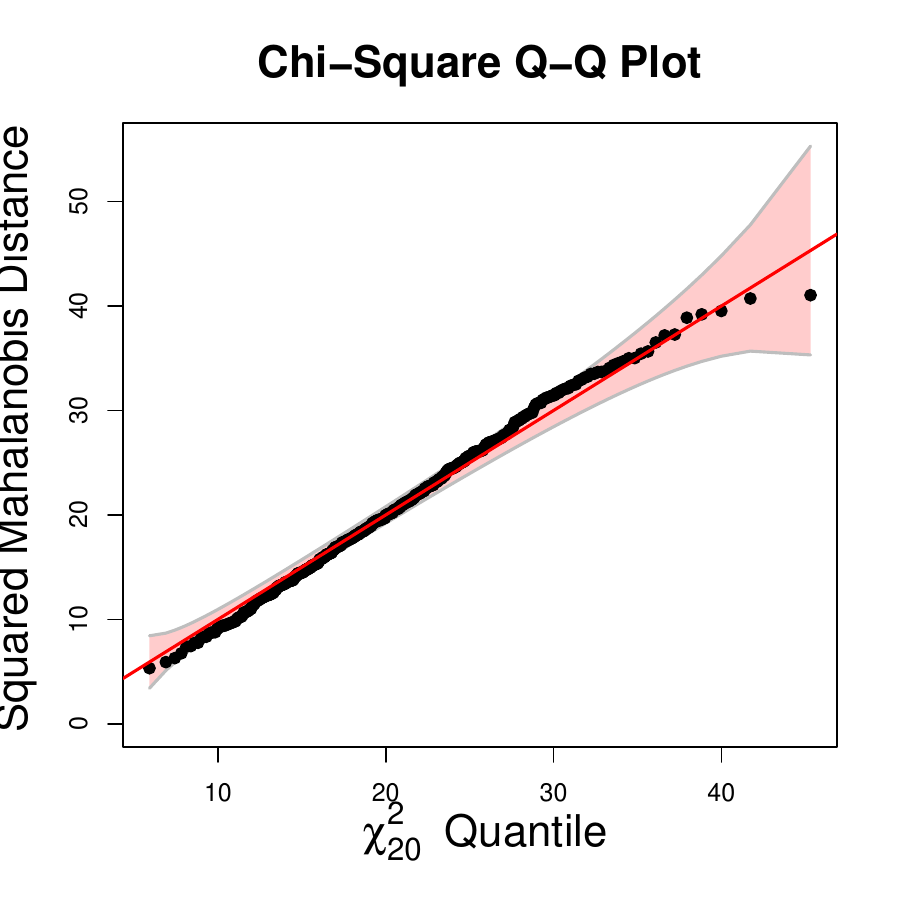}
					\caption{Logistic: Setting B, $n=20,000$}
				\end{subfigure}
				
				\begin{subfigure}[b]{0.48\textwidth}  
					\includegraphics[width=\textwidth]{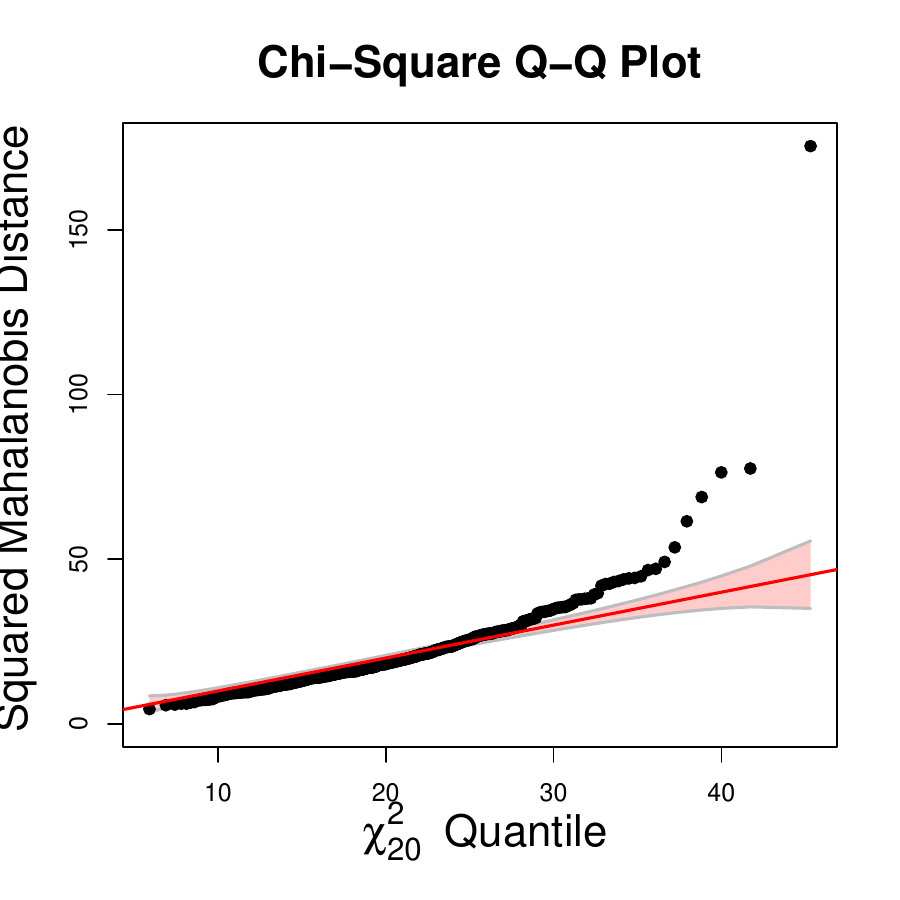}
					\caption{Poisson: Setting B, $n=5,000$}
				\end{subfigure}
				\begin{subfigure}[b]{0.48\textwidth}  
					\includegraphics[width=\textwidth]{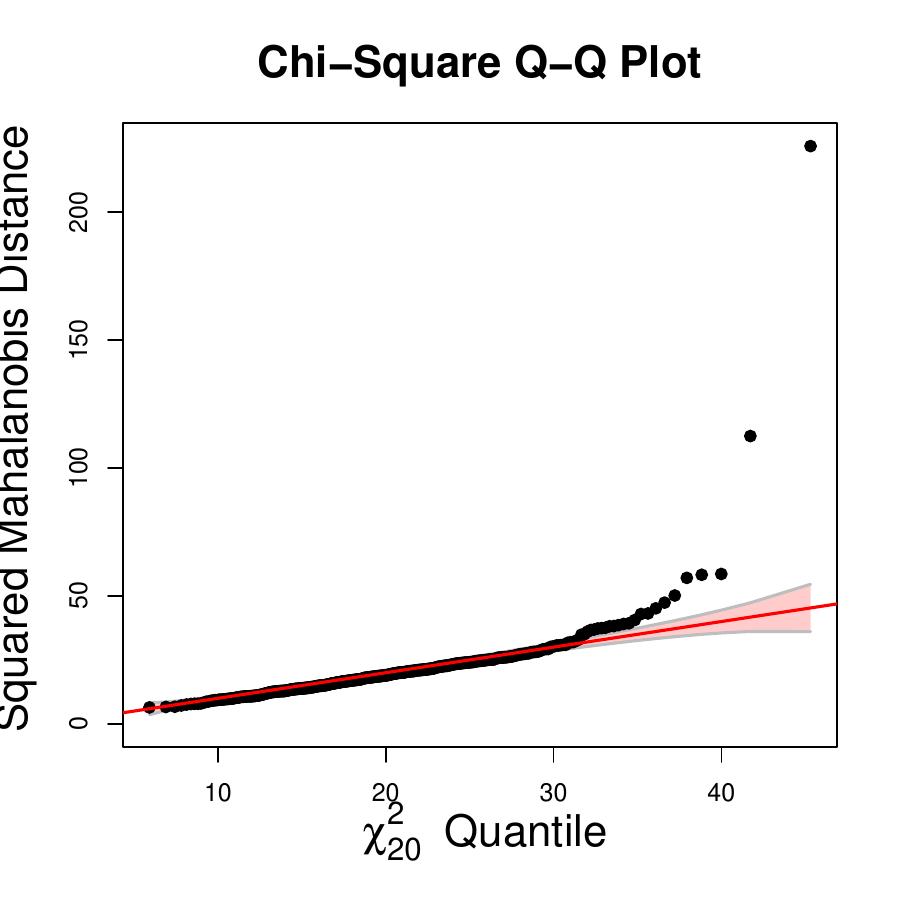}
					\caption{Poisson: Setting B, $n=20,000$}
				\end{subfigure}
			\end{minipage}
		}
		
		\caption{\footnotesize Chi-squared Q-Q plots assessing the asymptotic normality of the ``estimating equation's bias," as defined in Assumption~\ref{as:biasNorm}. 		\label{fig:CQplots}}
	\end{figure}

	\FloatBarrier
	
	\section{Additional Real Data Tables}\label{appen:real}

	\begin{table}[ht]
		\spacingset{1}
		\centering
		\begin{tabular}{lccccc}
			\hline
			$n=500,000$ & Level              & Orig: No death         & Orig: Death            & Synth: No death         & Synth: Death           \\
			\hline
			Sex           & Female & 232,872 (48.9)        & 10,399 (44.5)        & 257,016 (53.9)        & 13,186 (57.8)       \\
			& Male   & 243,771 (51.1)        & 12,958 (55.5)        & 220,162 (46.1)        & 9,636  (42.2)       \\
			Maternal age  &        & 28.53 (5.85)          & 27.85 (6.29)         & 26.54 (5.42)          & 26.29 (5.46)        \\
			BMI           &        & 26.20 (7.41)          & 27.62 (8.25)         & 26.60 (7.00)          & 27.88 (7.49)        \\
			Diabetes      & No     & 445,297 (93.4)        & 21,950 (94.0)        & 453,284 (95.0)        & 21,433 (93.9)       \\
			& Yes    & 31,346  (6.6)         & 1,407   (6.0)        & 23,894  (5.0)         & 1,389   (6.1)       \\
			Chronic HTN   & No     & 468,861 (98.4)        & 22,580 (96.7)        & 465,609 (97.6)        & 21,848 (95.7)       \\
			& Yes    & 7,782   (1.6)         & 777    (3.3)         & 11,569  (2.4)         & 974     (4.3)       \\
			Preg HTN      & No     & 450,037 (94.4)        & 21,892 (93.7)        & 458,385 (96.1)        & 21,096 (92.4)       \\
			& Yes    & 26,606   (5.6)        & 1,465   (6.3)        & 18,793   (3.9)        & 1,726   (7.6)       \\
			Eclampsia     & No     & 475,462 (99.8)        & 23,212 (99.4)        & 475,652 (99.7)        & 22,581 (98.9)       \\
			& Yes    & 1,181    (0.2)        & 145     (0.6)        & 1,526    (0.3)        & 241     (1.1)       \\
			Smoking pre   & No     & 428,683 (89.9)        & 19,577 (83.8)        & 437,311 (91.6)        & 19,639 (86.1)       \\
			& Yes    & 47,960  (10.1)        & 3,780  (16.2)        & 39,867   (8.4)        & 3,183   (13.9)      \\
			Smoking preg  & No     & 439,771 (92.3)        & 20,118 (86.1)        & 441,689 (92.6)        & 19,868 (87.1)       \\
			& Yes    & 36,872   (7.7)        & 3,239  (13.9)        & 35,489   (7.4)        & 2,954   (12.9)      \\
			Hospital birth& No     & 7,433   (1.6)         & 389    (1.7)         & 10,049   (2.1)        & 582     (2.6)       \\
			& Yes    & 469,210 (98.4)        & 22,968 (98.3)        & 467,129 (97.9)        & 22,240 (97.4)       \\
			Birth attendant& Doctor & 428,125 (89.8)        & 22,029 (94.3)        & 405,840 (85.0)        & 19,589 (85.9)       \\
			& Midwife & 44,655   (9.4)        & 855    (3.7)         & 68,540  (14.4)        & 2,750   (12.0)      \\
			& Other   & 3,863    (0.8)        & 473    (2.0)         & 2,798    (0.6)        & 483     (2.1)       \\
			Labor induction& No    & 362,992 (76.2)        & 19,975 (85.5)        & 385,486 (80.8)        & 19,675 (86.2)       \\
			& Yes    & 113,651 (23.8)        & 3,382  (14.5)        & 91,692   (19.2)       & 3,147   (13.8)      \\
			Forceps       & No     & 474,013 (99.4)        & 23,272 (99.6)        & 473,053 (99.1)        & 22,518 (98.7)       \\
			& Yes    & 2,630    (0.6)        & 85     (0.4)         & 4,125    (0.9)        & 304     (1.3)       \\
			\hline
		\end{tabular}
		\caption{Baseline characteristics by data source and one-year mortality (Part 1). Continuous variables are reported as mean (SD); categorical variables are counts (\%).}
		\label{tab:realTable1Part1}
	\end{table}
	
	\begin{table}[ht]
		\spacingset{1}
		\centering
		\begin{tabular}{lccccc}
			\hline
			$n=500,000$ & Level              & Orig: No death         & Orig: Death            & Synth: No death         & Synth: Death           \\
			\hline
			Vacuum         & No     & 464,252 (97.4)        & 23,138 (99.1)        & 462,965 (97.0)        & 22,354 (97.9)       \\
			& Yes    & 12,391  (2.6)         & 219    (0.9)         & 14,213  (3.0)         & 468     (2.1)       \\
			Breech         & No     & 451,732 (94.8)        & 17,948 (76.8)        & 455,045 (95.4)        & 18,788 (82.3)       \\
			& Yes    & 24,911  (5.2)         & 5,409  (23.2)        & 22,133  (4.6)         & 4,034  (17.7)       \\
			Birth order    &        & 2.51 (1.58)           & 2.79 (1.83)          & 2.42 (1.48)           & 2.56 (1.59)         \\
			Plurality      & Singleton & 460,311 (96.6)     & 20,003 (85.6)        & 460,924 (96.6)        & 19,429 (85.1)       \\
			& Twins     & 15,867  (3.3)       & 3,137  (13.4)        & 15,535  (3.3)         & 3,224  (14.1)       \\
			& 3+        & 465     (0.1)       & 217    (1.0)         & 719     (0.1)         & 169    (0.8)        \\
			Maternal race  & White & 356,270 (74.7)        & 14,556 (62.3)        & 362,617 (76.0)        & 16,307 (71.5)       \\
			& Black & 72,861  (15.3)        & 6,760  (28.9)        & 62,116  (13.0)        & 4,645  (20.4)       \\
			& Other & 47,512  (10.0)        & 2,041   (8.8)        & 52,445  (11.0)        & 1,870   (8.1)       \\
			Paternal race  & White & 346,451 (72.7)        & 13,913 (59.6)        & 358,066 (75.0)        & 15,713 (68.9)       \\
			& Black & 85,435  (17.9)        & 7,554  (32.3)        & 66,121  (13.9)        & 5,103  (22.4)       \\
			& Other & 44,757   (9.4)        & 1,890   (8.1)        & 52,991  (11.1)        & 2,006   (8.7)       \\
			Maternal edu   & $<$HS & 69,117  (14.5)        & 4,604  (19.7)        & 85,523  (17.9)        & 5,107  (22.4)       \\
			& HS    & 118,380 (24.8)        & 7,420  (31.8)        & 115,343 (24.2)        & 7,099  (31.1)       \\
			& Higher& 289,146 (60.7)        & 11,333 (48.5)        & 276,312 (57.9)        & 10,616 (46.6)       \\
			Paternal edu   & $<$HS & 74,665  (15.7)        & 4,704  (20.1)        & 81,462  (17.1)        & 4,726  (20.7)       \\
			& HS    & 144,247 (30.3)        & 8,744  (37.4)        & 142,994 (30.0)        & 8,939  (39.2)       \\
			& Higher& 257,731 (54.0)        & 9,909  (42.5)        & 252,722 (52.9)        & 9,157  (40.1)       \\
			Married        & No    & 191,419 (40.2)        & 12,769 (54.7)        & 145,815 (30.6)        & 9,139  (40.0)       \\
			& Yes   & 285,224 (59.8)        & 10,588 (45.3)        & 331,363 (69.4)        & 13,683 (60.0)       \\
			WIC            & No    & 279,454 (58.6)        & 13,201 (56.5)        & 256,751 (53.8)        & 13,485 (59.1)       \\
			& Yes   & 197,189 (41.4)        & 10,156 (43.5)        & 220,427 (46.2)        & 9,337  (40.9)       \\
			\hline
		\end{tabular}
		\caption{Baseline characteristics by data source and one-year mortality (Part 2). Continuous variables are reported as mean (SD); categorical variables are counts (\%).}
		\label{tab:realTable1Part2}
	\end{table}

	\begin{table}[ht]
		\spacingset{1}
		\centering
		\begin{tabular}{lcccc}
			\hline
			$n=500{,}000$ & Original & Syn-novel (500k) &  Syn-novel (50k) &  Syn-naive \\
			\hline
			Intercept & -2.641 (0.080) & -2.641 (0.077) & -2.657 (0.085) & -2.689 (0.066) \\
			Hospital birth & -0.326 (0.066) & -0.304 (0.061) & -0.279 (0.070) & -0.220 (0.045) \\
			Birth attendant: Midwife & -0.766 (0.037) & -0.666 (0.031) & -0.661 (0.034) & -0.161 (0.021) \\
			Birth attendant: Other & 0.726 (0.060) & 0.763 (0.055) & 0.807 (0.065) & 1.074 (0.053) \\
			Maternal age & -0.016 (0.001) & -0.019 (0.001) & -0.019 (0.002) & -0.005 (0.002) \\
			Married & -0.321 (0.017) & -0.292 (0.016) & -0.294 (0.018) & -0.201 (0.017) \\
			Birth order & 0.054 (0.005) & 0.062 (0.005) & 0.061 (0.005) & -0.001 (0.005) \\
			WIC & -0.402 (0.016) & -0.440 (0.017) & -0.453 (0.019) & -0.678 (0.016) \\
			BMI & 0.016 (0.001) & 0.017 (0.001) & 0.017 (0.001) & 0.015 (0.001) \\
			Sex: Male & 0.190 (0.014) & 0.203 (0.014) & 0.208 (0.015) & -0.119 (0.014) \\
			Diabetes & -0.179 (0.030) & -0.216 (0.031) & -0.230 (0.035) & 0.145 (0.030) \\
			Chronic HTN & 0.429 (0.042) & 0.478 (0.044) & 0.486 (0.048) & 0.313 (0.036) \\
			Preg HTN & 0.023 (0.029) & -0.046 (0.030) & -0.031 (0.035) & 0.306 (0.028) \\
			Eclampsia & 0.666 (0.094) & 0.650 (0.100) & 0.607 (0.112) & 0.844 (0.074) \\
			Forceps & -0.190 (0.112) & -0.321 (0.104) & -0.309 (0.116) & 0.290 (0.062) \\
			Vacuum & -0.756 (0.069) & -0.825 (0.068) & -0.808 (0.072) & -0.258 (0.049) \\
			Labor induction & -0.430 (0.020) & -0.530 (0.021) & -0.532 (0.022) & -0.406 (0.020) \\
			Breech & 1.406 (0.018) & 1.386 (0.018) & 1.408 (0.021) & 1.203 (0.020) \\
			Smoking pre & 0.032 (0.042) & 0.055 (0.039) & 0.043 (0.042) & 0.279 (0.032) \\
			Smoking preg & 0.477 (0.045) & 0.475 (0.042) & 0.486 (0.046) & 0.309 (0.033) \\
			Maternal edu: HS & -0.047 (0.022) & -0.051 (0.023) & -0.043 (0.024) & 0.006 (0.019) \\
			Maternal edu: Higher & -0.284 (0.024) & -0.286 (0.025) & -0.285 (0.027) & -0.300 (0.020) \\
			Paternal edu: HS & -0.031 (0.021) & -0.015 (0.022) & -0.017 (0.023) & 0.161 (0.019) \\
			Paternal edu: Higher & -0.220 (0.024) & -0.207 (0.025) & -0.210 (0.026) & -0.317 (0.021) \\
			Maternal race: Black & 0.429 (0.031) & 0.460 (0.030) & 0.464 (0.032) & 0.170 (0.028) \\
			Maternal race: Other & 0.113 (0.033) & 0.149 (0.033) & 0.142 (0.035) & -0.114 (0.032) \\
			Paternal race: Black & 0.275 (0.030) & 0.238 (0.029) & 0.237 (0.031) & 0.346 (0.027) \\
			Paternal race: Other & 0.075 (0.034) & 0.065 (0.034) & 0.076 (0.036) & 0.015 (0.031) \\
			Plurality: Twins & 1.043 (0.023) & 1.065 (0.024) & 1.044 (0.028) & 1.164 (0.023) \\
			Plurality: 3+ & 2.016 (0.090) & 2.160 (0.100) & 2.049 (0.136) & 1.154 (0.093) \\
			\hline
		\end{tabular}
		\caption{Logistic regression results for one-year infant mortality in the 2015 U.S. cohort. Estimates are shown for the MLE on the original data, Syn-novel applied to synthetic data with $m = n = 500{,}000$ and $m = 50{,}000$, and the Syn-naive fitted directly to synthetic data. Estimated SEs are reported in parentheses. \label{tab:realAnalysis}}
	\end{table}

\end{document}